\documentclass[sigconf,nonacm,prologue,table]{acmart}
\makeatletter
\def\@ACM@checkaffil{%
    \if@ACM@instpresent\else
    \ClassWarningNoLine{\@classname}{No institution present for an affiliation}%
    \fi
    \if@ACM@citypresent\else
    \ClassWarningNoLine{\@classname}{No city present for an affiliation}%
    \fi
    \if@ACM@countrypresent\else
        \ClassWarningNoLine{\@classname}{No country present for an affiliation}%
    \fi
}
\makeatother

\usepackage[british]{babel}
\usepackage[useregional]{datetime2}
\DTMlangsetup[en-GB]{showdayofmonth=false}
\usepackage{natbib}
\usepackage{cancel}
\usepackage{amsmath}
\usepackage{amsfonts}
\usepackage{amsthm}
\usepackage{mathtools}
\usepackage{wrapfig}
\usepackage{graphicx}
\usepackage{tikz}
\hypersetup{colorlinks=true, allcolors=blue}
\usepackage[page]{appendix} %
\usepackage{algorithm}
\usepackage{algorithmic}

\usepackage{subcaption}
\usepackage{multicol}
\usepackage{placeins}

\renewcommand{\vec}[1]{\boldsymbol{\mathbf{#1}}}
\renewcommand{\v}{\vec}

\usepackage{thmtools}
\usepackage{thm-restate}
\declaretheorem[name=Proposition]{proposition}
\declaretheorem[name=Corollary]{corollary}
\declaretheorem[name=Theorem]{theorem}
\declaretheorem[name=Lemma]{lemma}
\declaretheorem[name=Remark,style=remark]{remark}

\begin{document}
\title{Riemannian Geometry of Optimal Rebalancing\\in Dynamic Weight Automated Market Makers}

\date{March 2026}
\author{Matthew Willetts}
\affiliation{}

\begin{abstract}
We show that when a dynamic-weight AMM rebalances by creating arbitrage opportunities, the per-step log loss is the KL divergence between successive weight vectors.
The Fisher–Rao metric is therefore the natural Riemannian metric on the weight simplex.
The loss-minimising interpolation under the leading-order expansion of this KL cost is SLERP (Spherical Linear Interpolation) in the Hellinger coordinates $\eta_i = \sqrt{w_i}$: a geodesic on the positive orthant of the unit sphere, traversed at constant speed.
The SLERP midpoint equals the (AM+GM)/normalise heuristic of prior work~\cite{willetts2024optimalrebalancingdynamicamms}, so the heuristic lies on the geodesic.
This identity holds for any number of tokens and any magnitude of weight change; using this link, all dyadic points on the geodesic can be reached by recursive AM-GM bisection without trigonometric functions.
SLERP's relative sub-optimality on the full KL cost is proportional to the squared magnitude of the overall weight change and to $1/f^2$, where $f$ is the number of interpolation steps.
Under driftless GBM prices, the fractional value loss from each weight update is price-independent, and the cross term between weight and price changes telescopes, so the constant-price geometry carries over.
LVR exposure introduces a finite optimal step count $f^*$, which lies in the perturbative regime where SLERP remains near-optimal.
\end{abstract}

\maketitle

\section{Introduction}

Dynamic AMM pools rebalance their holdings by introducing arbitrage opportunities that disappear once reserves match the desired target.
In Temporal Function Market Making (TFMM)~\cite{tfmm_litepaper}, this rebalancing is driven by changing the weights of a geometric mean market maker (G3M) pool from block to block.

The pool's asset management strategy produces new target weights at a given cadence.\footnote{This can be any long-only zero-leverage strategy, for example: track a market-cap weighted index, or a trend-following strategy.}
Spreading a weight update over multiple intermediate steps reduces the total arbitrage cost~\cite{tfmm_litepaper}, because the cost is convex in the step size, so the pool interpolates towards the target.

Prior work~\cite{willetts2024optimalrebalancingdynamicamms} derived the optimal weight midpoint in the limit of small steps, and from this obtained approximately-optimal trajectories via arithmetic and geometric mean interpolation: the (AM+GM)/\allowbreak normalise heuristic.
This heuristic captures ${\sim}95\%$ of the improvement available from full numerical optimisation, but the geometric reason for its effectiveness was left unexplained.

This paper provides that explanation.
Against a backdrop of constant prices, we show that:
\begin{itemize}
  \item The per-step arbitrage loss is a KL divergence between weight vectors (Theorem~\ref{thm:kl_divergence}), making the Fisher--Rao metric the natural Riemannian structure on the weight
simplex.
  \item The loss-minimising interpolation under the leading-order expansion is SLERP (Spherical Linear Interpolation)~\cite{shoemake1985} in Hellinger coordinates (Corollary~\ref{cor:slerp_optimal}).
  \item The SLERP midpoint equals the (AM+GM)/normalise midpoint exactly, for any number of tokens (Theorem~\ref{thm:slerp_amgm}); this is why the heuristic works.
  \item Exact SLERP trajectories at power-of-two step counts can be computed by trig-free recursive bisection (Corollary~\ref{cor:bisection}).
  \item SLERP's sub-optimality on the exact KL cost vanishes as $O(1/f^3)$ in the step count (Theorem~\ref{thm:suboptimality}).
\end{itemize}
We then use these results to attack the case where prices are stochastic, and show that:
\begin{itemize}
  \item Under driftless GBM prices (martingale assumption), the retention ratio is price-independent.  The cross term between weight and price changes telescopes, so the path optimisation is unchanged (\S\ref{sec:stochastic}).
  \item LVR exposure introduces a finite optimal step count $f^*$, which lies in the perturbative regime where SLERP remains near-optimal (\S\ref{sec:optimal_f}).
\end{itemize}
We validate these results numerically in \S\nolinebreak\ref{sec:experiments}.

\section{Background}
\label{sec:background}

\subsection{Dynamic Weight Automated Market Makers}
\paragraph{Arbitrage in Geometric Mean Market Makers when weights change}

A \emph{block} is the atomic time unit of the blockchain: all state (reserves, weights, prices) is fixed within a block and may change between consecutive blocks.\footnote{On Ethereum, blocks are produced every 12 seconds; on L2 chains such as Base, every 2 seconds.}

Consider an $N$-token G3M pool with weights $\v w = (w_1, \ldots, w_N)$ on the simplex ($\sum_i w_i = 1$, $w_i > 0$) and reserves $\v R\in\mathbb{R}^N_+$.
The pool invariant is $\prod_i R_i^{w_i} = k$~\cite{balancer,evansG3Ms}.
When the pool is in equilibrium, its quoted prices match market prices.

When weights change directly from $\v w^{\mathrm{start}}$ to $\v w^{\mathrm{end}}$ at constant market prices, an arbitrage opportunity is created.
After arbitrage, the reserves satisfy~\cite{willetts2024optimalrebalancingdynamicamms}
\begin{equation}
    R_i^{\mathrm{end}} = R_i^{\mathrm{start}} \frac{w_i^{\mathrm{end}}}{w_i^{\mathrm{start}}} \prod_{j=1}^N \left(\frac{w_j^{\mathrm{start}}}{w_j^{\mathrm{end}}}\right)^{w_j^{\mathrm{end}}}.
    \label{eq:reserve_change}
\end{equation}

\paragraph{Approximately Optimal Weight Interpolations}

Within each block-to-block step, market prices are modelled as constant.
If, instead of updating weights directly from $\v w^{\mathrm{start}}$ to $\v w^{\mathrm{end}}$ in one step (i.e. across two blocks), the pool goes via an intermediate value $\Tilde{\v w}$ (i.e. across three blocks), so that it is arbitraged twice, $\v w^{\mathrm{start}} \rightarrow \Tilde{\v w} \rightarrow \v w^{\mathrm{end}}$, the total cost paid to arbitrageurs is reduced~\cite{tfmm_litepaper}.
$\Tilde{\v w}$ can take a broad range of possible values and still improve the pool's performance.
Dividing up a weight change in this way is beneficial for the same reason that splitting a large order reduces market impact~\cite{almgren2001}: each sub-step faces a smaller price displacement (Figure~\ref{fig:interpolation}).

\begin{figure*}[h]
\centering
    \includegraphics[width=0.8\textwidth]{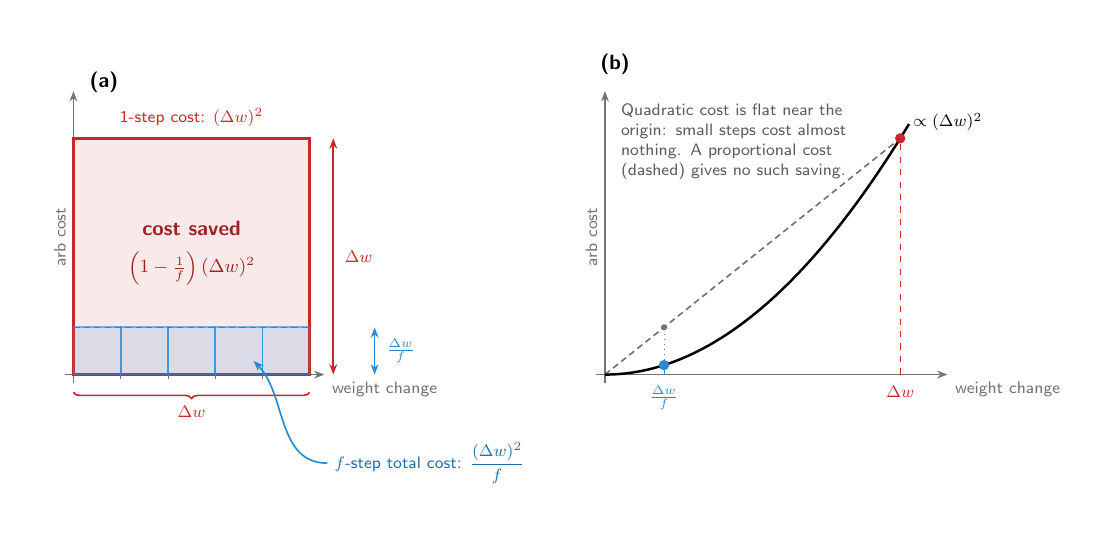}
    \vspace{-1cm}
    \caption{Why weight interpolation reduces rebalancing cost.
    (a)~For a G3M pool, rebalancing cost is quadratic in the weight change $\Delta w$, so splitting into $f$ equal sub-steps of $\Delta w/f$ reduces total cost by a factor of $f$.
    (b)~The quadratic cost curve is flat near the origin: small weight changes cost almost nothing, while a single large change incurs a disproportionately high arb loss.}
    \label{fig:interpolation}
\end{figure*}

Assuming the weight change is small, the optimal intermediate weights for the two-step process are
\begin{equation}
\Tilde{w}_i^* = \frac{w_i(t_f)}{W_0\left(\frac{e w_i(t_f)}{w_i(t_0)}\right)},
\label{eq:best_w_tilde}
\end{equation}
as obtained in~\cite{willetts2024optimalrebalancingdynamicamms}, where $W_0(\cdot)$ is the principal branch of the Lambert W function and $e$ is Euler's number.

Bounding this result from above and below by the arithmetic and geometric means respectively yields an approximately-optimal $f$-step weight trajectory $\{\breve{\v w}(t_k)\}_{k=1,...,f-1}$ from $\v w(t_0)=\v w^{\mathrm{start}}$ to $\v w(t_f)=\v w^{\mathrm{end}}$ as the average of the linear interpolation between these values and the geometric interpolation between these values~\cite{willetts2024optimalrebalancingdynamicamms}:
\begin{align}
    w_i^{\mathrm{AM}}(t_k) & = (1-\frac{k}{f})w_i(t_0) + \frac{k}{f}w_i(t_f) \label{eq:approx_optimal_traj_am},\\
    w_i^{\mathrm{GM}}(t_k) & = {(w_i(t_0))^{(1-\frac{k}{f})}(w_i(t_f))^{\frac{k}{f}}},\label{eq:approx_optimal_traj_gm}\\
    \breve{w}_i(t_k) &= \frac{w_i^{\mathrm{AM}}(t_k)+w_i^{\mathrm{GM}}(t_k)}{\sum_{j=1}^N\left({w_j^{\mathrm{AM}}(t_k)+w_j^{\mathrm{GM}}(t_k)}\right)} \label{eq:approx_optimal_traj}.
\end{align}

\subsection{Information Geometry and the Fisher--Rao metric}

The weight simplex carries a natural geometric structure that, as we show in \S\ref{sec:slerp}, governs the arbitrage cost.

The probability simplex $\Delta^{N-1} = \{\mathbf{p} \in \mathbb{R}^N_+ : \sum_i p_i = 1\}$ carries a natural Riemannian structure induced by the Kullback--Leibler (KL) divergence~\cite{kullback1951}.
For two distributions $p, q \in \Delta^{N-1}$, the KL divergence
\[
D_{\mathrm{KL}}(p \,\|\, q) = \sum_i p_i \log(p_i / q_i)
\]
is asymmetric and does not satisfy the triangle inequality, so it
is not itself a metric.
Its local second-order expansion, however, around $q = p$ yields a positive-definite quadratic form,
\begin{equation}
  D_{\mathrm{KL}}(p \,\|\, p + \mathrm{d}p) \;=\; \tfrac{1}{2} \sum_{i=1}^{N} \frac{(\mathrm{d}p_i)^2}{p_i} \;+\; O(\|\mathrm{d}p\|^3).
  \label{eq:kl_expansion}
\end{equation}
This expansion defines the \emph{Fisher--Rao metric} $g$~\cite{fisher1925,rao1945}. On the ambient space $\mathbb{R}^N_{+}$, the metric is diagonal, $g_{ij}(p)=\delta_{ij}/p_i$, with the corresponding squared infinitesimal line element,
\begin{equation}
    \mathrm{d}s^2=\sum_{i=1}^N \frac{(\mathrm{d}p_i)^2}{p_i}.
    \label{eq:fisher}
\end{equation}
The geometry of the simplex $(\Delta^{N-1}, g)$ is that of the positive orthant of an ($N-1$)-sphere.
By the theorem of \citet{cencov}, this is the unique Riemannian metric (up to a scale factor) invariant under sufficient statistics, establishing its role as the natural geometric structure.

\section{Weights change with constant prices}
\subsection{Arbitrage loss is a KL divergence}

\begin{restatable}{proposition}{retention}[Retention Ratio]
\label{prop:retention}
For an $N$-asset G3M pool with weights changing from $\v w^{\mathrm{start}}$ one block to $\v w^{\mathrm{end}}$ the next block, with zero fees, perfect arbitrage, and constant prices, the ratio of final value to initial value of the pool is the retention ratio
\begin{equation}
r:=\prod_{j=1}^N \left(\frac{w_j^{\mathrm{start}}}{w_j^{\mathrm{end}}}\right)^{w_j^{\mathrm{end}}}.
\label{eq:retention}
\end{equation}
\end{restatable}
\begin{proof}
    See Appendix~\ref{app:retention}.
\end{proof}

\begin{theorem}[Arbitrage loss as KL divergence]
\label{thm:kl_divergence}
The log-loss from a weight update is the KL divergence between new and old weights:
\begin{equation}
    -\log r = \sum_{i=1}^N w_i^{\mathrm{end}} \log\!\frac{w_i^{\mathrm{end}}}{w_i^{\mathrm{start}}} = D_{\mathrm{KL}}\!\left(\v w^{\mathrm{end}} \,\|\, \v w^{\mathrm{start}}\right).
    \label{eq:kl_loss}
\end{equation}
In particular, $r \leq 1$ with equality if and only if $\v w^{\mathrm{end}} = \v w^{\mathrm{start}}$, since $D_{\mathrm{KL}} \geq 0$.
\end{theorem}
\begin{proof}
Take $-\log$ of Eq~\eqref{eq:retention} and recognise the definition of $D_{\mathrm{KL}}$.
Note the asymmetry: $D_{\mathrm{KL}}(\v w^{\mathrm{end}} \| \v w^{\mathrm{start}})$ uses the new weights as the reference distribution, because the exponents $w_j^{\mathrm{end}}$ in the retention ratio determine the post-trade value allocation.
Non-negativity follows from Gibbs' inequality~\cite{cover_thomas2006}.
\end{proof}

Weight vectors are positive and sum to one, so they live on the probability simplex; Theorem~\ref{thm:kl_divergence} thus holds for weight changes of any magnitude.
The Fisher--Rao metric is the Hessian of the KL divergence at the diagonal~\cite{amari_info_geo}; its quadratic form gives the loss kernel (Eq~\eqref{eq:tfmm_loss_kernel} below).

\begin{restatable}{corollary}{quadraticloss}[Quadratic loss kernel]
\label{cor:quadratic_loss}
For small weight changes $\Delta w_i = w_i^{\mathrm{end}} - w_i^{\mathrm{start}}$, the leading-order loss is
\begin{equation}
  -\log r \;\approx\; \sum_{i=1}^N \frac{(\Delta w_i)^2}{2\, w_i^{\mathrm{start}}}.
  \label{eq:tfmm_loss_kernel}
\end{equation}
\end{restatable}
\begin{proof}
  Taylor-expand Theorem~\ref{thm:kl_divergence} to second order; the linear term vanishes by the simplex constraint (Appendix~\ref{app:taylor}).
\end{proof}
This is half the Fisher--Rao quadratic form, and it defines the infinitesimal line element on the weight simplex.

\begin{proposition}[Lambert W midpoint as KL stationary point]
\label{prop:lambert_kl}
For $f=2$ (one intermediate weight vector $\v m$ between start and end), the unconstrained component-wise stationarity condition
\[
\frac{\partial}{\partial m_i} \left[D_{\mathrm{KL}}(\v m \,\|\, \v w^{\mathrm{start}}) + D_{\mathrm{KL}}(\v w^{\mathrm{end}} \,\|\, \v m)\right] = 0
\]
yields exactly the Lambert~W formula $\tilde{w}_i^* = w_i^{\mathrm{end}} / W_0(e\, w_i^{\mathrm{end}} / w_i^{\mathrm{start}})$ of Eq~\eqref{eq:best_w_tilde}.
\end{proposition}
\begin{proof}
Differentiating: $\partial/\partial m_i [m_i \log(m_i/w_i^s) + w_i^e \log(w_i^e/m_i)] = \log(m_i/w_i^s) + 1 - w_i^e/m_i = 0$, which rearranges to $m_i \exp(m_i/w_i^s - 1) = w_i^e$, giving $m_i = w_i^e / W_0(e\,w_i^e/w_i^s)$.
\end{proof}
The simplex constraint $\sum_i m_i = 1$ is then imposed by renormalisation, which is why the formula is exact only for small weight changes: the full constrained problem couples the components through a Lagrange multiplier that has no closed-form solution~\cite{willetts2024optimalrebalancingdynamicamms}.
The Lambert~W midpoint optimises the unconstrained component-wise KL cost, SLERP optimises the leading-order expansion of the constrained cost, and the two agree through second order (Table~\ref{tab:agreement}).

\paragraph{The Hellinger embedding and geodesics}
\label{ssec:hellinger}

The coordinate transformation $\eta_i = \sqrt{w_i}$ maps $\Delta^{N-1}$ isometrically (up to a factor) onto the positive orthant of the unit sphere $S^{N-1}_+$, under which the Fisher--Rao metric reduces to a multiple of the round metric.\footnote{This map $w_i \mapsto \sqrt{w_i}$ is the \emph{Hellinger embedding}~\cite{hellinger1909,amari_info_geo}.}
Geodesics on the sphere are great circles, computable in closed form via spherical linear interpolation (SLERP)~\cite{shoemake1985}.
We exploit this correspondence throughout: interpolation is performed on $S^{N-1}_+$ and mapped back to weights by squaring.

The metric, Eq~\eqref{eq:fisher}, simplifies under the coordinate change
\begin{equation}
    \eta_i = \sqrt{w_i}, \qquad\text{so that}\qquad w_i = \eta_i^2,\quad \mathrm{d}w_i = 2\eta_i\, \mathrm{d}\eta_i,
    \label{eq:hellinger_map}
\end{equation}
\begin{equation}
    \Rightarrow \mathrm{d}s^2 = \sum_{i=1}^N \frac{4\eta_i^2\, (\mathrm{d}\eta_i)^2}{\eta_i^2} = 4\sum_{i=1}^N (\mathrm{d}\eta_i)^2.
    \label{eq:sphere_metric}
\end{equation}
The per-step loss, Eq~\eqref{eq:tfmm_loss_kernel}, is therefore $\tfrac{1}{2}\,\mathrm{d}s^2 = 2\sum_i(\mathrm{d}\eta_i)^2$: half the Fisher--Rao quadratic form evaluated on the unit sphere.
The constraint $\sum_i w_i = 1$ becomes $\sum_i \eta_i^2 = 1$: the point $\v\eta = (\eta_1,\ldots,\eta_N)$ lies on the positive orthant of the unit sphere $S^{N-1}$.
The metric is $4\times$ the round metric inherited from the ambient Euclidean space $\mathbb{R}^N$.
Figure~\ref{fig:hellinger} illustrates this correspondence.

\begin{figure*}[h]
\centering
    \includegraphics[width=0.8\textwidth]{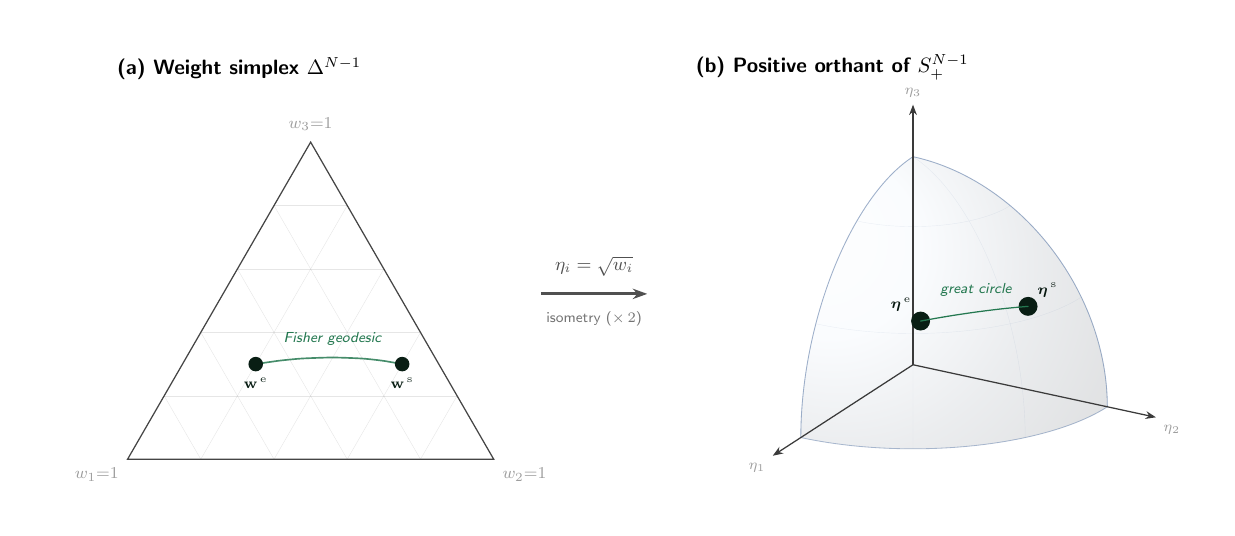}
    \caption{The Hellinger embedding maps the weight simplex $\Delta^{N-1}$ to the positive orthant of the unit sphere $S^{N-1}_+$ via $\eta_i = \sqrt{w_i}$.
    Under this isometry, the Fisher--Rao metric becomes a scaling of the standard round metric, and geodesics (shortest paths under the arbitrage cost) become great circles, computable in closed form via SLERP.}
    \label{fig:hellinger}
\end{figure*}

\subsection{SLERP on the sphere under Hellinger embedding}
\label{sec:slerp}
We derive the weight interpolation that minimises the quadratic loss (Eq~\eqref{eq:tfmm_loss_kernel}).

\begin{corollary}[SLERP optimality]
\label{cor:slerp_optimal}
Among all $f$-step paths on the simplex from $\v w^{\mathrm{start}}$ to~$\v w^{\mathrm{end}}$, the path minimising the sum of squared Fisher--Rao arc-lengths $\sum_{k=1}^{f} (\mathrm{d}s_k)^2$
is the constant-speed geodesic: SLERP in Hellinger coordinates, where the geodesic angle is $\Omega = \arccos(\mathrm{BC})$ and $\mathrm{BC} = \sum_i\sqrt{w_i^s w_i^e}$ is the Bhattacharyya coefficient~\cite{bhattacharyya1943}.
\end{corollary}
\begin{proof}
Under the Hellinger embedding (\S\ref{ssec:hellinger}), the Fisher--Rao metric is $4\times$ the round metric on $S^{N-1}_+$, so $\mathrm{d}s_k = 2\alpha_k$ where $\alpha_k$ is the geodesic angle subtended by step~$k$.
By convexity of $x \mapsto x^2$, $\sum_{k=1}^f \alpha_k^2 \geq (\sum_k \alpha_k)^2/f$ with equality iff all $\alpha_k$ are equal (Appendix~\ref{app:constant_speed_proof}).
The minimum over paths is attained by the geodesic (shortest total angle $\sum_k \alpha_k = \Omega$), traversed at constant angular speed ($\alpha_k = \Omega/f$): that is, SLERP~\cite{shoemake1985}.
\end{proof}
Along the SLERP path, each step subtends equal arc $\Omega/f$, giving per-step loss $\mathrm{d}s_k^2/2 = 2\Omega^2/f^2$ and total loss $2\Omega^2/f$.
The per-step loss equals the quadratic loss kernel of Corollary~\ref{cor:quadratic_loss} at leading order.\footnote{The base-point quadratic form $\sum_i (\Delta w_i)^2/(2w_i)$ differs from the exact $\mathrm{d}s^2/2$ by $O(\|\Delta\v\eta\|^3)$ per step, giving a total discrepancy of $O(\Omega^3/f^2)$.}

\subsection{The SLERP formula}
\label{ssec:slerp}

Given start weights $\v w^{\mathrm{start}}$ and end weights $\v w^{\mathrm{end}}$, the loss-minimising multi-step interpolation proceeds as in Algorithm~\ref{alg:slerp}.
\begin{algorithm}[h]
\caption{SLERP weight interpolation on the positive sphere}
\label{alg:slerp}
\begin{algorithmic}[1]
\REQUIRE Start weights $\mathbf{w}^{\mathrm{start}} \in \Delta^{N-1}$, end weights $\mathbf{w}^{\mathrm{end}} \in \Delta^{N-1}$, number of steps $f$
\ENSURE Interpolated weight sequence $\mathbf{w}(0), \mathbf{w}(1), \ldots, \mathbf{w}(f)$
\STATE \textbf{Map to the sphere:}
\STATE $\boldsymbol{\eta}^{\mathrm{start}} \leftarrow \left(\sqrt{w_1^{\mathrm{start}}},\,\ldots,\,\sqrt{w_N^{\mathrm{start}}}\right)$
\STATE $\boldsymbol{\eta}^{\mathrm{end}} \leftarrow \left(\sqrt{w_1^{\mathrm{end}}},\,\ldots,\,\sqrt{w_N^{\mathrm{end}}}\right)$
\STATE \textbf{Compute geodesic angle:}
\STATE $\Omega \leftarrow \arccos\!\left(\sum_{i=1}^N \sqrt{w_i^{\mathrm{start}}\, w_i^{\mathrm{end}}}\right)$ \hfill \COMMENT{Bhattacharyya coefficient~\cite{bhattacharyya1943}}
\STATE \textbf{SLERP interpolation:}
\FOR{$k = 0, 1, \ldots, f$}
  \STATE $t \leftarrow k / f$
  \STATE $\boldsymbol{\eta}(t) \leftarrow \dfrac{\sin\!\left((1-t)\,\Omega\right)}{\sin\Omega}\,\boldsymbol{\eta}^{\mathrm{start}} +
\dfrac{\sin\!\left(t\,\Omega\right)}{\sin\Omega}\,\boldsymbol{\eta}^{\mathrm{end}}$
  \STATE \textbf{Map back to weights:} $w_i(k) \leftarrow \left[\eta_i(t)\right]^2$ for $i = 1, \ldots, N$
\ENDFOR
\end{algorithmic}
\end{algorithm}
As established in Corollary~\ref{cor:slerp_optimal}, consecutive points on the SLERP path are separated by a constant angular increment of $\Omega/f$, giving per-step loss
\begin{equation}
    \mathcal{L}_k = (\mathrm{d}s_k)^2 = \frac{2\Omega^2}{f^2}
    \label{eq:slerp_perstep}
\end{equation}
for all $k$, and total loss $\mathcal{L}_{\mathrm{total}} = 2\Omega^2/f$, which decreases as $1/f$ with the number of steps.

\subsection{Two-token case}

For $N=2$ tokens with weights $(w,\, 1-w)$, the sphere $S^1$ is a quarter-circle in the positive quadrant.
The Hellinger embedding maps $w \mapsto (\sqrt{w},\, \sqrt{1-w})$, parameterised by a single angle $\theta = \arcsin(\sqrt{w})$.
SLERP reduces to linear interpolation in $\theta$:
\begin{equation}
    w(k) = \sin^2\!\left(\theta^{\mathrm{start}} + \frac{k}{f}\left(\theta^{\mathrm{end}} - \theta^{\mathrm{start}}\right)\right),
    \label{eq:two_token_slerp}
\end{equation}
where $\theta^{\mathrm{start}} = \arcsin(\sqrt{w^{\mathrm{start}}})$ and $\theta^{\mathrm{end}} = \arcsin(\sqrt{w^{\mathrm{end}}})$.

\subsection{SLERP midpoint equals (AM+GM)/normalise midpoint}
\label{sec:slerp_equals_amgm}

Does the (AM+GM)/normalise heuristic~\cite{willetts2024optimalrebalancingdynamicamms}, Eq~\eqref{eq:approx_optimal_traj}, connect to the SLERP geodesics?
The midpoints of the two trajectories are identical.

\begin{restatable}{theorem}{slerpamgm}
\label{thm:slerp_amgm}
For any $N$-token pool with start weights $\v w^{\mathrm{start}}$ and end weights $\v w^{\mathrm{end}}$, the SLERP midpoint
\[
\boldsymbol{\eta}(t) = \frac{\sin\!\bigl((1-t)\,\Omega\bigr)}{\sin\Omega}\,\boldsymbol{\eta}^{\mathrm{start}} +
\frac{\sin\!\bigl(t\,\Omega\bigr)}{\sin\Omega}\,\boldsymbol{\eta}^{\mathrm{end}}
\]
evaluated at $t = k/f = 1/2$ is exactly equal to the (AM+GM)/\allowbreak normalise midpoint of~\cite{willetts2024optimalrebalancingdynamicamms}.
\end{restatable}

\begin{proof}
The SLERP midpoint on the sphere is
\[\v\eta_{\mathrm{mid}} = \frac{\v\eta^{\mathrm{start}} + \v\eta^{\mathrm{end}}}{2\cos(\Omega/2)},\]
giving weights
\begin{equation}
    w_i^{\mathrm{SLERP}} = \frac{\left(\sqrt{w_i^{\mathrm{start}}} + \sqrt{w_i^{\mathrm{end}}}\right)^2}{\sum_j \left(\sqrt{w_j^{\mathrm{start}}} + \sqrt{w_j^{\mathrm{end}}}\right)^2},
    \label{eq:slerp_midpoint_expanded}
\end{equation}
where we have used $4\cos^2(\Omega/2) = 2(1+\cos\Omega) = \|\v\eta^{\mathrm{start}} + \v\eta^{\mathrm{end}}\|^2$.
Expanding the numerator via $(\sqrt{a}+\sqrt{b})^2 = a + b + 2\sqrt{ab}$:
\begin{equation}
    \left(\sqrt{w_i^{\mathrm{start}}} + \sqrt{w_i^{\mathrm{end}}}\right)^2 = 2\!\left(\underbrace{\frac{w_i^{\mathrm{start}} + w_i^{\mathrm{end}}}{2}}_{\mathrm{AM}_i} + \underbrace{\sqrt{w_i^{\mathrm{start}}\, w_i^{\mathrm{end}}}}_{\mathrm{GM}_i}\right).
    \label{eq:square_identity}
\end{equation}
The factor of 2 cancels upon normalisation:
\begin{equation}
    w_i^{\mathrm{SLERP}} = \frac{\mathrm{AM}_i + \mathrm{GM}_i}{\sum_j \left(\mathrm{AM}_j + \mathrm{GM}_j\right)} = \breve{w}_i,
\end{equation}
which is precisely the (AM+GM)/normalise formula (Eq~\eqref{eq:approx_optimal_traj}, from~\cite{willetts2024optimalrebalancingdynamicamms} at $t=1/2$).
\end{proof}

This identity holds for any number of tokens and any magnitude of weight change.
The heuristic midpoint is the geodesic midpoint.
See Figure~\ref{fig:midpoint} for an informal visual proof.

\begin{figure*}[h]
\centering
    \includegraphics[width=\textwidth]{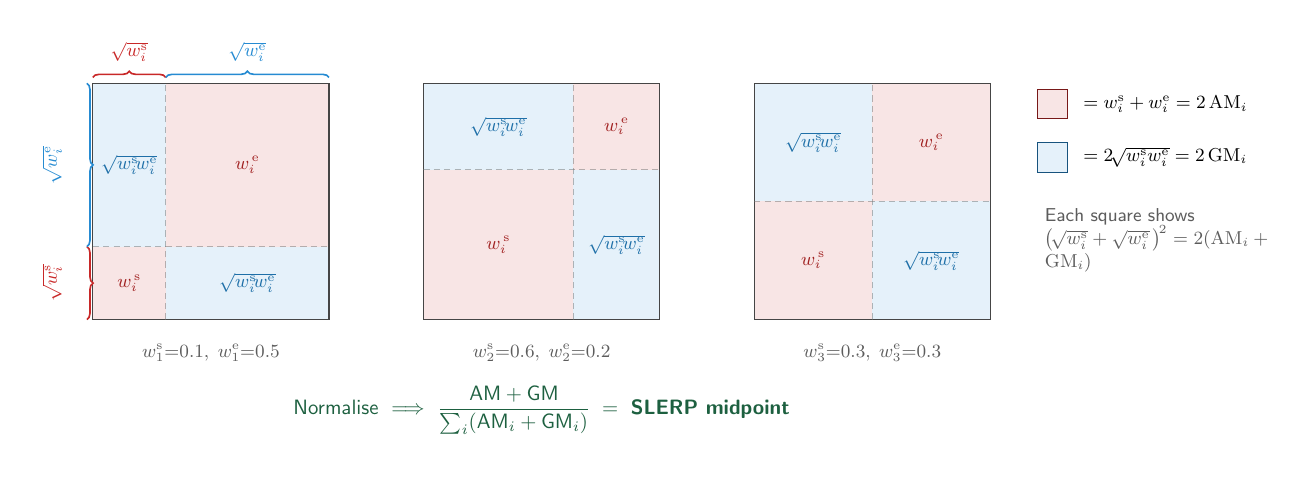}
    \caption{Visual proof that the SLERP midpoint equals (AM+GM)/normalise (in the N=3 case).
    To return to weight space we undo the Hellinger embedding $\eta_i = \sqrt{w_i}$, i.e.\ \emph{square}.
    Each square decomposes:
    $(\!\sqrt{w_i^{\mathrm{s}}} + \sqrt{w_i^{\mathrm{e}}})^2
    = (w_i^{\mathrm{s}} + w_i^{\mathrm{e}}) + 2\!\sqrt{w_i^{\mathrm{s}}\, w_i^{\mathrm{e}}}$
    ${}= 2(\mathrm{AM}_i + \mathrm{GM}_i)$.
    Red diagonal blocks contribute $2\,\mathrm{AM}_i$; blue off-diagonal blocks contribute $2\,\mathrm{GM}_i$.
    The factor of~2 cancels with normalisation, giving us the (AM+GM)/normalise heuristic exactly.}
    \label{fig:midpoint}
\end{figure*}

\subsection{Trig-free recursive bisection}
\label{sec:bisection}

\begin{corollary}[Trig-free SLERP via recursive bisection]
\label{cor:bisection}
Let $f = 2^d$ for some integer $d \geq 1$.
The full $f$-step SLERP trajectory $\v w(k/f)$ for $k = 0, 1, \ldots, f$ can be computed by recursive midpoint bisection: at each level, insert the (AM+GM)/\allowbreak normalise midpoint between each consecutive pair of weights.
After $d$ levels, the resulting $2^d + 1$ points lie exactly on the SLERP geodesic.
\end{corollary}
\begin{proof}
SLERP restricted to any sub-arc of a great circle is itself a SLERP between the sub-arc's endpoints.
By Theorem~\ref{thm:slerp_amgm}, the (AM\allowbreak+GM)\allowbreak/normalise midpoint equals the SLERP midpoint of that sub-arc.
The result follows by induction on the bisection depth~$d$.
\end{proof}

This means an on-chain implementation needs only addition, multiplication, square root, and normalisation (no $\arccos$ or $\sin$ evaluations) to produce the optimal interpolation at power-of-two step counts.

\begin{algorithm}[b]
\caption{Trig-free SLERP via recursive bisection}
\label{alg:bisection}
\begin{algorithmic}[1]
\REQUIRE Start weights $\mathbf{w}^{\mathrm{start}} \in \Delta^{N-1}$, end weights $\mathbf{w}^{\mathrm{end}} \in \Delta^{N-1}$, depth $d$ (gives $f = 2^d$ steps)
\ENSURE $2^d + 1$ points on the SLERP geodesic: $\mathbf{w}^{(0)}, \mathbf{w}^{(1)}, \ldots, \mathbf{w}^{(2^d)}$
\STATE Initialise list $\mathcal{W} \leftarrow [\mathbf{w}^{\mathrm{start}},\; \mathbf{w}^{\mathrm{end}}]$
\FOR{level $\ell = 1$ to $d$}
    \STATE $\mathcal{W}' \leftarrow [\,]$
    \FOR{each consecutive pair $(\mathbf{w}^{(a)}, \mathbf{w}^{(b)})$ in $\mathcal{W}$}
        \STATE Compute midpoint: $m_i \leftarrow \frac{w_i^{(a)} + w_i^{(b)}}{2} + \sqrt{w_i^{(a)}\, w_i^{(b)}}$ for each $i$
        \STATE Normalise: $\hat{m}_i \leftarrow m_i \,/\, \sum_j m_j$
        \STATE Append $\mathbf{w}^{(a)},\; \hat{\mathbf{m}}$ to $\mathcal{W}'$
    \ENDFOR
    \STATE Append $\mathbf{w}^{\mathrm{end}}$ to $\mathcal{W}'$
    \STATE $\mathcal{W} \leftarrow \mathcal{W}'$
\ENDFOR
\RETURN $\mathcal{W}$
\end{algorithmic}
\end{algorithm}

Algorithm~\ref{alg:bisection} makes Corollary~\ref{cor:bisection} practical.
Each level requires only $O(N)$ arithmetic operations per midpoint insertion and no trigonometric functions.
Figure~\ref{fig:bisection} illustrates the construction on the Hellinger sphere.

\begin{figure*}[h]
\centering
    \includegraphics[width=0.8\textwidth]{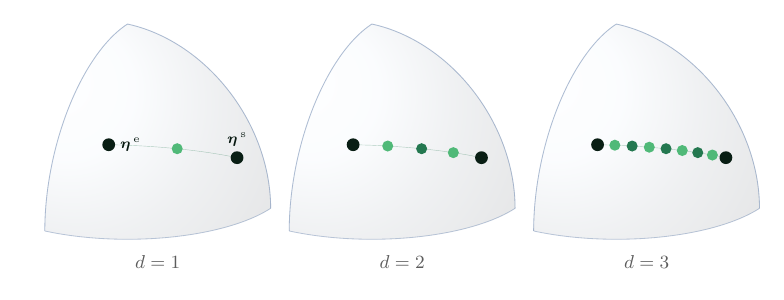}
    \caption{Recursive bisection on the positive orthant of $S^{N-1}_+$.
    At each depth~$d$, the (AM+GM)/normalise midpoint is inserted between every pair of adjacent points (bright green: newly inserted; dark green: inherited from previous depth).
    By depth~3, the $2^d+1 = 9$ points densely trace the geodesic arc, all computed without trigonometric functions.}
    \label{fig:bisection}
\end{figure*}

\subsection{Comparison at general $t$ and with Lambert~W}
\label{ssec:comparison_general_t}

For $t \neq 1/2$, the SLERP and (AM+GM)/normalise trajectories differ.
SLERP weights at general $t$ are
\begin{equation}
    w_i^{\mathrm{SLERP}}(t) = \left[\frac{\sin\!\left((1-t)\,\Omega\right)}{\sin\Omega}\sqrt{w_i^{\mathrm{start}}} + \frac{\sin\!\left(t\,\Omega\right)}{\sin\Omega}\sqrt{w_i^{\mathrm{end}}}\right]^{\!2},
\end{equation}
while (AM+GM)/normalise uses
\[
    \breve{w}_i(t) \propto (1-t)\,w_i^{\mathrm{start}} + t\,w_i^{\mathrm{end}} + \left(w_i^{\mathrm{start}}\right)^{1-t}\!\left(w_i^{\mathrm{end}}\right)^{t}.
\]

\begin{lemma}[Agreement at general $t$]
\label{lem:agreement}
Let $u_i = \Delta w_i / w_i^{\mathrm{start}}$ where $\Delta w_i = w_i^{\mathrm{end}} - w_i^{\mathrm{start}}$.
For any $t \in [0,1]$, the SLERP, (AM+GM)/\allowbreak normalise, and Lambert~W (normalised) trajectories agree through $O(u^2)$ and first differ at~$O(u^3)$.
At the midpoint $t = 1/2$, SLERP and (AM+GM)/normalise agree exactly (Theorem~\ref{thm:slerp_amgm}).
\end{lemma}
\begin{proof}
See Appendix~\ref{app:taylor_agreement}.
\end{proof}

Table~\ref{tab:agreement} summarises these results.
The discrepancy between SLERP / (AM+GM)/normalise and the Lambert~W value arises because they optimise different objectives: Lambert~W maximises the finite-step retention ratio $r$, while SLERP minimises the quadratic approximation to $-\log r$.
For multi-step interpolation (large $f$), the sub-optimality bound (Theorem~\ref{thm:suboptimality}) guarantees that all three converge; the choice is practical (on-chain cost and implementation simplicity) rather than a matter of loss optimality.

\begin{table*}[h]
  \centering
  \caption{Agreement between midpoint formulae in powers of $u_i = \Delta w_i / w_i^{\mathrm{start}}$.}
  \begin{tabular}{lccc}
      \hline
      Pair of methods & Midpoint & General $t$ & Source of $O(u^3)$ difference \\
      \hline
      SLERP vs (AM+GM)/normalise & \textbf{Exact} & $O(u^3)$ & SLERP coefficients vs $(1{+}u)^t$ \\
      SLERP vs Lambert~W (normalised) & $O(u^3)$ & $O(u^3)$ & Quadratic vs exact KL objective \\
      (AM+GM)/normalise vs Lambert~W (normalised) & $O(u^3)$ & $O(u^3)$ & Quadratic vs exact KL objective \\
      \hline
  \end{tabular}
  \label{tab:agreement}
\end{table*}

\subsection{Sub-optimality of SLERP on the exact KL cost}
\label{sec:suboptimality}

SLERP minimises the quadratic approximation to the per-step KL cost (Corollary~\ref{cor:slerp_optimal}).
Theorem~\ref{thm:suboptimality} bounds the gap to the global optimum on the exact cost.

\begin{restatable}{theorem}{suboptimalitybound}[Sub-optimality bound]
\label{thm:suboptimality}
Let $C_{\mathrm{SLERP}}(f)$ and $C_*(f)$ denote the total exact KL cost
\[
\textstyle\sum_{k=1}^f D_{\mathrm{KL}}\!\left(\v w^{(k)} \,\|\, \v w^{(k-1)}\right)
\]
along the $f$-step SLERP path and the globally optimal $f$-step path, respectively.
For weights bounded away from zero ($\min_i w_i \geq \epsilon > 0$) and $f \geq 4\Omega/\epsilon$:
\begin{equation}
    0 \;\leq\; C_{\mathrm{SLERP}}(f) - C_*(f) \;\leq\; \frac{A\,\Omega^4}{f^3},
    \label{eq:suboptimality}
\end{equation}
where $A$ depends on the weight configuration and $\epsilon$ but not on $f$.
The relative sub-optimality is $O(\Omega^2/f^2)$.
\end{restatable}

\begin{proof}
Deferred to Appendix~\ref{app:suboptimality_proof}. The explicit constant is $A = 86\,N/\epsilon^4$.
\end{proof}
The bound is qualitatively informative (rate in $f$ and $\Omega$) rather than quantitatively tight: the constant $A$ is conservative, and the numerics show much smaller gaps than the bound predicts.
For the $N=3$ setup ($\epsilon = 0.05$), the condition $f \geq 4\Omega/\epsilon \approx 42$ is easily satisfied at any practical step count.

Expanding the per-step KL cost beyond the quadratic approximation that SLERP minimises (Appendix~\ref{app:suboptimality_proof}), the leading remainder is $R_k = -\frac{1}{6}\sum_i w_i (u_i^{(k)})^3 + O(u^4)$, which encodes the asymmetry of the KL divergence: it is negative when weights increase ($u_i > 0$, so the true KL cost is \emph{less} than the quadratic) and positive when weights decrease.
SLERP is the geodesic of the symmetric Fisher--Rao metric ($\alpha=0$) and ignores this asymmetry entirely.
The exact KL-optimal path would take slightly larger steps where weights increase and slightly smaller ones where they decrease. 
But the per-step gain is only $O(u^2)$, accumulating to $O(\Omega^4/f^3)$ (Theorem~\ref{thm:suboptimality}).
A $[1,1]$ Pad\'e approximant resums this cubic remainder into a rational function matching the exact KL cost through $O(u^4)$ (Appendix~\ref{app:pade}).

For the $N=3$ setup of \S\nolinebreak\ref{sec:experiments} ($\Omega \approx 0.52$), the bound gives relative sub-optimality $O(10^{-4})$ at $f=50$, which is what the brute-force numerics show.

\section{Weights change with stochastic prices}
\label{sec:stochastic}

The preceding section models prices as constant over interpolation steps.
We now extend to stochastic prices.
We model each asset price as a driftless geometric Brownian motion, $\mathrm{d}P_i/P_i = \sigma_i\,\mathrm{d}W_i$~\cite{black1973}, so prices are martingales, and show that the core results are unaffected.

\paragraph{Price-independence of the retention ratio}

At time~$t_k$, the pool has weights $\v w^{(k)}$ and is in no-arbitrage equilibrium at prices $\v P^{(k)}$.
At~$t_{k+1}$, the weights update to $\v w^{(k+1)}$ and prices change to $\v P^{(k+1)}$.
Equating the two expressions for the G3M invariant under the new weights, before arbitrage (old reserves, new weight exponents) and after (new equilibrium at the updated prices), yields a multiplicative factorisation of the value change:
\begin{equation}
\frac{V_{k+1}}{V_k} = \underbrace{\prod_i \!\left(\frac{w_i^{(k)}}{w_i^{(k+1)}}\right)^{\!w_i^{(k+1)}}}_{r} \;\cdot\; \underbrace{\prod_i \!\left(\frac{P_i^{(k+1)}}{P_i^{(k)}}\right)^{\!w_i^{(k+1)}}}_{\text{price return at new weights}}
\label{eq:value_factorisation}
\end{equation}

\begin{proposition}[Price-independence]
\label{prop:price_independence}
The retention ratio $r = \prod_i (w_i^{(k)}/w_i^{(k+1)})^{w_i^{(k+1)}}$ depends only on the weight change, not on the prices.
In particular, the KL divergence characterisation of Theorem~\ref{thm:kl_divergence} holds identically under driftless GBM.
\end{proposition}
\begin{proof}
Prices appear only in the second factor of Eq~\eqref{eq:value_factorisation}; they cancel completely within~$r$.
\end{proof}

In log form, the value change decomposes as
\[
\log\!\frac{V_{k+1}}{V_k}
  = -D_{\mathrm{KL}}\!\bigl(\v w^{(k+1)} \big\| \v w^{(k)}\bigr)
    + \sum_i w_i^{(k+1)}\log\!\frac{P_i^{(k+1)}}{P_i^{(k)}}\,.
\]
The first term is the rebalancing cost from Theorem~\ref{thm:kl_divergence}; the second is a weighted log-price return evaluated at the \emph{new} weights.

\paragraph{The cross term and its telescoping}

Comparing the pool that changes weights to one that holds weights fixed, the effective cost of the weight change is
\[
-\log \frac{V'_{\text{change}}}{V'_{\text{no change}}} = D_{\mathrm{KL}}(\v w^{(k+1)} \| \v w^{(k)}) - \sum_i \Delta w_i \log \frac{P_i'}{P_i}
\]
where $\Delta w_i = w_i^{(k+1)} - w_i^{(k)}$.
The second term couples the weight change to the price change.
Increasing weight in an appreciating asset is cheaper than the KL cost alone because the price move partially offsets the rebalancing.  For a depreciating asset, it is more expensive.

Over the full $f$-step trajectory, the total expected log-loss under driftless GBM is
\begin{equation}
\mathbb{E}\!\left[-\log \frac{V_T}{V_0}\right] = \sum_{k=1}^f D_{\mathrm{KL}}(\v w^{(k)} \| \v w^{(k-1)}) + \frac{\Delta t}{2} \sum_{k=1}^f \sum_i w_i^{(k)} \sigma_i^2
\label{eq:total_loss_stochastic}
\end{equation}
where $\mathbb{E}[\log(P_i'/P_i)] = -\sigma_i^2 \Delta t / 2$ under driftless GBM.

\begin{proposition}[Cross-term telescoping]
\label{prop:telescoping}
The cross term between weight and price changes does not affect the optimal interpolation path.
\end{proposition}
\begin{proof}
The second sum in Eq~\eqref{eq:total_loss_stochastic} uses the new weights $w_i^{(k)}$ at each step.
Rewriting in terms of the old weights $w_i^{(k-1)}$:
\[
\sum_{k=1}^f \sum_i w_i^{(k)} \sigma_i^2 = \sum_{k=1}^f \sum_i w_i^{(k-1)} \sigma_i^2 + \sum_i \bigl(w_i^{(f)} - w_i^{(0)}\bigr) \sigma_i^2.
\]
The correction is endpoint-only, so it drops out of the path variation.
\end{proof}

\begin{corollary}[Survival of the constant-price results]
\label{cor:survival}
Under driftless GBM prices, the following results of \S\ref{sec:background}--\S\ref{sec:suboptimality} hold without modification for the rebalancing cost component:
(i)~KL divergence characterisation of the per-step loss (Theorem~\ref{thm:kl_divergence});
(ii)~SLERP optimality for the leading-order expansion (Corollary~\ref{cor:slerp_optimal});
(iii)~the (AM+GM)/normalise identity (Theorem~\ref{thm:slerp_amgm});
(iv)~recursive bisection (Corollary~\ref{cor:bisection});
(v)~the sub-optimality bound $O(\Omega^4/f^3)$ (Theorem~\ref{thm:suboptimality}).
\end{corollary}
\begin{proof}
Immediate from Propositions~\ref{prop:price_independence} and~\ref{prop:telescoping}: the retention ratio depends only on weights, and the price cross-term is path-independent.\footnote{Note that the telescoping in Proposition~\ref{prop:telescoping} is an Abel summation identity. An alternative ordering of weight and price changes is discussed in Appendix~\ref{app:ordering}.}
Items (iii)--(v) are algebraic properties of the Hellinger embedding that do not involve prices.
\end{proof}

The path optimisation therefore reduces to minimising the rebalancing cost $\sum_k D_{\mathrm{KL}}$ plus a path-dependent potential from the variance-drag or LVR term.
SLERP remains the optimal path for the rebalancing cost component; the potential introduces a correction analysed in \S\ref{sec:optimal_f}.
The precise form of this potential depends on the LP's objective; we distinguish the two main cases below.

\paragraph{Log-return maximisation}

Under log-return maximisation, the potential is the variance drag $V_{\log}(\v w) = \frac{1}{2}\sum_i w_i \sigma_i^2$, which is linear on the simplex; if all assets have the same volatility, it is constant and SLERP is exactly optimal regardless of the number of steps.

\paragraph{Expected-value maximisation}

Under expected-value maximisation, since GBM increments are independent across blocks, the expected value factorises as
\[
\mathbb{E}[V_T/V_0] = \prod_k r_k \cdot \prod_k
\exp\bigl(-\ell(\v w^{(k)})\,\Delta t\bigr)
\]
where $\ell$ is the LVR rate~\cite{lvr,lvr_arb}

\begin{equation}
\ell(\v w) = \tfrac{1}{2}\!\left[\textstyle\sum_i w_i \sigma_i^2 - \v w^T \Sigma \v w\right]
\label{eq:lvr}
\end{equation}
with $\Sigma_{ij} = \sigma_i \sigma_j \rho_{ij}$ (Appendix~\ref{app:lvr_potential} gives the moment-generating-function derivation).
For uncorrelated equal-volatility assets, $\ell(\v w) = \frac{\sigma^2}{2}(1 - \|\v w\|^2)$: maximised at the centroid and zero at the vertices.
Positive correlations reduce $\ell$; in the equal-correlation case $\rho_{ij} = \rho$, the LVR rate becomes $\ell = \frac{\sigma^2}{2}(1-\rho)(1 - \|\v w\|^2)$, and the rest of the analysis goes through with $\sigma^2 \to \sigma^2(1-\rho)$.

We focus on expected-value maximisation for the remainder; all results apply to log-return maximisation with $\ell$ replaced by $V_{\log}$.

\section{Optimal step count under LVR}
\label{sec:optimal_f}

Under constant prices, the total rebalancing cost $2\Omega^2/f$ can be made arbitrarily small by increasing the step count~$f$, and the optimal strategy is simply to maximise~$f$.
Under GBM, however, each step occupies one block and thereby exposes the pool to one additional block of LVR, creating a tradeoff between rebalancing cost and volatility exposure.

\paragraph{The modified geodesic equation}

In the continuum limit, the combined cost of rebalancing and LVR along a path $\v w(s)$, $s \in [0,1]$, is the functional
\[
\mathcal{C}[\v w(\cdot)] = \frac{1}{2f} \int_0^1 g_{ij} \dot{w}^i \dot{w}^j \, ds + T \int_0^1 \ell(\v w(s)) \, ds
\]
where $T = f\,\Delta t_{\mathrm{block}}$ is the total rebalancing duration.
This is a Lagrangian mechanics problem: the Fisher--Rao term is the kinetic energy and $\ell$ is the potential.

\begin{restatable}{proposition}{modifiedgeodesic}[Modified geodesic equation]
\label{prop:modified_geodesic}
The Euler--Lagrange equation for $\mathcal{C}$ is
\begin{equation}
\frac{1}{f}\!\left(\ddot{w}^i + \Gamma^i_{jk} \dot{w}^j \dot{w}^k\right) = T\, g^{ij} \frac{\partial \ell}{\partial w^j}
\label{eq:modified_geodesic}
\end{equation}
where $\Gamma^i_{jk}$ are the Christoffel symbols of the Fisher--Rao metric~\cite{hobson2006}.
\end{restatable}
\begin{proof}
Deferred to Appendix~\ref{app:modified_geodesic_proof}.
\end{proof}

\noindent When the right-hand side vanishes ($T\sigma^2 = 0$), this reduces to the free geodesic equation whose solution is SLERP; the LVR potential acts as a forcing term that deflects the geodesic away from high-$\ell$ regions of the simplex.

\begin{restatable}{corollary}{pendulum}[Two-token pendulum]
\label{cor:pendulum}
For two tokens with one volatile asset ($\sigma > 0$) and a num\'eraire ($\sigma_2 = 0$), the Hellinger coordinate $\theta = \arcsin(\sqrt{w})$ satisfies
\begin{equation}
\ddot{\theta} = \frac{fT\sigma^2}{16}\sin(4\theta)
\label{eq:pendulum}
\end{equation}
\end{restatable}
\begin{proof}
Deferred to Appendix~\ref{app:modified_geodesic_proof}.
\end{proof}

\noindent The potential is maximised at equal weighting ($\theta = \pi/4$, maximum LVR) and zero at the vertices, so the LVR forcing pushes paths toward concentrated portfolios.
When all assets share the same volatility, the potential is constant on the simplex and the equation collapses to the free geodesic, i.e.\ SLERP.
For $N \geq 3$ with heterogeneous volatilities, the LVR potential landscape has richer critical-point structure; the Jacobi geodesic of Appendix~\ref{app:jacobi} illustrates this numerically for $N=3$.
For $\mu = fT\sigma^2/16 \ll 1$ (the practical regime at $f = f^*$), the first-order correction to the SLERP baseline $\theta_0(s) = \theta_{\mathrm{start}} + \Omega s$ is $\epsilon_1(s) = \mu\!\int_0^1 G(s,s')\sin(4\theta_0(s'))\,ds'$, where $G(s,s') = \min(s,s')(1{-}\max(s,s'))$ is the Green's function for the two-point boundary problem (Appendix~\ref{app:asymptotics}).

\paragraph{The finite optimum}

In TFMM, each interpolation step occupies one block ($\Delta t_{\mathrm{block}}$), so $T = f\,\Delta t_{\mathrm{block}}$ and the total expected cost along a SLERP path decomposes as
\begin{equation}
C(f) \approx \underbrace{\frac{2\Omega^2}{f}}_{\searrow\;\text{in}\;f} + \underbrace{f\,\Delta t_{\mathrm{block}}\,\bar\ell}_{\nearrow\;\text{in}\;f}
\label{eq:cost_tradeoff}
\end{equation}
where $\bar\ell = \int_0^1 \ell(\v w_{\mathrm{SLERP}}(s))\,ds$ is the LVR rate averaged along the fixed SLERP curve.\footnote{For $N \geq 3$ with heterogeneous volatilities, $\bar\ell$ is a quadratic form integrated over a great-circle arc on the Hellinger sphere; it evaluates to a polynomial in trigonometric functions of $\Omega$, so $f^*$ is efficiently computable for any $N$.}
The two terms are plotted in Figure~\ref{fig:cost_tradeoff}.

\begin{remark}[Resolution of the apparent circularity]
\label{rem:circularity}
The tradeoff appears circular: $f^*$ depends on $\bar\ell$, which depends on the SLERP path, which seems to assume $f$ is already chosen.
The circularity is benign.
The SLERP path in weight space is a fixed curve (the great circle on the Hellinger sphere between $\v w^{\mathrm{start}}$ and $\v w^{\mathrm{end}}$) and does not depend on~$f$.
The parameter $f$ only determines how many waypoints are placed along this fixed curve.
Thus $\bar\ell$ is a constant determined entirely by the endpoints, and $f^*$ is well-defined.
At the next order the LVR-corrected path (Appendix~\ref{app:asymptotics}) does depend on $f$ through $\mu = fT\sigma^2/16$, but the induced change to $\bar\ell$ is $O(\mu^2)$, which closes the self-consistency loop.
\end{remark}

\begin{figure}[t]
\centering
\includegraphics[width=0.85\columnwidth]{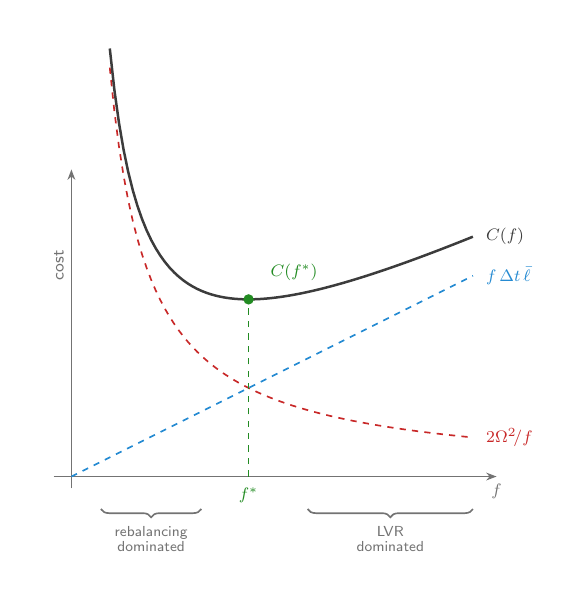}
\caption{The total expected cost $C(f)$ (solid) decomposes into a decreasing rebalancing cost $2\Omega^2/f$ (dashed, red) and an increasing LVR cost $f\,\Delta t\,\bar\ell$ (dashed, blue).
The minimum occurs at $f^*$ (green dot), where the two contributions are equal.}
\label{fig:cost_tradeoff}
\end{figure}

\begin{proposition}[Optimal step count]
\label{prop:optimal_f}
Setting $dC/df = 0$ yields
\begin{equation}
f^* = \Omega\sqrt{\frac{2}{\Delta t_{\mathrm{block}}\,\bar\ell}}
\label{eq:optimal_f}
\end{equation}
At the optimum the rebalancing and LVR costs are equal, and the minimum total cost is $C(f^*) = 2\sqrt{2\,\Omega^2\,\Delta t_{\mathrm{block}}\,\bar\ell}$, which scales as $\Omega\sigma$ rather than $\Omega^2/f$.
\end{proposition}
\begin{proof}
The first term in Eq~\eqref{eq:cost_tradeoff} is convex in $1/f$ and the second is linear in $f$; setting the derivative to zero gives $2\Omega^2/f^2 = \Delta t_{\mathrm{block}}\,\bar\ell$. Substituting back yields both $f^*$ and $C(f^*)$.
\end{proof}
The cost near the minimum is flat: $C(f^* + \delta f) / C(f^*) \approx 1 + (\delta f / f^*)^2$, so $f^*$ is robust to parameter misestimation.
Since $f^* \propto 1/\sigma$, a factor-of-two error in the volatility estimate gives a factor-of-two error in $f^*$ but only ${\sim}25\%$ excess cost.

\paragraph{The optimum is always SLERP-near-optimal}

The dimensionless parameter $\lambda = fT\sigma^2/\Omega^2$ controls the strength of the LVR perturbation to the free geodesic.
Evaluating at $f^*$:
\[
\lambda^* = \frac{2\sigma^2}{\bar\ell}
\]
For two tokens at average weight $\bar{w}$, $\bar\ell \approx \frac{\sigma^2}{2}\bar{w}(1-\bar{w})$, giving $\lambda^* \approx 4/(\bar{w}(1-\bar{w})) \sim 8$--$16$ for typical interior weights ($\bar{w} \in [0.3, 0.7]$): $\lambda^* = O(1)$.
So the optimum sits in the perturbative regime (Appendix~\ref{app:asymptotics}); $\lambda \gg 1$ means $f \gg f^*$, and cutting $f$ back to $f^*$ sheds the excess LVR exposure without changing the interpolation curve.

\paragraph{Worked example}

Consider the rebalancing $(0.5, 0.5) \to (0.7, 0.3)$ with $\sigma \approx 3\%$/day and $\Delta t_{\mathrm{block}} = 12\,\text{s} \approx 1.4\times 10^{-4}$~days:
\[
\Omega \approx 0.21\;\text{rad}, \quad \bar\ell \approx 1.1\times 10^{-4}\;\text{day}^{-1}, \quad f^* \approx 2400\;\text{blocks} \approx 8\;\text{hr}
\]
The first-order correction to SLERP is $O(10^{-2})$ in weight space, consistent with the L-BFGS-B experiments of \S\ref{sec:experiments_stochastic}.
The matched-asymptotic boundary-layer width $\delta \approx 5 > 1$ (Appendix~\ref{app:asymptotics}), so the perturbative expansion is valid across the entire trajectory.

\section{Fee integration}
\label{sec:fees}

G3M pools charge a swap fee $\gamma \in (0,1)$ (e.g.\ $\gamma = 0.997$ for 30\,bp), which creates a no-arbitrage band $[\gamma\, m_p,\; \gamma^{-1} m_p]$ around the pool's marginal price~\cite{willetts2024closedform,diamond2024optimalfees}.
Arbitrage is profitable only when the marginal price exits this band and the expected profit exceeds the arbitrageur's execution cost (gas, priority fees, capital).
The \emph{effective} no-arb region is therefore wider than the fee band alone: on high-gas chains the gas threshold can dominate the fee threshold, while on low-gas chains the fee band is the binding constraint.
Both effects are empirically observable in the sawtooth dynamics of live TFMM pools~\cite{willetts2026poolsportfoliosobservedarbitrage}: allocation drift accumulates linearly over blocks until the effective threshold is breached, an arbitrageur trades, and the cycle restarts.

The preceding sections derive exact results under zero fees.
We now show that the path optimality of SLERP extends to the fee-adjusted setting: fees create a constant offset to the total cost that does not depend on the interpolation path, so the path optimisation and the fee accounting decouple.

\paragraph{The no-arbitrage band and weight changes}

Along SLERP, each step advances the Hellinger coordinate by $\Omega/f$, giving per-component displacements $|\delta\eta_i| \leq \Omega/f$.
Since $\delta w_i = 2\eta_i\,\delta\eta_i$ and $\eta_i = \sqrt{w_i}$, the relative weight change satisfies
$|\delta w_i / w_i| \leq 2\Omega/(f\sqrt{w_{\min}})$.
The G3M marginal price ratio satisfies $\delta m_{ij}/m_{ij} = \delta w_i/w_i - \delta w_j/w_j$,
so the worst-case price displacement per step is at most $4\Omega/(f\sqrt{w_{\min}})$.
Per-step weight changes fall entirely inside the fee band when
\begin{equation}
f \geq f_{\mathrm{threshold}} = \frac{4\Omega}{\sqrt{w_{\min}}\,(1-\gamma)}.
\label{eq:f_threshold}
\end{equation}
For the worked example of \S\ref{sec:optimal_f} ($\Omega = 0.21$, $w_{\min} = 0.3$, $\gamma = 0.997$): $f_{\mathrm{threshold}} \approx 510$.
Above $f_{\mathrm{threshold}}$, individual weight changes do not trigger standalone arbitrage; displacement accumulates over multiple steps until either the cumulative drift breaches the fee band or a price-driven arb event clears it.\footnote{In the idealised constant-price model, the cost above $f_{\mathrm{threshold}}$ is exactly $2\Omega^2/f_{\mathrm{threshold}}$, independent of $f$ (Appendix~\ref{app:fee_mechanics}).}

\paragraph{Price-driven arbitrage rate}

Under GBM, the per-block log-price change has variance $\sigma^2\Delta t_{\mathrm{block}}$; the expected number of blocks for the price to exit the fee band (half-width $\approx (1{-}\gamma)$ in log terms) is $(1{-}\gamma)^2/(\sigma^2\Delta t_{\mathrm{block}})$.\footnote{Standard exit-time result for a symmetric random walk from an interval of half-width~$a$: $\mathbb{E}[\tau] = a^2/\mathrm{Var}(\text{step})$.}
The fraction of blocks in which price-driven arb fires is therefore
\begin{equation}
\nu = \frac{\sigma^2 \Delta t_{\mathrm{block}}}{(1-\gamma)^2}.
\label{eq:nu}
\end{equation}
For Ethereum L1 ($\sigma \approx 3\%$/day, $\Delta t_{\mathrm{block}} = 12\,\text{s}$, $1{-}\gamma = 0.003$): $\nu \approx 0.014$, i.e.\ price-driven arb fires roughly every $1/\nu \approx 72$ blocks.
The fee band width $(1{-}\gamma)$ governs both scales: $f_{\mathrm{threshold}} \propto (1{-}\gamma)^{-1}$ for weight-driven absorption and $1/\nu \propto (1{-}\gamma)^{2}$ for price-driven arb.
At the practical operating point, weight drift breaches the band ($n_w = f/f_{\mathrm{threshold}} \approx 5$ blocks) well before price-driven arb would fire ($1/\nu \approx 72$ blocks).

\subsection{SLERP optimality with fees}

The pool's total value (reserves plus accumulated fee revenue) changes by exactly $-\pi$ per arb event, where $\pi$ is the arbitrageur's profit: the fee revenue $(1{-}\gamma)\,\delta_{\mathrm{in}}\, p_{\mathrm{in}}$ stays with the pool, so only the arber's net extraction is lost.
The total cost over $f$ steps is therefore $\sum_k \pi_\gamma^{(k)}$.
The following result shows that fee revenue creates a path-independent offset, so the path optimisation reduces to the zero-fee problem.

\begin{proposition}[Fee-revenue telescoping]
\label{prop:fee_telescoping}
For a monotonic interpolation path from $\v w^{\mathrm{start}}$ to $\v w^{\mathrm{end}}$ with per-step weight changes small relative to the fee band, the total fee revenue across all arb events is
\begin{equation}
    F_{\mathrm{total}} = \frac{1-\gamma}{2\gamma}\, V \sum_{i=1}^N \bigl|w_i^{\mathrm{end}} - w_i^{\mathrm{start}}\bigr| \;+\; O\!\left((1{-}\gamma)^2\right),
    \label{eq:fee_telescoping}
\end{equation}
which depends on the endpoints and $\gamma$ but not on the interpolation path or step count.
\end{proposition}
\begin{proof}
At equilibrium, $R_i\, p_i = V w_i$.
The optimal arb trade to correct a weight displacement $\Delta w_i$ at leading order has inflow
$\delta_{\mathrm{in}} \approx V\!\sum_{i:\Delta w_i > 0}\!|\Delta w_i|/\gamma$~\cite{willetts2024closedform,evansG3Ms}, so the per-event fee revenue is $(1{-}\gamma)\,\delta_{\mathrm{in}} \approx \frac{1-\gamma}{\gamma}\,V \cdot \tfrac{1}{2}\sum_i |\Delta w_i^{(k)}|$ (using $\sum_i \Delta w_i = 0$).
A path is \emph{monotonic} if each component $w_i$ moves in one direction across all steps.
For such paths, the absolute changes telescope: $\sum_k |\Delta w_i^{(k)}| = |w_i^{\mathrm{end}} - w_i^{\mathrm{start}}|$.
Summing over components gives Eq~\eqref{eq:fee_telescoping}.

The correction arises because after a fee-adjusted arb the pool's reserves are not at exact equilibrium (there is a residual displacement of $O(1{-}\gamma)$ inside the fee band).
This perturbs the trade size for the next step by $O(1{-}\gamma)$, giving $O((1{-}\gamma)^2)$ correction to the fee revenue per step and $O(f(1{-}\gamma)^2)$ total.
For $\gamma = 0.997$ this is ${<}0.3\%$ of the leading term.
\end{proof}

\begin{corollary}[SLERP optimality with fees]
\label{cor:slerp_fees}
Under the leading-order expansion of $D_{\mathrm{KL}}$, SLERP minimises the pool's total net loss (arb extraction minus fee revenue) among all interpolation paths from $\v w^{\mathrm{start}}$ to $\v w^{\mathrm{end}}$.
Fees affect the optimal step count but not the optimal interpolation curve.
\end{corollary}
\begin{proof}
Deferred to Appendix~\ref{app:slerp_fees_proof}.
The argument shows that SLERP simultaneously minimises the KL cost (Corollary~\ref{cor:slerp_optimal}) and generates at least as much fee revenue as any alternative path.
\end{proof}

\subsection{The net cost-of-carry and optimal step count}

The total fee revenue $F_{\mathrm{total}}$ is one of several mechanisms that reduce the pool's cost relative to the zero-fee prediction.
Others include arber under-extraction (54--71\% of optimal trade volume, capturing only 61--89\% of available profit~\cite{willetts2026poolsportfoliosobservedarbitrage}), and on low-gas chains, incidental routing by DEX aggregators that rebalances the pool while paying fees~\cite{willetts2026poolsportfoliosobservedarbitrage}.
We capture the combined effect of all mechanisms in a single empirical observable, the net cost-of-carry fraction:
\[
\phi = \frac{\text{time-averaged pool loss (net of fees)}}{\text{theoretical zero-fee LVR rate}\;\ell(\v w)\Delta t}.
\]
On L1, $\phi \in (0.1, 0.5)$; on L2 with sufficient routing flow, $\phi$ can be negative (pool is a net fee earner)~\cite{willetts2026poolsportfoliosobservedarbitrage}.
Empirically, $\phi$ decreases over time as the arbitrageur ecosystem matures~\cite{willetts2026poolsportfoliosobservedarbitrage}.
For a new pool with no operating history, $\phi$ can be estimated from comparable pools on the same chain (same gas regime, fee tier, and similar TVL); the ranges reported in~\cite{willetts2026poolsportfoliosobservedarbitrage} provide initial calibration.

\begin{proposition}[Fee-adjusted optimal step count]
\label{prop:fee_adjusted_f}
For $\phi > 0$, the fee-adjusted cost along a SLERP path is
\begin{equation}
    C_\phi(f) = \frac{2\Omega^2}{\min(f,\; f_{\mathrm{threshold}})} + f\,\Delta t_{\mathrm{block}}\,\bar\ell\,\phi,
    \label{eq:fee_adjusted_cost}
\end{equation}
with optimum
\begin{equation}
    f^*_\phi = \min\!\left(\Omega\sqrt{\frac{2}{\Delta t_{\mathrm{block}}\,\bar\ell\,\phi}},\;\; f_{\mathrm{threshold}}\right).
    \label{eq:f_star_phi}
\end{equation}
When $\phi \leq 0$, $C_\phi$ is non-increasing in $f$ and no finite optimum exists.
\end{proposition}
\begin{proof}
The first term is $2\Omega^2/f$ for $f \leq f_{\mathrm{threshold}}$ (Corollary~\ref{cor:slerp_optimal}) and plateaus at $2\Omega^2/f_{\mathrm{threshold}}$ for $f > f_{\mathrm{threshold}}$ (Appendix~\ref{app:fee_mechanics}): the fee band batches weight changes into chunks of displacement $\Omega/f_{\mathrm{threshold}}$, and the quadratic cost per chunk times the number of chunks is independent of~$f$.
The second term replaces the bare LVR rate with the empirical $\phi\,\bar\ell$ (incorporating all fee, gas, and microstructure effects).
For $\phi > 0$ and $f \leq f_{\mathrm{threshold}}$: $dC_\phi/df = -2\Omega^2/f^2 + \Delta t_{\mathrm{block}}\,\bar\ell\,\phi = 0$ gives $f^* = \Omega\sqrt{2/(\Delta t_{\mathrm{block}}\,\bar\ell\,\phi)}$; the plateau caps this at $f_{\mathrm{threshold}}$.
For $\phi \leq 0$: $dC_\phi/df \leq 0$ everywhere, so $C_\phi$ is non-increasing.
\end{proof}

\paragraph{Two deployment regimes}

The distinction between deployment contexts is primarily about gas costs, not swap fees (both regimes use standard fees, e.g.\ 30\,bp).

On \emph{high-gas chains}, gas costs raise the effective arb threshold above the fee band and make multi-hop routing uneconomical.
Empirically $\phi \in (0.1, 0.5)$~\cite{willetts2026poolsportfoliosobservedarbitrage}: the finite-$f$ tradeoff of Proposition~\ref{prop:fee_adjusted_f} applies, and $f_{\mathrm{threshold}}$ (Eq~\eqref{eq:f_threshold}) caps the useful step count.

On \emph{low-gas chains} with sufficient routing flow, DEX aggregators route through the pool as one leg of broader cross-venue trades, paying the same 30\,bp fee while extracting profit from other venues.
On Base L2, 98\% of trades in a TFMM pool extracted less than \$0.01, with many subsidising the pool~\cite{willetts2026poolsportfoliosobservedarbitrage}: this incidental routing pushes $\phi \leq 0$, and $f_{\mathrm{threshold}}$ does not bind.

\section{Experiments}
\label{sec:experiments}

\subsection{Constant prices}
\label{sec:experiments_constant}

We compare interpolation methods numerically using the same $N=3$ token setup as~\cite{willetts2024optimalrebalancingdynamicamms}: $\v w(t_0)=\{0.05, 0.55, 0.4\}$, $\v w(t_f)=\{0.4, 0.5, 0.1\}$, $f=1000$ steps.

\paragraph{Weight trajectories and block-to-block changes}

Weight trajectories and block-to-block changes are plotted in Appendix~\ref{app:trajectories} (Figures~\ref{fig:traj} and~\ref{fig:change}).
At this scale, (AM+GM)/normalise and SLERP are visually indistinguishable, as expected from the exact midpoint equivalence (Theorem~\ref{thm:slerp_amgm}) and third-order agreement at general $t$ (\S\ref{ssec:comparison_general_t}).
For SLERP, block-to-block changes vary smoothly; constant metric speed in the Fisher--Rao metric does not correspond to constant weight increments in Cartesian coordinates, but the variation in increment size is small ($\sim 10^{-4}$ relative to the mean increment).

\paragraph{Per-step loss uniformity}

\begin{figure*}[htbp]
    \centering
    \includegraphics[width=\textwidth]{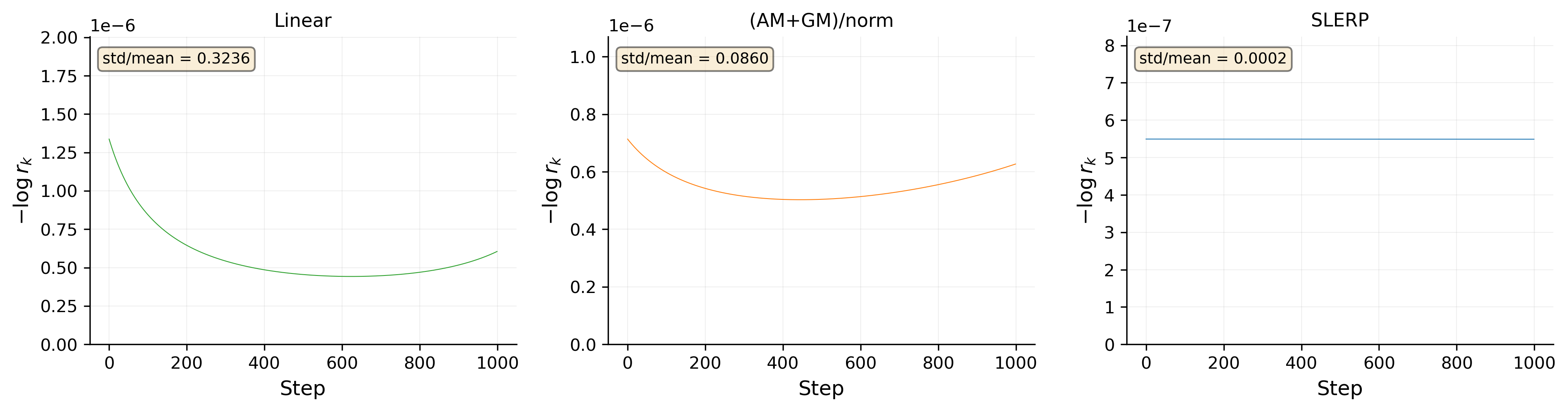}
    \caption{Per-step arbitrage loss $-\log r_k$ across 1000 steps, $N=3$ setup. Left to right: linear (std/mean $= 0.32$), (AM+GM)/normalise ($0.086$), SLERP ($0.0002$). SLERP achieves near-perfect uniformity, orders of magnitude better than other methods.}
    \label{fig:perstep_loss}
\end{figure*}

Figure~\ref{fig:perstep_loss} shows the per-step arbitrage loss, $-\log r_k$ across all 1000 steps, to directly test how different interpolation methods give different variation in loss per step.

\begin{table}[htbp]
    \centering
    \caption{Per-step loss uniformity: ratio of standard deviation to mean of $-\log r_k$.}
    \begin{tabular}{lc}
    \hline
    Method & std/mean \\
    \hline
    Linear & 0.3236 \\
    Geometric & 0.2155 \\
    (AM+GM)/normalise & 0.0860 \\
    SLERP & \textbf{0.0002} \\
    \hline
    \end{tabular}
    \label{tab:uniformity}
\end{table}

SLERP achieves near-perfect loss uniformity (Table~\ref{tab:uniformity}): constant-speed traversal of the geodesic equalises the per-step cost by construction.
The residual $0.02\%$ variation comes from higher-order terms beyond the quadratic approximation (Eq~\eqref{eq:tfmm_loss_kernel}).

\begin{figure*}[htbp]
    \centering
    \begin{subfigure}{\textwidth}
        \centering
        \includegraphics[width=\textwidth]{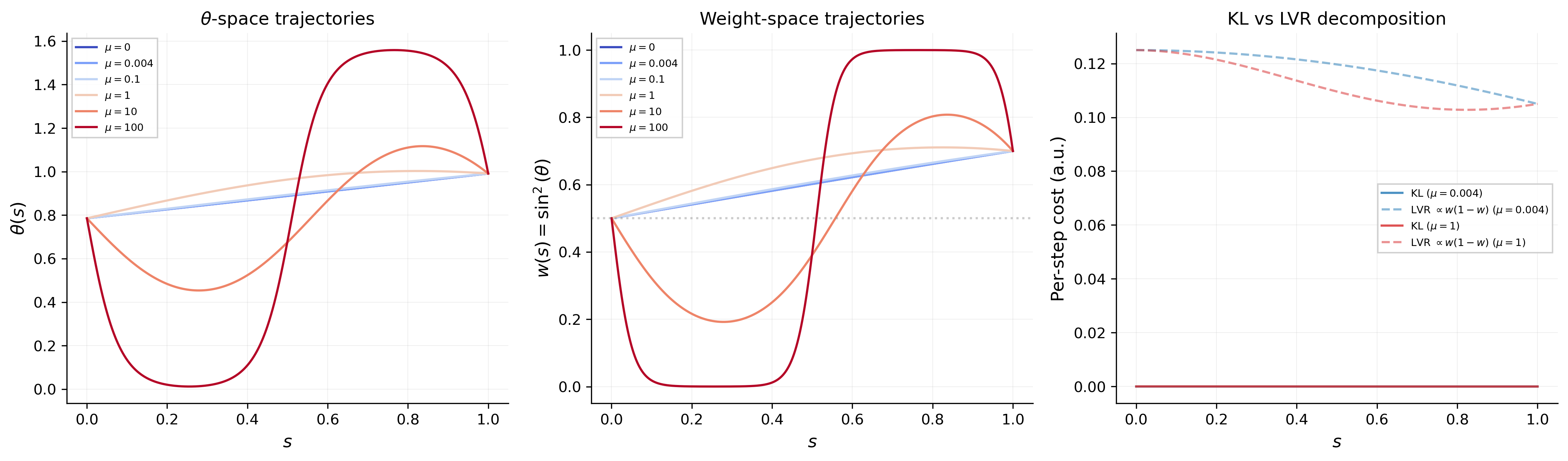}
        \caption{Two-token pendulum (Corollary~\ref{cor:pendulum}), $(0.5, 0.5) \to (0.7, 0.3)$, varying $\mu = fT\sigma^2/16$. Left: $\theta$-space trajectories. Centre: weight-space trajectories. Right: per-step KL and LVR cost decomposition for $\mu = 0.004$ and $\mu = 1$. At practical parameters ($\mu \ll 1$), the trajectory is indistinguishable from SLERP.}
        \label{fig:pendulum_traj}
    \end{subfigure}
    \begin{subfigure}{\textwidth}
        \centering
        \includegraphics[width=\textwidth]{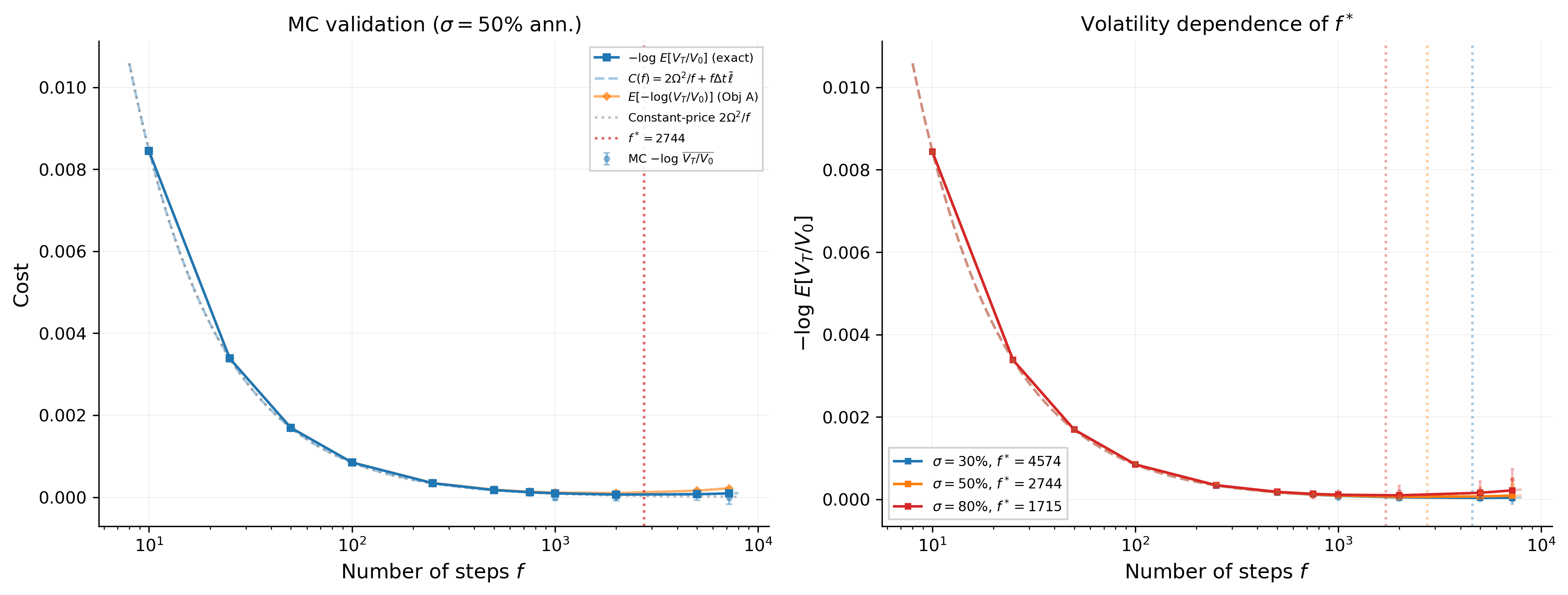}
        \caption{Left: expected log-loss vs.\ step count at $\sigma = 50\%$ ann.\ ($10\,000$ GBM paths per point); the analytical $C(f)$ (Eq~\eqref{eq:cost_tradeoff}) matches the MC estimate, and the constant-price curve $2\Omega^2/f$ underestimates above $f^*$. Right: $f^*$ vs.\ volatility; predicted values (dotted vertical lines) match the simulated optima, confirming $f^* \propto 1/\sigma$.}
        \label{fig:mc_validation}
    \end{subfigure}
    \begin{subfigure}{\textwidth}
        \centering
        \includegraphics[width=\textwidth]{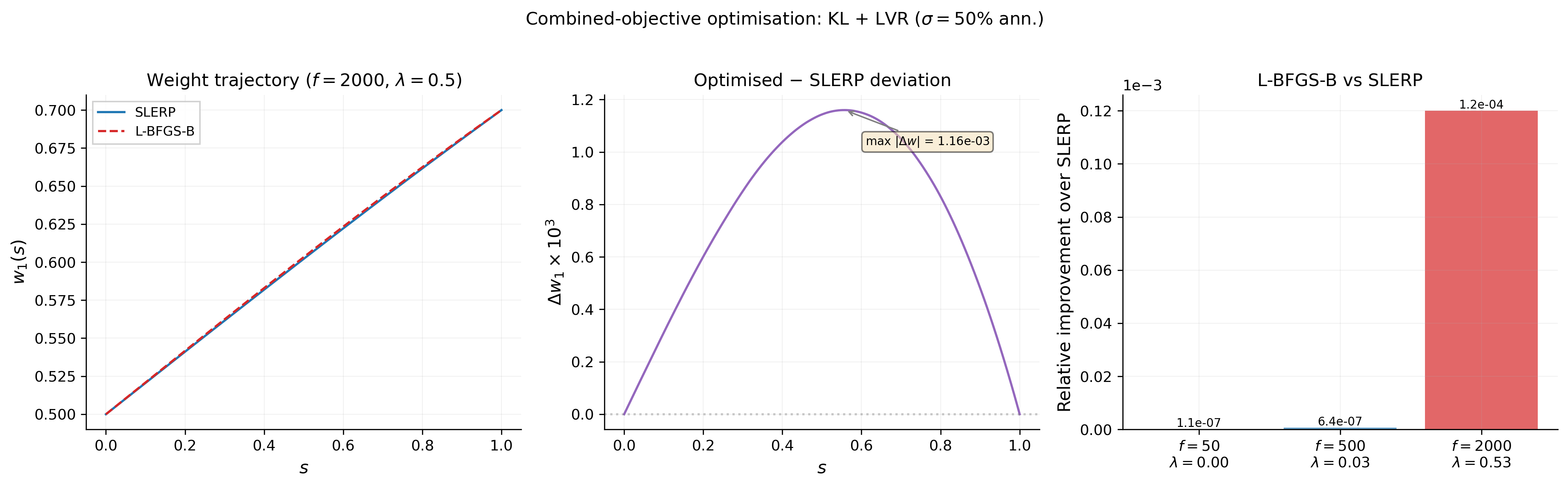}
        \caption{Combined-objective (KL\,+\,LVR) optimisation, $\sigma = 50\%$ ann. Left: SLERP and L-BFGS-B trajectories at $f = 2000$ ($\lambda = 9.0$). Centre: optimised$\,-\,$SLERP deviation (max $|\Delta w| = 1.16 \times 10^{-3}$). Right: relative improvement of L-BFGS-B over SLERP at three $(\!f, \lambda)$ operating points.}
        \label{fig:lbfgsb_stoch}
    \end{subfigure}
    \caption{Stochastic-price experiments, two-token setup $(0.5, 0.5) \to (0.7, 0.3)$, $\sigma = 50\%$ ann.}
    \label{fig:stochastic_experiments}
\end{figure*}

\begin{figure*}[htbp]
    \centering
    \includegraphics[width=\textwidth]{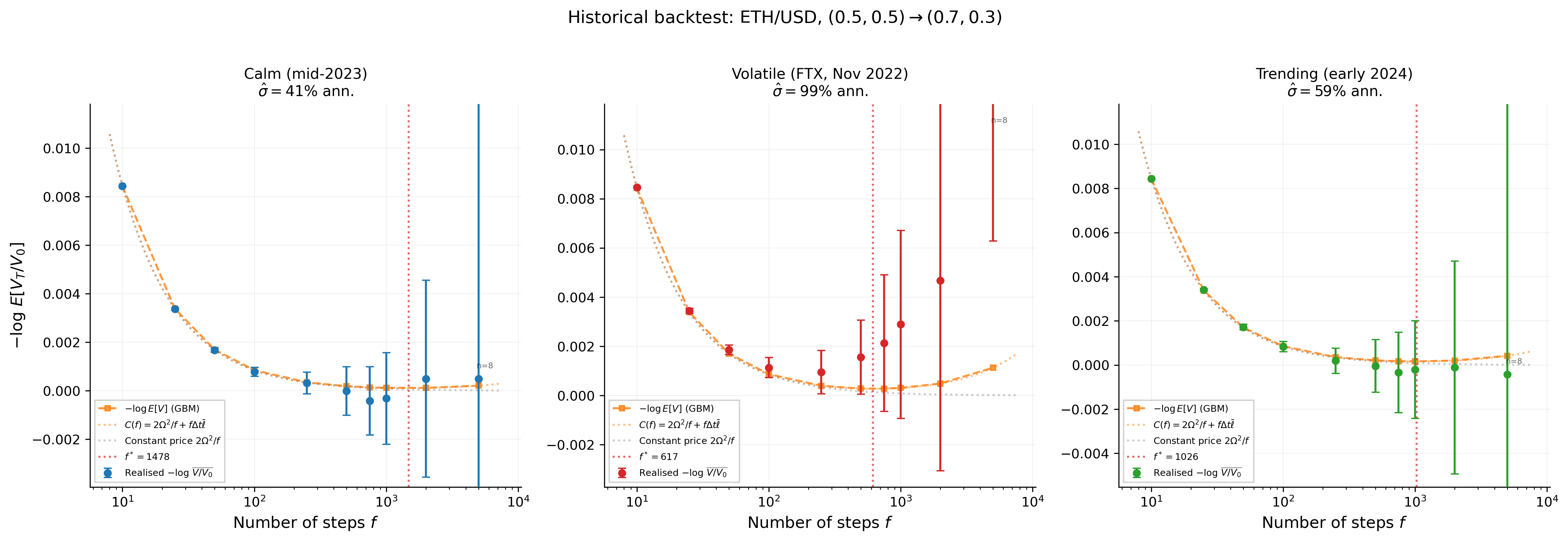}
    \caption{Historical backtest: ETH/USD, $(0.5, 0.5) \to (0.7, 0.3)$, across three market regimes (calm, volatile, trending). Dashed curves show the analytical $C(f)$; dots show realised $-\log(V_T/V_0)$ with 95\% bootstrap confidence intervals. Vertical lines mark $f^*$ (red, dashed) and the constant-price $2\Omega^2/f$ reference. The GBM prediction matches realised costs; $f^*$ sits near the observed minimum in each regime.}
    \label{fig:historical_backtest}
\end{figure*}

\paragraph{Brute-force numerical verification}

To verify that SLERP is near-optimal on the exact KL cost, we ran L-BFGS-B~\cite{byrd1995} optimisation over all $f-1$ intermediate weight vectors, initialised from the SLERP solution, for the $N=3$ setup.
At $f=1000$, the optimiser converged in 29 iterations without improving the objective beyond SLERP's value.
At $f=50$, the retained fraction improved from SLERP's $0.989\,047\,93$ to $0.989\,048\,06$, a relative gain of ${\sim}1.2 \times 10^{-5}$ (compare the $O(\Omega^2/f^2)$ bound of Theorem~\ref{thm:suboptimality}).

\paragraph{Difference between SLERP and other methods}

SLERP-minus-linear weight differences are $\sim0.03$, with the same arch shape as the optimal-minus-linear plot in~\cite{willetts2024optimalrebalancingdynamicamms}.
SLERP-minus-(AM+GM)/normalise differences are an order of magnitude smaller ($\sim0.003$), and differences between SLERP and the brute-force optimal trajectory are of order $10^{-7}$ (Figure~\ref{fig:diff}, Appendix~\ref{app:trajectories}).

\paragraph{Near-boundary behaviour}

Current G3M implementations enforce a minimum weight of $1\%$, where the Fisher--Rao metric diverges and the quadratic approximation is least accurate.
SLERP's advantage over other methods in fact \emph{widens} near the boundary: at $w_{\min} = 0.01$, linear interpolation incurs ${\sim}20\%$ more loss, and (AM+GM)/normalise ${\sim}3.5\%$ more.
For interior weights ($w_{\min} \geq 0.2$), all three methods are within $1\%$ of each other.
Full near-boundary experiments (convergence across step counts, per-step uniformity, and single-midpoint comparisons where Lambert~W outperforms SLERP by up to $10\%$) are in Appendix~\ref{app:near_boundary}.
Near the boundary, varying speed along the SLERP curve (fast through high-LVR regions, slow near vertices) reduces the total rebalancing-plus-LVR cost $C(f^*)$ by ${\sim}3\%$; see Appendix~\ref{app:jacobi}.

\subsection{Stochastic prices}
\label{sec:experiments_stochastic}

We use the two-token setup $(0.5, 0.5) \to (0.7, 0.3)$ with $\Delta t_{\mathrm{block}} = 12$\,s and annualised volatility $\sigma = 50\%$ unless stated otherwise.

\paragraph{Pendulum trajectories}
Figure~\ref{fig:pendulum_traj} shows solutions of the two-token pendulum (Corollary~\ref{cor:pendulum}) for $\mu$ from $0$ (free geodesic) to $100$ (boundary-layer regime).
At $\mu = 0$ the trajectory is SLERP; as $\mu$ increases, the LVR potential deflects the path away from equal weighting toward the vertices (cf.\ Appendix~\ref{app:asymptotics}).
At $\mu = 100$ the trajectory jumps to a vertex and remains there for the bulk of the interpolation.
At $\mu = 0.004$ (the practical regime at $f = f^*$), the KL cost dominates and the LVR contribution is nearly flat; at $\mu = 1$ the LVR cost varies significantly along the path.

\paragraph{Monte Carlo validation of $f^*$}
Figure~\ref{fig:mc_validation} compares the analytical cost $C(f)$ against Monte Carlo estimates ($10\,000$ GBM paths per point).
The predicted $f^*_{E[V]} = 2744$ coincides with the observed minimum.
At $\sigma = 30\%$, $50\%$, and $80\%$ annualised, the predicted $f^*$ values (4574, 2744, 1715) agree with the simulated optima, as the $f^* \propto 1/\sigma$ scaling of Proposition~\ref{prop:optimal_f} requires.

\paragraph{SLERP near-optimality under KL\,+\,LVR}
Figure~\ref{fig:lbfgsb_stoch} compares SLERP against L-BFGS-B~\cite{byrd1995} optimisation of the combined KL\,+\,LVR objective.
At $f = 2000$ ($\lambda = 9.0$), the optimised trajectory deviates from SLERP by at most $|\Delta w| = 1.16 \times 10^{-3}$, the order of magnitude predicted by the Green's function correction (Appendix~\ref{app:asymptotics}).
The relative cost improvement over SLERP is ${\sim}10^{-7}$ at $f = 50$ ($\lambda = 0.01$), ${\sim}6 \times 10^{-7}$ at $f = 500$ ($\lambda = 0.56$), and $1.2 \times 10^{-4}$ at $f = 2000$ ($\lambda = 9.0$).
At $f = f^*$, where $\lambda^* = O(1)$, there is nothing to gain from departing from SLERP.

\paragraph{Historical backtest}
Figure~\ref{fig:historical_backtest} replays the rebalancing on three historical ETH/USD windows with annualised volatilities ranging from $41\%$ (calm, mid-2023) to $99\%$ (FTX crash, Nov 2022).\footnote{USD used to sidestep fluctuations in stablecoin value which mean their volatility is not identically zero.}
The GBM-based $C(f)$ and $f^*$ track the realised costs well; in the volatile regime, heavier tails widen the confidence intervals but $f^*$ still sits near the observed minimum.
The trending window has nonzero drift, which the driftless-GBM model does not capture.
However, the retention ratio is price-independent (Proposition~\ref{prop:price_independence}) regardless of drift, and the cross-term telescoping (Proposition~\ref{prop:telescoping}) requires only that price increments are independent across blocks, not that they are driftless.
Drift affects the LVR rate $\bar\ell$ but not the path optimality; over typical rebalancing windows (${\sim}8$ hours at $f^*$), drift is small relative to volatility.

\section{Related Work}
\label{sec:related}

\paragraph{Dynamic weight AMMs}
Several protocols implement G3M pools with time-varying weights.
Balancer's Liquidity Bootstrapping Pools~\cite{balancer} change weights on a
fixed schedule to create a Dutch-auction-like price decline for token distribution; all deployed implementations (Balancer V1--V3, Fjord Foundry~\cite{fjordfoundry}) use linear weight interpolation.
Aera Finance V1~\cite{aera} adjusts weights via
external controllers, performing linear weight interpolation.
In each case the weight schedule shapes the sequence of arbitrage
opportunities the pool offers: an LBP is a Dutch auction designed to maximise
price discovery, while a TFMM rebalancing is an auction on the total surplus paid to arbitrageurs.
The interpolation curve governs the fine properties of the auction in both
settings, but only the TFMM objective (cost minimisation) has been studied
formally~\cite{willetts2024optimalrebalancingdynamicamms}.

\paragraph{Arbitrage cost and LVR}
The LVR framework of \citet{lvr} characterises the adverse-selection cost
that arbitrageurs extract from AMM liquidity providers, with extensions to
fee-adjusted settings~\cite{lvr_arb}.
\citet{evansG3Ms} analyses LP returns in G3Ms with time-varying weights.
These results quantify the per-block cost of holding a position but do not
optimise the schedule of weight changes within a rebalancing episode.
TWAMM~\cite{twamm} solves the analogous problem for trade
execution: a large \emph{swap} is decomposed into infinitesimal sub-orders to reduce price impact, motivated by the same convexity that drives weight
interpolation here.

\paragraph{Information geometry and portfolio theory}
\citet{pal2017exponentiallyconcavefunctionsnew} study a
geometry on the simplex induced by an L-divergence, with a generalised
Pythagorean theorem governing optimal rebalancing frequency in stochastic
portfolio theory.
Their framework and ours both arrive at Riemannian structure on the weight
simplex, but from different divergences for different reasons: L-divergence
from portfolio excess growth, KL divergence from arbitrage cost.

\section{Discussion}
\label{sec:discussion}

\paragraph{A unifying geometric perspective}
In the language of information geometry~\cite{amari_info_geo,amari_nagaoka2000}, the arbitrage cost is the canonical divergence of the dually-flat simplex (Appendix~\ref{app:info_geo}).
The Fisher--Rao metric provides a single framework that explains why the various candidate interpolation methods (linear, geometric, (AM+GM)/normalise, Lambert~W) all perform similarly.
Each corresponds to a different geodesic or approximate geodesic on the weight simplex (Appendix~\ref{app:info_geo}); they nearly agree because the simplex is nearly flat at the scale of typical per-step weight changes.
The sub-optimality bound (Theorem~\ref{thm:suboptimality}) makes this precise: for any reasonable step count, the choice between methods is practical, not a matter of loss optimality.

\paragraph{Accuracy of the quadratic approximation}
The quadratic loss kernel, Eq~\eqref{eq:tfmm_loss_kernel}, has cubic remainder in the per-step weight changes: for $f$-step interpolation the cumulative error is $O(1/f^2)$.
A $[1,1]$ Pad\'e approximant resums this remainder into a rational function that matches the exact KL cost through $O(u^4)$ and captures the cubic asymmetry between weight increases and decreases (Appendix~\ref{app:pade}).
The Lambert~W approach of~\cite{willetts2024optimalrebalancingdynamicamms} optimises each weight component independently and renormalises.
It maximises the exact finite-step retention ratio rather than a quadratic surrogate, which makes it more accurate near the simplex boundary where the quadratic approximation is least tight (Appendix~\ref{app:near_boundary}).

\paragraph{Stochastic prices and practical deployment}
The price-independence of the retention ratio (Proposition~\ref{prop:price_independence}) means that a TFMM deployer can choose the interpolation curve without forecasting prices.
The only new parameter is $f^*$, which requires an estimate of average LVR along the path; for a pool with known asset volatilities, this reduces to evaluating a single integral over the SLERP curve (Remark~\ref{rem:circularity}).
On chains where fee income exceeds LVR extraction ($\phi < 0$), even $f^*$ is unnecessary: the constant-price analysis applies and more steps is always preferable.
If the target weights change mid-interpolation (e.g.\ a new strategy signal arrives before the current interpolation completes), the pool can restart SLERP from its current weights to the new target; since any sub-arc of a geodesic is itself a geodesic, the already-completed portion of the interpolation was optimal for the distance covered.

\paragraph{Fees and market microstructure}
Swap fees and gas costs qualitatively change the tradeoff (\S\ref{sec:fees}).
Fee revenue from arb trades telescopes along monotonic paths (Proposition~\ref{prop:fee_telescoping}), creating a path-independent offset that does not change the optimal interpolation curve (Corollary~\ref{cor:slerp_fees}).
The fee band caps the useful step count at $f_{\mathrm{threshold}}$ (Eq~\eqref{eq:f_threshold}), and the fee-adjusted optimum (Proposition~\ref{prop:fee_adjusted_f}) incorporates the empirical cost-of-carry $\phi$.
On L2, where low gas costs enable high-frequency multi-hop routing at standard fees (e.g.\ 30\,bp), $\phi$ can be negative, eliminating the finite-$f$ tradeoff entirely.
The per-block weight change caps of~\cite{willetts2024multiblockmevopportunities} bound the arbitrage opportunity available to an attacker with last-look/first-look across consecutive blocks; larger $f$ for the same total weight change means smaller per-block changes, which makes this attack harder.

\paragraph{Implementation considerations}
Computing SLERP requires one $\arccos$ evaluation, $N$ square roots, and per-step $\sin$ calls, comparable in cost to geometric mean interpolation.
On-chain, where transcendental functions are expensive, (AM+GM)/normalise~\cite{willetts2024optimalrebalancingdynamicamms} is still the pragmatic default.
The recursive bisection algorithm (Corollary~\ref{cor:bisection}) shows that this heuristic computes exact geodesic points at power-of-two step counts using only elementary arithmetic.
No trigonometric functions are needed.

\section{Concluding remarks}

The arbitrage cost of rebalancing a G3M pool is governed by the Fisher--Rao metric, and the optimal interpolation is its geodesic (SLERP), computable on-chain without trigonometric functions via recursive bisection.
This structure survives stochastic prices, swap fees, and gas costs.

The geometric framework generalises beyond G3M pools.
An AMM whose trading function is parameterised by a vector will typically have its arbitrage cost induce a Riemannian metric on its parameter space, and the optimal rebalancing trajectory will be a geodesic of that metric.
For G3M, the metric is Fisher--Rao because the cost is a KL divergence; other designs (e.g. virtual balance interpolation in concentrated liquidity designs) would have different metrics and different geodesics.
Identifying these metrics and their geodesics could yield analogous rebalancing strategies for other designs.
\FloatBarrier

\newpage
\bibliography{biblio}

\newpage
\onecolumn
\begin{appendices}
\setcounter{figure}{0}
\setcounter{equation}{0}
\setcounter{table}{0}
\renewcommand\thefigure{\thesection.\arabic{figure}}
\renewcommand\theequation{\thesection.\arabic{equation}}
\renewcommand\thetable{\thesection.\arabic{table}}

\section{Retention Ratio}

\subsection{Deferred Proof: Proposition~\ref{prop:retention} [Retention ratio]}
\label{app:retention}

\retention*

\begin{proof}
Start with Eq~\eqref{eq:reserve_change},
\begin{equation}
    R_i^{\mathrm{end}} = R_i^{\mathrm{start}} \frac{w_i^{\mathrm{end}}}{w_i^{\mathrm{start}}} \prod_{j=1}^N \left(\frac{w_j^{\mathrm{start}}}{w_j^{\mathrm{end}}}\right)^{w_j^{\mathrm{end}}},
\end{equation}
which we can rearrange to get
\begin{equation}
    \frac{R_i^{\mathrm{end}}}{R_i^{\mathrm{start}}} \frac{w_i^{\mathrm{start}}}{w_i^{\mathrm{end}}} = \prod_{j=1}^N \left(\frac{w_j^{\mathrm{start}}}{w_j^{\mathrm{end}}}\right)^{w_j^{\mathrm{end}}},
    \label{eq:value_ratio}
\end{equation}
and we thus define $r:= \prod_{j=1}^N \left(\frac{w_j^{\mathrm{start}}}{w_j^{\mathrm{end}}}\right)^{w_j^{\mathrm{end}}}$.
For given reserves and market prices $\v m_p$, pool value is $V=\sum_{i=1}^N R_i m_{p,i}$.
By the action of arbitrageurs, the pool holds the minimum value possible under its trading function. The weight vector then gives the division of value between assets~\cite{tfmm_litepaper}: the G3M marginal price of token $i$ is $w_i k / R_i$~\cite{evansG3Ms}, which at equilibrium equals $p_i$, giving $R_i p_i = w_i V$ for all $i$. So we have
\begin{align}
    V^{\mathrm{start}} w^{\mathrm{start}}_i &= R^{\mathrm{start}}_i m_{p,i} \\
    V^{\mathrm{end}} w^{\mathrm{end}}_i &= R^{\mathrm{end}}_i m_{p,i}
\end{align}
Dividing these componentwise, rearranging, and substituting Eq~\eqref{eq:value_ratio} gives
\begin{align}
    \frac{V^{\mathrm{end}}}{V^{\mathrm{start}}} \frac{w^{\mathrm{end}}_i}{w^{\mathrm{start}}_i} &= \frac{R^{\mathrm{end}}_i}{R^{\mathrm{start}}_i}\\
    \Rightarrow \frac{V^{\mathrm{end}}}{V^{\mathrm{start}}} &= \frac{R_i^{\mathrm{end}}}{R_i^{\mathrm{start}}} \frac{w_i^{\mathrm{start}}}{w_i^{\mathrm{end}}}\\
    \Rightarrow \frac{V^{\mathrm{end}}}{V^{\mathrm{start}}} &=\prod_{j=1}^N \left(\frac{w_j^{\mathrm{start}}}{w_j^{\mathrm{end}}}\right)^{w_j^{\mathrm{end}}} \\
     \Rightarrow \frac{V^{\mathrm{end}}}{V^{\mathrm{start}}} &=r 
\end{align}
as required.
\end{proof}

\subsection{Deferred Proof: Corollary~\ref{cor:quadratic_loss} [Quadratic loss kernel]}
\label{app:taylor}

\quadraticloss*

\begin{proof}
We derive Eq~\eqref{eq:tfmm_loss_kernel} in detail.
Starting from the retention ratio, Eq~\eqref{eq:retention}:
\begin{equation}
    \log r = \sum_{i=1}^N w_i^{\mathrm{end}} \log\!\left(\frac{w_i^{\mathrm{start}}}{w_i^{\mathrm{end}}}\right).
\end{equation}
Writing $w_i^{\mathrm{end}} = w_i + \Delta w_i$ (where $w_i \equiv w_i^{\mathrm{start}}$ for brevity):
\begin{align}
    \log r &= \sum_{i=1}^N (w_i + \Delta w_i)\log\!\left(\frac{w_i}{w_i + \Delta w_i}\right) \nonumber\\
    &= -\sum_{i=1}^N (w_i + \Delta w_i)\log\!\left(1 + \frac{\Delta w_i}{w_i}\right).
\end{align}
Expanding $\log(1+x) = x - x^2/2 + O(x^3)$ with $x = \Delta w_i / w_i$:
\begin{align}
    \log r &= -\sum_{i=1}^N (w_i + \Delta w_i)\left(\frac{\Delta w_i}{w_i} - \frac{(\Delta w_i)^2}{2w_i^2} + O\!\left(\frac{(\Delta w_i)^3}{w_i^3}\right)\right) \nonumber\\
    &= -\sum_{i=1}^N \left(\Delta w_i + \frac{(\Delta w_i)^2}{w_i} - \frac{(\Delta w_i)^2}{2w_i} - \frac{(\Delta w_i)^3}{2w_i^2} + \cdots\right) \nonumber\\
    &= -\underbrace{\sum_{i=1}^N \Delta w_i}_{=\,0} - \sum_{i=1}^N \frac{(\Delta w_i)^2}{2w_i} + O(\Delta w^3).
\end{align}
The linear term vanishes because $\sum_i \Delta w_i = 0$ (weights sum to 1 at both endpoints), leaving
\begin{equation}
    -\log r \;\approx\; \sum_{i=1}^N \frac{(\Delta w_i)^2}{2w_i},
\end{equation}
as required.
\end{proof}

\section{Deferred proof: Constant metric speed minimises total loss (used in Corollary~\ref{cor:slerp_optimal})}
\label{app:constant_speed_proof}

\begin{quote}
\textbf{Claim.} Let $s_0 < s_1 < \cdots < s_f$ be a partition of a path of total arc-length $S = s_f - s_0$.
The sum $\sum_{k=1}^f (\Delta s_k)^2$, where $\Delta s_k = s_k - s_{k-1}$, is minimised when all increments are equal: $\Delta s_k = S/f$ for all $k$.
\end{quote}

\begin{proof}
By the QM--AM inequality:
\[
    \frac{1}{f}\sum_{k=1}^f (\Delta s_k)^2 \geq \left(\frac{1}{f}\sum_{k=1}^f \Delta s_k\right)^{\!2} = \frac{S^2}{f^2}.
\]
Hence $\sum_k (\Delta s_k)^2 \geq S^2/f$, with equality iff all $\Delta s_k$ are equal.
\end{proof}

This result, combined with the identification of per-step loss as $(\mathrm{d}s_k)^2$, establishes that the loss-minimising parameterisation of any path is constant metric speed.
The further optimisation over the choice of path then selects the geodesic, since among all paths with a given total arc-length $S$, the one with smallest $S$ gives the smallest $S^2/f$.

\section{Deferred proof: Theorem~\ref{thm:suboptimality} [sub-optimality bound]}
\label{app:suboptimality_proof}

\suboptimalitybound*

\begin{proof}
Write the total exact cost as $C = Q + R$, where
\[
    Q = \sum_{k=1}^f \sum_i \frac{(\Delta w_i^{(k)})^2}{2\,w_i^{(k-1)}}, \qquad R = \sum_{k=1}^f R_k,
\]
and $R_k = \sum_i w_i^{(k-1)}\, \tilde{h}(u_i^{(k)})$ with $u_i^{(k)} = \Delta w_i^{(k)}/w_i^{(k-1)}$ and $\tilde{h}(u) = (1{+}u)\log(1{+}u) - u - u^2/2$ (the KL remainder beyond the quadratic; note this differs from the full per-component KL function $h(u) = (1{+}u)\log(1{+}u) - u$ used in Appendix~\ref{app:pade}).
SLERP minimises $Q$, so $\nabla Q\big|_{\mathrm{SLERP}} = 0$ and $\nabla C\big|_{\mathrm{SLERP}} = \nabla R\big|_{\mathrm{SLERP}}$.

\paragraph{Step 1: per-step change bound}
Along the SLERP path the per-step angular displacement on the Hellinger sphere is $\Omega/f$.
The weight change satisfies $\Delta w_i = (\eta_i^{(k)} + \eta_i^{(k-1)})(\eta_i^{(k)} - \eta_i^{(k-1)})$, so $|\Delta w_i| \leq 2|\Delta\eta_i|$ (since $|\eta_i| \leq 1$) and
\[
    |u_i^{(k)}| = \frac{|\Delta w_i^{(k)}|}{w_i^{(k-1)}} \leq \frac{2|\Delta\eta_i^{(k)}|}{\epsilon} \leq \frac{2\Omega}{f\epsilon},
\]
where the last inequality uses $|\Delta\eta_i| \leq |\boldsymbol{\Delta\eta}| \leq \Omega/f$ and $w_i \geq \epsilon$.

\paragraph{Step 2: gradient of the remainder}
Each component of $\nabla R$ with respect to $w_i^{(k)}$ receives contributions from steps $k$ and $k{+}1$.
Since $\tilde{h}'(u) = \log(1{+}u) - u$ satisfies $|\tilde{h}'(u)| \leq u^2/(1 - |u|)$ for $|u| < 1$, and $|u_i^{(k)}| \leq 2\Omega/(f\epsilon) \leq 1/2$ for $f \geq 4\Omega/\epsilon$, each gradient component has magnitude
\[
    \left|\frac{\partial R}{\partial w_i^{(k)}}\right| \leq 2 \cdot \frac{(2\Omega/(f\epsilon))^2}{1/2} = \frac{16\,\Omega^2}{f^2 \epsilon^2}.
\]
There are $(f{-}1)$ intermediate weight vectors, each with $N{-}1$ degrees of freedom on the simplex, giving
\[
    \|\nabla R\|^2 \leq (f{-}1)(N{-}1)\left(\frac{16\,\Omega^2}{f^2\epsilon^2}\right)^{\!2} \leq \frac{256\,N\,\Omega^4}{f^3\,\epsilon^4}.
\]

\paragraph{Step 3: Hessian lower bound}
The Hessian of $C$ restricted to the simplex constraints is block-diagonal in the intermediate weight vectors.
Each block receives contributions from two consecutive KL terms.
For the first, $D_{\mathrm{KL}}(\v w^{(k)} \| \v w^{(k-1)})$, differentiating the first-argument contribution $w_i^{(k)}\log(w_i^{(k)}/w_i^{(k-1)})$ twice gives $\partial^2/\partial(w_i^{(k)})^2 = 1/w_i^{(k)} \geq 1$ (since $w_i \leq 1$).
For the second, $D_{\mathrm{KL}}(\v w^{(k+1)} \| \v w^{(k)})$, the second-argument contribution $-w_i^{(k+1)}\log w_i^{(k)}$ gives $\partial^2/\partial(w_i^{(k)})^2 = w_i^{(k+1)}/(w_i^{(k)})^2 \geq 1/2$ (using $w_i^{(k+1)}/w_i^{(k)} = 1 + u_i \geq 1/2$ for $|u| \leq 1/2$, and $w_i^{(k)} \leq 1$).
On the simplex hyperplane, each block therefore has $\lambda_{\min} \geq 1 + 1/2 = 3/2$.

\paragraph{Step 4: quadratic bound}
Since SLERP is a critical point of $Q$ and the Hessian of $C$ at SLERP is positive definite, the standard bound on the cost improvement available from a gradient step gives
\[
    C_{\mathrm{SLERP}} - C_* \leq \frac{\|\nabla C\|^2}{2\,\lambda_{\min}(\nabla^2 C)} \leq \frac{1}{3}\cdot\frac{256\,N\,\Omega^4}{f^3\,\epsilon^4} = \frac{256\,N}{3\,f^3\,\epsilon^4}\,\Omega^4,
\]
so $A = 256N/(3\epsilon^4) \leq 86\,N/\epsilon^4$ suffices (using $\lambda_{\min} \geq 3/2$).
The relative sub-optimality is $A\,\Omega^4/f^3$ divided by the total SLERP cost $2\Omega^2/f$, giving $O(\Omega^2/f^2)$.
\end{proof}

\section{Deferred proof: modified geodesic and two-token pendulum}
\label{app:modified_geodesic_proof}

\modifiedgeodesic*
\begin{proof}
The cost functional has Lagrangian
$L = \tfrac{1}{2f}g_{ij}\dot{w}^i\dot{w}^j + T\ell(\v w)$,
where $g_{ij} = \delta_{ij}/w_i$ is the Fisher--Rao metric.
The Euler--Lagrange equation gives:
\[
\frac{1}{f}\frac{d}{ds}\!\left(\frac{\dot{w}^i}{w_i}\right) = T\frac{\partial \ell}{\partial w^i} + \frac{1}{2f}\frac{(\dot{w}^i)^2}{w_i^2}
\]
Expanding the left-hand side:
\[
\frac{1}{f}\!\left(\frac{\ddot{w}^i}{w_i} - \frac{(\dot{w}^i)^2}{w_i^2}\right) = T\frac{\partial \ell}{\partial w^i} + \frac{1}{2f}\frac{(\dot{w}^i)^2}{w_i^2}
\]
Multiplying through by $w_i$ and recognising that $\Gamma^i_{jk} = -\delta_{ij}\delta_{ik}/(2w_i)$ for the Fisher--Rao metric on the simplex, so that $\Gamma^i_{jk}\dot{w}^j\dot{w}^k = -(\dot{w}^i)^2/(2w_i)$, the kinetic terms combine to give
\[
\frac{1}{f}\!\left(\ddot{w}^i + \Gamma^i_{jk}\dot{w}^j\dot{w}^k\right) = T\,w_i\frac{\partial \ell}{\partial w^i} = T\,g^{ij}\frac{\partial \ell}{\partial w^j}
\]
where $g^{ij} = w_i\,\delta_{ij}$ is the inverse Fisher--Rao metric.
\end{proof}

\pendulum*
\begin{proof}
For $N=2$ with weights $(w, 1-w)$, $\sigma_1 = \sigma$, $\sigma_2 = 0$, the Hellinger coordinate $\theta = \arcsin(\sqrt{w})$ gives $w = \sin^2\theta$ and $1-w = \cos^2\theta$.
The Fisher--Rao kinetic term becomes $\frac{1}{2f}[(\dot{w})^2/w + (\dot{w})^2/(1-w)] = \frac{2}{f}\dot\theta^2$ (using $\dot{w} = 2\sin\theta\cos\theta\,\dot\theta = \sin(2\theta)\dot\theta$).
The LVR potential is $\ell = \frac{\sigma^2}{2}w(1-w) = \frac{\sigma^2}{8}\sin^2(2\theta)$.
The Lagrangian is $L = \frac{2}{f}\dot\theta^2 + T\frac{\sigma^2}{8}\sin^2(2\theta)$, and the Euler--Lagrange equation gives $\frac{4}{f}\ddot\theta = T\frac{\sigma^2}{4}\sin(4\theta)$, which simplifies to Eq~\eqref{eq:pendulum}.
\end{proof}

\section{Taylor agreement of interpolation methods}
\label{app:taylor_agreement}

Write $w_i \equiv w_i^{\mathrm{start}}$ and $u_i = \Delta w_i / w_i$ where $\Delta w_i = w_i^{\mathrm{end}} - w_i^{\mathrm{start}}$.
The simplex constraint gives $\sum_i w_i u_i = 0$.
Define $s^2 = \sum_j w_j u_j^2$.
We expand all three methods to second order in $u$.

\paragraph{SLERP}
Under the Hellinger embedding, $\sqrt{w_i^{\mathrm{end}}} = \sqrt{w_i}\sqrt{1+u_i} = \sqrt{w_i}(1 + u_i/2 - u_i^2/8 + O(u^3))$.
The Bhattacharyya coefficient is
\[
\cos\Omega = \sum_j w_j \sqrt{1+u_j} = 1 + \tfrac{1}{2}\underbrace{\sum_j w_j u_j}_{=\,0} - \tfrac{1}{8}s^2 + O(u^3),
\]
so $\Omega^2 = 2(1-\cos\Omega) + O(u^4) = s^2/4 + O(u^3)$.
The SLERP coefficients satisfy, for small $\Omega$,
\[
\frac{\sin(x\Omega)}{\sin\Omega} = x + \frac{x(1-x^2)}{6}\Omega^2 + O(\Omega^4).
\]
Since $\Omega^2 = O(u^2)$, any product of an $\Omega^2$ correction with a term already $O(u)$ is $O(u^3)$, so
\[
\eta_i(t) = \sqrt{w_i}\!\left[(\alpha + \beta) + \frac{tu_i}{2} - \frac{tu_i^2}{8}\right] + O(u^3),
\]
where $\alpha = \sin((1\!-\!t)\Omega)/\sin\Omega$, $\beta = \sin(t\Omega)/\sin\Omega$.
Using $(1\!-\!t)^3 + t^3 = 1 - 3t(1\!-\!t)$, the sum evaluates to $\alpha + \beta = 1 + \frac{t(1-t)}{2}\Omega^2 + O(u^4)$.\footnote{This comes from expanding $\sin((1-t)\Omega)/\sin\Omega + \sin(t\Omega)/\sin\Omega$ to order $\Omega^2$.}
Squaring and substituting $\Omega^2 = s^2/4$:
\begin{equation}
w_i^{\mathrm{SLERP}}(t) = w_i + t\,\Delta w_i + \frac{t(1-t)}{4}\,w_i\!\left(s^2 - u_i^2\right) + O(u^3).
\label{eq:slerp_expansion}
\end{equation}

\paragraph{(AM+GM)/normalise}
The arithmetic component (Eq~\eqref{eq:approx_optimal_traj_am}) is $w_i^{\mathrm{AM}}(t) = w_i(1 + tu_i)$.
The geometric component (Eq~\eqref{eq:approx_optimal_traj_gm}) is $w_i^{\mathrm{GM}}(t) = w_i(1+u_i)^t = w_i(1 + tu_i + \frac{t(t-1)}{2}u_i^2 + O(u^3))$.
Their unnormalised sum is
\[
w_i^{\mathrm{AM}} + w_i^{\mathrm{GM}} = w_i\!\left[2 + 2tu_i + \tfrac{t(t-1)}{2}u_i^2\right] + O(u^3).
\]
The normalisation factor, using $\sum_j w_j u_j = 0$, is
\[
Z = \sum_j (w_j^{\mathrm{AM}} + w_j^{\mathrm{GM}}) = 2 + \tfrac{t(t-1)}{2}s^2 + O(u^3).
\]
Dividing and expanding $1/Z$:
\begin{equation}
\breve{w}_i(t) = w_i + t\,\Delta w_i + \frac{t(1-t)}{4}\,w_i\!\left(s^2 - u_i^2\right) + O(u^3),
\label{eq:amgm_expansion}
\end{equation}
which matches Eq~\eqref{eq:slerp_expansion} exactly through $O(u^2)$.

\paragraph{Lambert~W at the midpoint}
The optimal 2-step midpoint of~\cite{willetts2024optimalrebalancingdynamicamms} is $\tilde{w}_i^* = w_i(1+u_i)/W_0(e(1+u_i))$.
Expanding $W_0(e(1+u))$ around $u=0$ (where $W_0(e)=1$) by implicit differentiation of $W e^W = e(1+u)$:
\[
W_0(e(1+u)) = 1 + \tfrac{u}{2} - \tfrac{3u^2}{16} + O(u^3).
\]
Then $\tilde{w}_i^* = w_i(1+u_i)(1 + u_i/2 - 3u_i^2/16)^{-1} = w_i(1 + u_i/2 - u_i^2/16 + O(u^3))$.
This requires renormalisation; the normalisation factor is $\sum_j \tilde{w}_j^* = 1 - s^2/16 + O(u^3)$, giving
\begin{equation}
\hat{w}_i^{\mathrm{LW}} = w_i + \tfrac{1}{2}\Delta w_i + \tfrac{1}{16}\,w_i(s^2 - u_i^2) + O(u^3).
\label{eq:lambert_expansion}
\end{equation}
Setting $t = 1/2$ in Eq~\eqref{eq:slerp_expansion} gives the same second-order term $\tfrac{(1/2)(1/2)}{4} = \tfrac{1}{16}$, confirming agreement through~$O(u^2)$.

\paragraph{The $O(u^3)$ difference}
At general $t \neq 1/2$, the third-order terms of SLERP and (AM+GM)/normalise differ (the former involves $\Omega^4$ corrections from the SLERP coefficients, the latter involves cubic terms from $(1+u)^t$).
At $t = 1/2$, these methods agree exactly to all orders by Theorem~\ref{thm:slerp_amgm}.
The Lambert~W midpoint first differs from both at $O(u^3)$, reflecting the fact that it optimises the exact retention ratio while the other two optimise the quadratic surrogate.

\section{Why (AM+GM)/normalise works: the $\alpha$-family of geodesics}
\label{app:info_geo}

The midpoint equivalence (Theorem~\ref{thm:slerp_amgm}) follows from the binomial identity, but it has a geometric interpretation in terms of a family of ``straight lines'' on the probability simplex~\cite{amari_info_geo}.

Information geometry parameterises a family of connections on the simplex by $\alpha \in [-1, 1]$.
Each value of $\alpha$ gives a different notion of geodesic:

\begin{itemize}
    \item \textbf{$\alpha = -1$ (arithmetic/linear interpolation).}
    Geodesics are straight lines in the ordinary simplex coordinates $w_i$.
    The midpoint is $w_i^{\mathrm{mid}} = (w_i^{\mathrm{start}} + w_i^{\mathrm{end}})/2$.
    This is the interpolation used by liquidity bootstrapping pools and by TFMM when weights change linearly~\cite{tfmm_litepaper}.

    \item \textbf{$\alpha = +1$ (geometric interpolation).}
    Geodesics are straight lines in the log-coordinates $\theta_i = \log w_i$.
    The midpoint is $w_i^{\mathrm{mid}} \propto \sqrt{w_i^{\mathrm{start}}\, w_i^{\mathrm{end}}}$ (after renormalisation).

    \item \textbf{$\alpha = 0$ (Fisher--Rao metric geodesic).}
    Geodesics are great circles on the sphere under the Hellinger embedding.
    The midpoint is the SLERP midpoint.
\end{itemize}

AM and GM are the midpoints of the two extreme geodesics in this family.
Why does their \emph{sum} land on the middle geodesic?

The Hellinger embedding forces this, and it is not a general principle about geodesic families.
The $\alpha=0$ midpoint is computed by averaging the Hellinger coordinates $\sqrt{w_i}$ and squaring back.
Squaring a sum of square roots produces the binomial identity:
\begin{equation}
    \left(\sqrt{a} + \sqrt{b}\right)^2 = \underbrace{(a + b)}_{\text{AM term}} + \underbrace{2\sqrt{ab}}_{\text{GM term}}
\end{equation}
The sum appears because squaring distributes as addition.
Had the embedding involved a different power (say, cube roots), the decomposition would be different and AM\,+\,GM would not appear.
The square root is forced by the Fisher--Rao metric: $w_i \mapsto w_i^{1/2}$ is the unique power that makes the metric Euclidean (up to scale), and the binomial identity is a consequence of that specific power.

The $\alpha=-1$ (arithmetic) midpoint automatically satisfies $\sum_i w_i^{\mathrm{mid}} = 1$, while the $\alpha=+1$ (geometric) midpoint does not and requires renormalisation.
The (AM+GM)/normalise formula inherits this need for renormalisation from its geometric-mean component.

For general $t \neq 1/2$, the three $\alpha$-geodesics diverge.
The (AM+GM)/\allowbreak{}normalise trajectory averages the $\alpha = \pm 1$ trajectories pointwise and no longer lies exactly on the $\alpha = 0$ geodesic.
The third-order agreement reported in Table~\ref{tab:agreement} reflects the smooth dependence of geodesics on~$\alpha$: near any base point, geodesics of the different connections share the same tangent vector and second-order curvature, differing only at third order.

\subsection{Differential-geometric context}
\label{app:formal_connections}

For readers familiar with differential geometry, the three $\alpha$-values above correspond to specific affine connections on the simplex manifold.

The $\alpha=0$ connection is the \emph{Levi-Civita connection} of the Fisher--Rao metric~\cite{hobson2006}: the unique torsion-free, metric-compatible connection.
Its geodesics solve the standard Euler--Lagrange equation for arc-length minimisation, which is how the SLERP optimality result (Corollary~\ref{cor:slerp_optimal}) is derived, without ever computing Christoffel symbols.
The Hellinger embedding sidesteps the Christoffel computation entirely: it maps the simplex isometrically onto a sphere, where the geodesics are known to be great circles.
(This is analogous to solving geodesic problems by calculus of variations in a convenient coordinate system, rather than by the geodesic equation $\ddot{x}^\mu + \Gamma^\mu_{\nu\rho}\dot{x}^\nu\dot{x}^\rho = 0$ directly.)

The $\alpha=\pm 1$ connections are the \emph{mixture} ($m$) and \emph{exponential} ($e$) connections.
These are torsion-free but not metric-compatible, and they are \emph{dually flat}: the simplex has zero curvature in the $\alpha=-1$ coordinates ($w_i$) and separately in the $\alpha=+1$ coordinates ($\log w_i$), but with respect to different connections.

This dual flatness gives rise to a \emph{canonical divergence}: the unique divergence function compatible with the dually-flat structure~\cite{eguchi1983}.
On the probability simplex, the canonical divergence is the KL divergence~\cite{amari_info_geo}.
The fact that arbitrage cost equals the KL divergence (Theorem~\ref{thm:kl_divergence}) means rebalancing cost is the canonical divergence of the simplex.

The three interpolation methods in the main text correspond to the three canonical geodesics: linear ($\alpha=-1$, flat in $w_i$), geometric ($\alpha=+1$, flat in $\log w_i$), and SLERP ($\alpha=0$, the metric geodesic that sits between them).

\newpage
\section{Explicit computation: SLERP vs Lambert~W for $N=2$}
\label{app:slerp_vs_lambert}

For $N=2$ with weights $(w, 1-w)$, the SLERP midpoint ($f=2$) is computed as follows.

Map to the sphere:
\begin{align}
    \v\eta^{\mathrm{start}} &= \bigl(\sqrt{w^{\mathrm{start}}},\; \sqrt{1-w^{\mathrm{start}}}\bigr), \quad
    \v\eta^{\mathrm{end}} = \bigl(\sqrt{w^{\mathrm{end}}},\; \sqrt{1-w^{\mathrm{end}}}\bigr). \nonumber\\
    \Omega &= \arccos\!\left(\sqrt{w^{\mathrm{start}} w^{\mathrm{end}}} + \sqrt{(1-w^{\mathrm{start}})(1-w^{\mathrm{end}})}\right).
\end{align}
The midpoint on the sphere is
\begin{equation}
    \v\eta_{\mathrm{mid}} = \frac{1}{2\cos(\Omega/2)}\left(\v\eta^{\mathrm{start}} + \v\eta^{\mathrm{end}}\right).
\end{equation}
The midpoint weight is $w_{\mathrm{mid}} = \eta_{\mathrm{mid},1}^2$.

For $w^{\mathrm{start}} = 0.5$, $w^{\mathrm{end}} = 0.9$:
\begin{align}
    \Omega &= \arccos\!\left(\sqrt{0.45} + \sqrt{0.05}\right) \approx 0.4636\,\text{rad}, \nonumber\\
    w_{\mathrm{mid}}^{\mathrm{SLERP}} &\approx 0.729.
\end{align}
Compare with $w_{\mathrm{mid}}^{\mathrm{Lambert}} \approx 0.717$ from~\cite{willetts2024optimalrebalancingdynamicamms}.
The ${\sim}1.7\%$ difference reflects higher-order terms that matter at this (large) step size.
For $f \geq 10$, per-step changes are $\leq 0.04$ and the two methods agree to $< 0.1\%$.

\section{Weight trajectories and block-to-block changes}
\label{app:trajectories}

\begin{figure}[h]
    \centering
    \includegraphics[width=\textwidth]{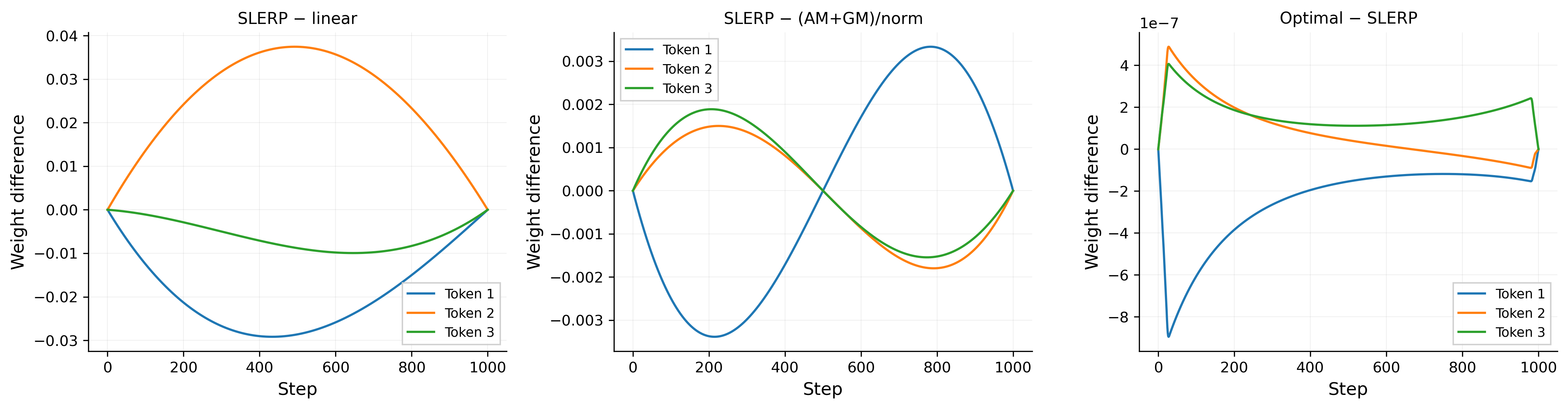}
    \caption{Weight differences between interpolation methods. Left: SLERP minus linear ($\sim0.03$), reflecting the non-linearity of the geodesic. Centre: SLERP minus (AM+GM)/normalise ($\sim0.003$), an order of magnitude smaller, with an S-curve vanishing at both endpoints and the midpoint, consistent with Theorem~\ref{thm:slerp_amgm}. Right: brute-force optimal minus SLERP ($\sim10^{-7}$), confirming SLERP is near-optimal.}
    \label{fig:diff}
\end{figure}

The figures below reproduce the $N{=}3$ setup of Figures~1--2 of~\cite{willetts2024optimalrebalancingdynamicamms}, with the addition of SLERP, to allow direct visual comparison.

\begin{figure*}[t]
\begin{minipage}{0.45\textwidth}
    \caption{Weight interpolations, $N=3$. At this scale, (b) and (c) are visually indistinguishable.}
    \label{fig:traj}
\centering
    \begin{subfigure}{\textwidth}
         \centering
        \includegraphics[width=\textwidth]{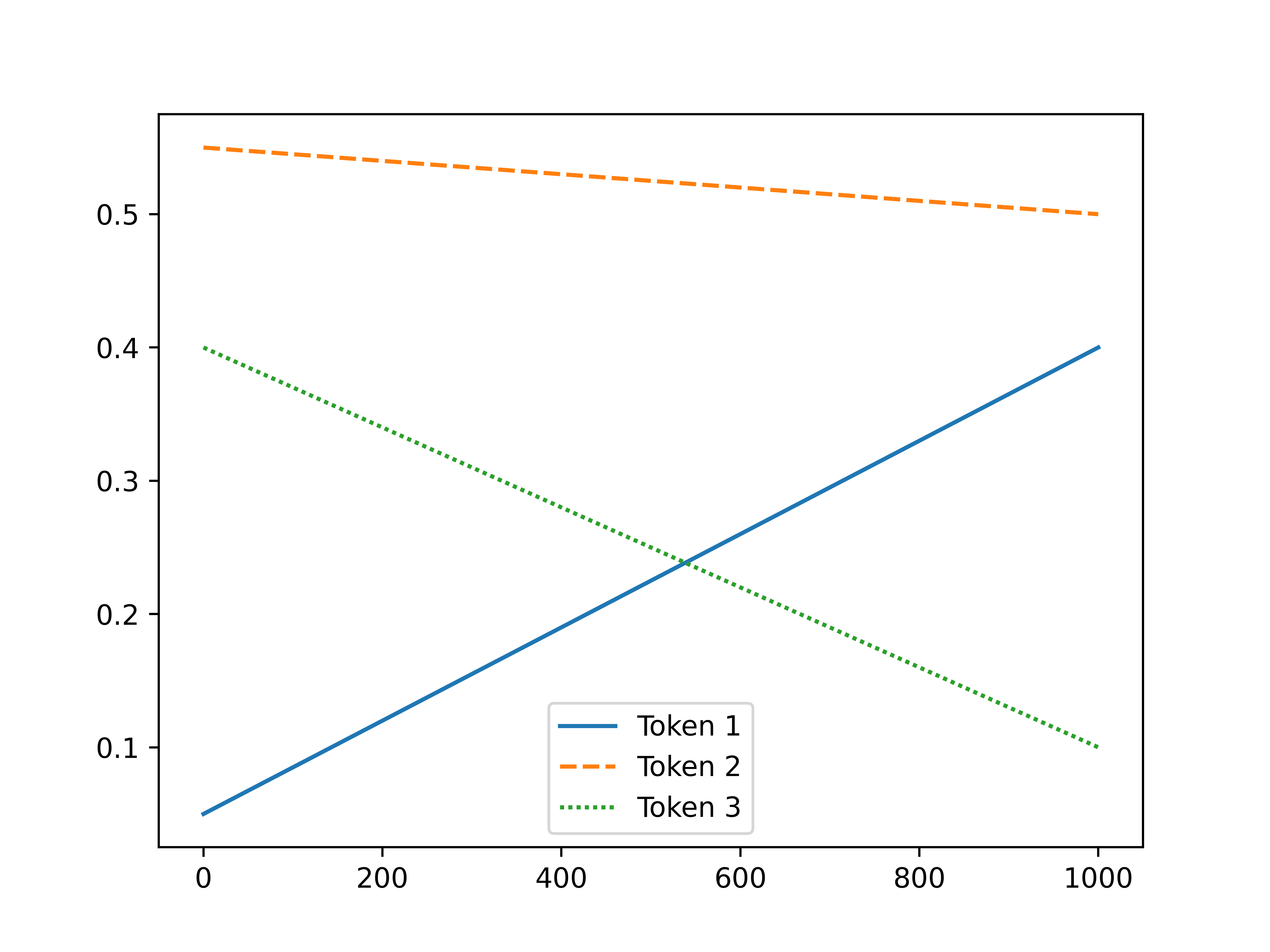}
        \caption{Linear}
        \label{fig:linear}
    \end{subfigure}
    \begin{subfigure}{\textwidth}
        \centering
        \includegraphics[width=\textwidth]{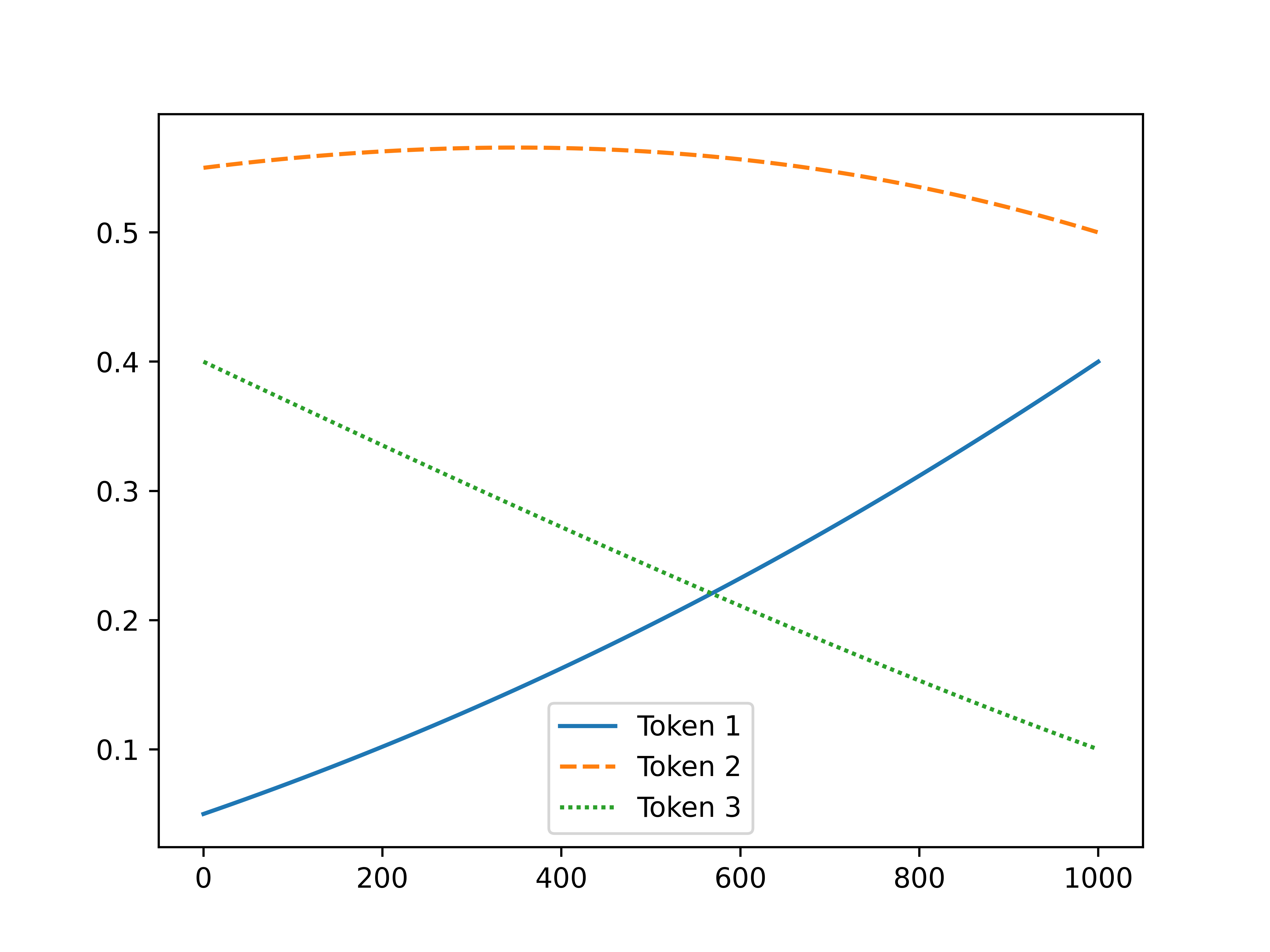}
        \caption{(AM+GM)/normalise}
        \label{fig:approx}
    \end{subfigure}
    \begin{subfigure}{\textwidth}
        \centering
        \includegraphics[width=\textwidth]{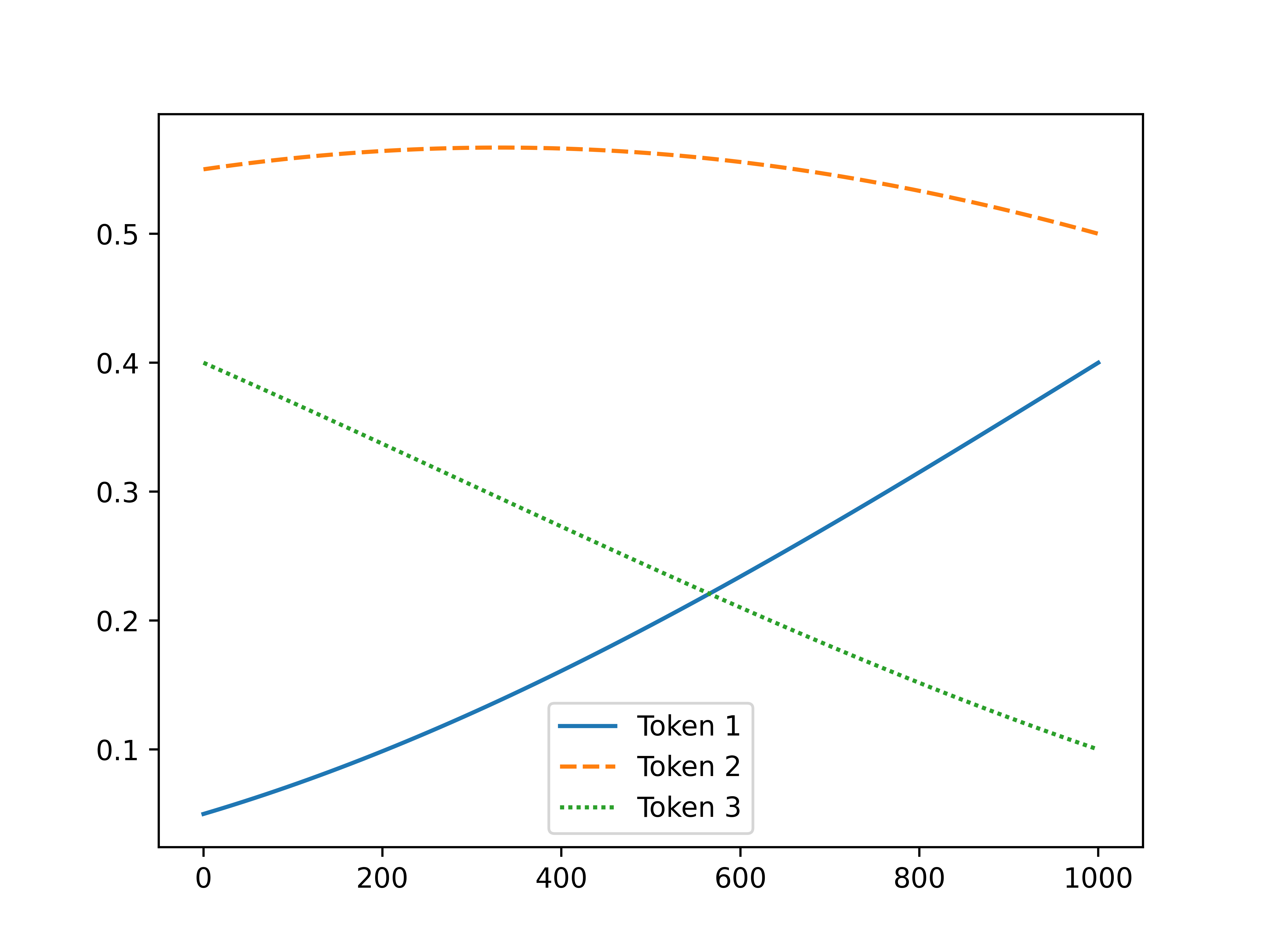}
        \caption{SLERP}
        \label{fig:slerp_traj}
    \end{subfigure}
\end{minipage}\hfill
\begin{minipage}{0.45\textwidth}
    \vspace{0.1em}
    \caption{Block-to-block weight changes ($\v w(t+1) - \v w(t)$). Constant metric speed produces smoothly varying (not constant) weight changes in Cartesian coordinates.}
    \label{fig:change}
    \centering
    \begin{subfigure}{\textwidth}
        \centering
        \includegraphics[width=\textwidth]{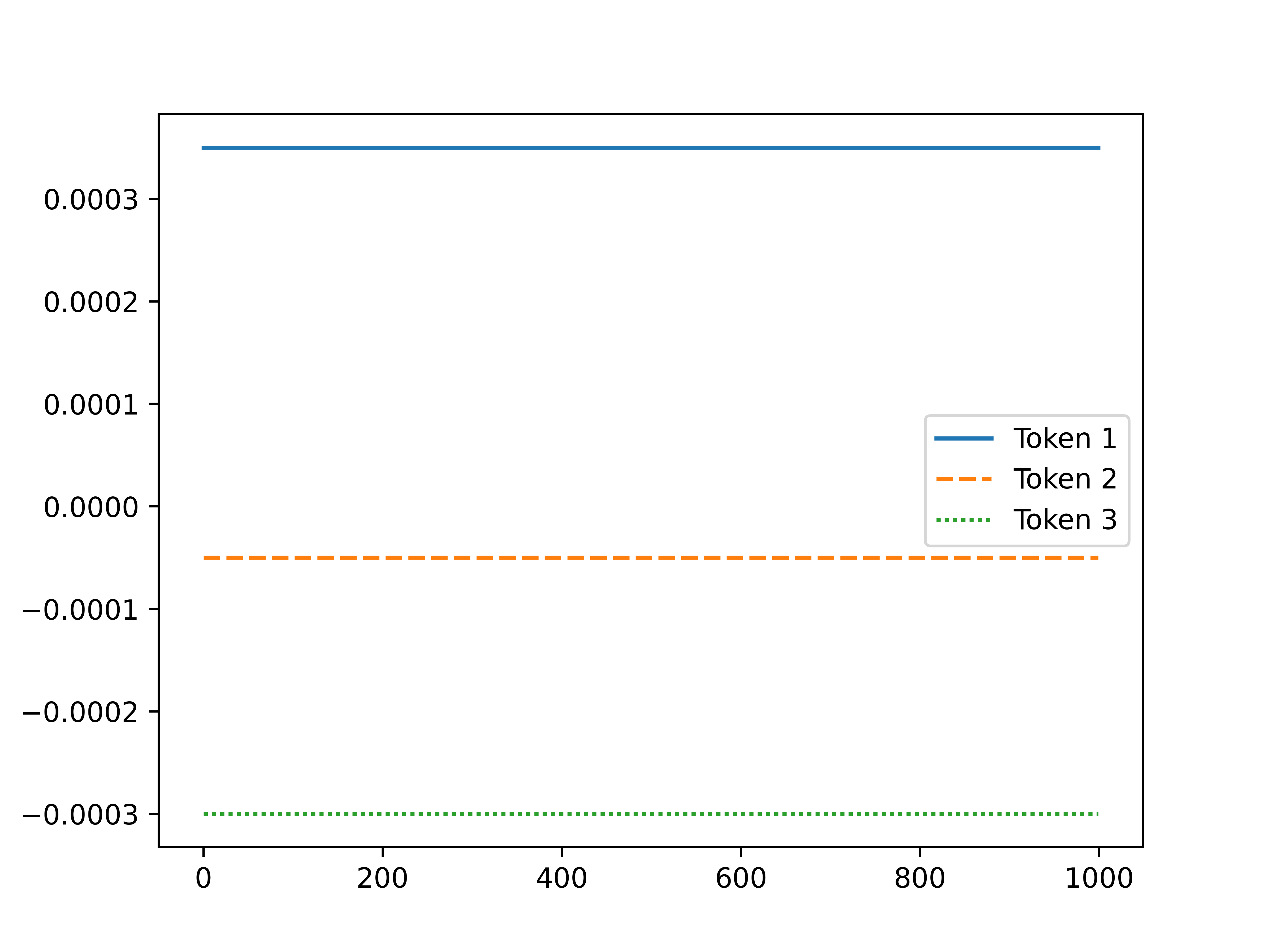}
        \caption{Linear}
        \label{fig:linear_change}
    \end{subfigure}
    \begin{subfigure}{\textwidth}
       \centering
        \includegraphics[width=\textwidth]{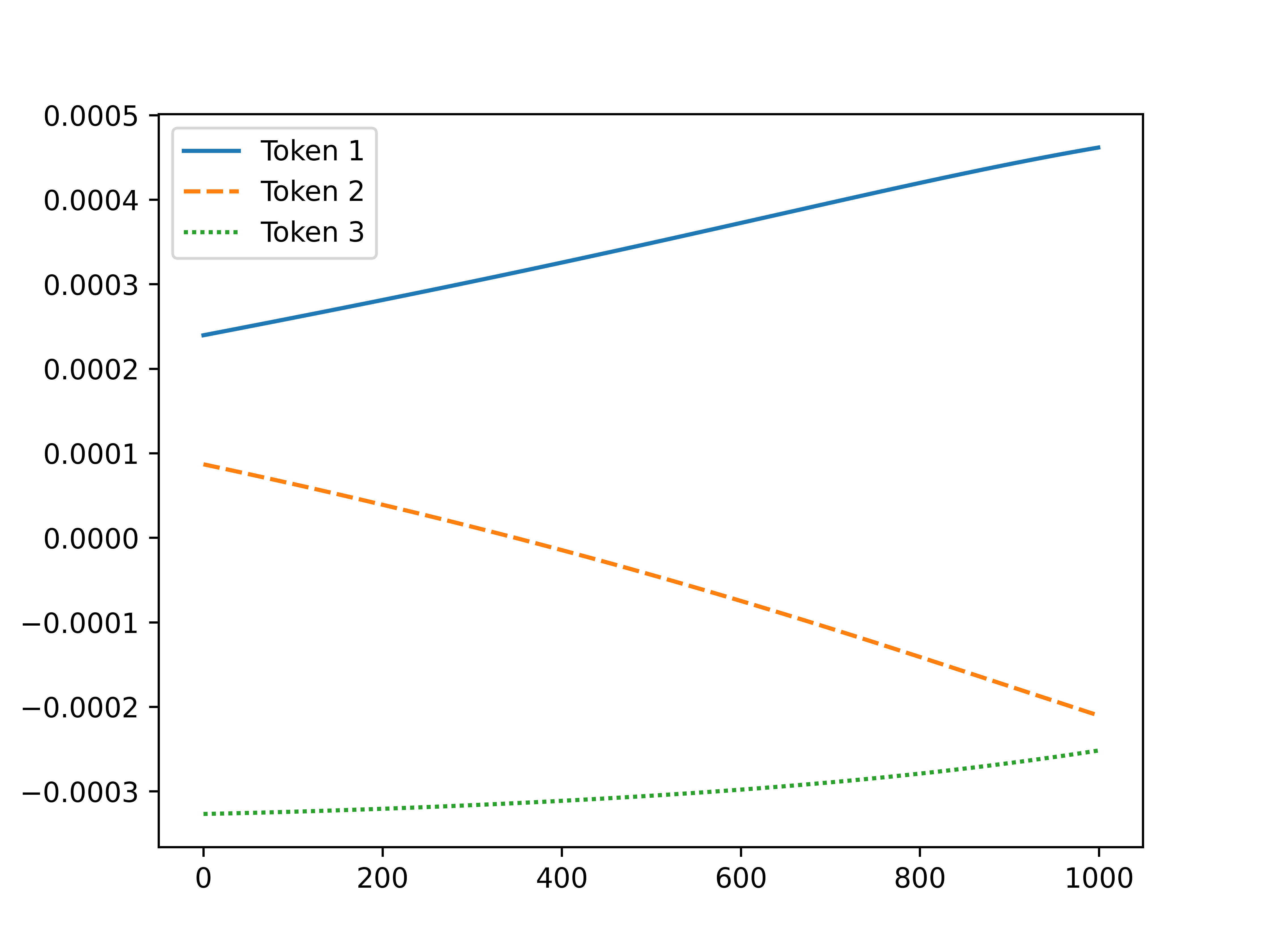}
        \caption{(AM+GM)/normalise}
        \label{fig:approx_change}
    \end{subfigure}
    \begin{subfigure}{\textwidth}
        \centering
        \includegraphics[width=\textwidth]{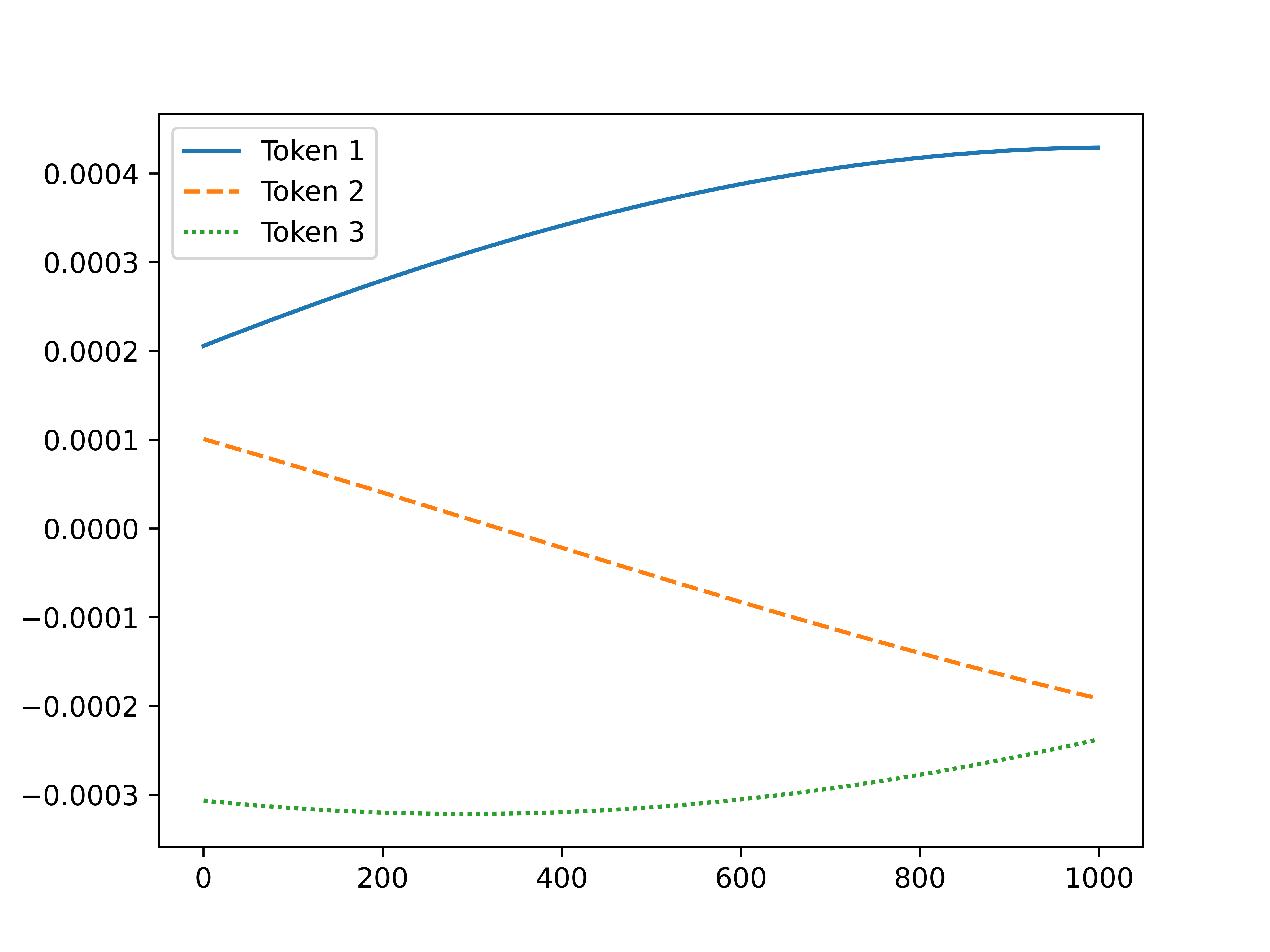}
        \caption{SLERP}
        \label{fig:slerp_change}
    \end{subfigure}
\end{minipage}
\end{figure*}

At this scale, (AM+GM)/normalise and SLERP trajectories are visually indistinguishable; block-to-block changes for SLERP vary smoothly with amplitude ${\sim}10^{-4}$.

\FloatBarrier

\section{Near-boundary experiments}
\label{app:near_boundary}

The sub-optimality bound (Theorem~\ref{thm:suboptimality}) assumes weights bounded away from zero.
Near the simplex boundary the Fisher--Rao metric diverges and the quadratic approximation underlying SLERP is least accurate.
Current G3M implementations enforce a minimum weight of $1\%$, so $w_{\min} = 0.01$ is the practical boundary.

\paragraph{Single midpoint ($f=2$)}
For $N=2$ with $\v w: (w_{\min},\, 1{-}w_{\min}) \to (1{-}w_{\min},\, w_{\min})$, we compare the direct Lambert~W midpoint formula of~\cite{willetts2024optimalrebalancingdynamicamms},
$m_i = w_i^{\mathrm{end}} / W\!\bigl(w_i^{\mathrm{end}} \cdot e\, /\, w_i^{\mathrm{start}}\bigr)$
(normalised), against SLERP and the brute-force L-BFGS-B optimum.

\begin{table}[h]
    \centering
    \caption{Single-midpoint ($f=2$) loss ratios for symmetric $N=2$ weight changes at varying boundary distance. Lambert~W, which maximises the exact retention ratio, beats SLERP near the boundary; both converge for small weight changes.}
    \label{tab:midpoint_boundary}
    \small
    \begin{tabular}{@{}cccc@{}}
    \hline
    $w_{\min}$ & SLERP/Opt. & Lambert\,W/Opt. & Lambert\,W/SLERP \\
    \hline
    0.01 & 1.178 & 1.053 & 0.894 \\
    0.02 & 1.122 & 1.041 & 0.928 \\
    0.05 & 1.059 & 1.024 & 0.967 \\
    0.10 & 1.025 & 1.012 & 0.987 \\
    0.20 & 1.005 & 1.003 & 0.998 \\
    0.30 & 1.001 & 1.001 & 1.000 \\
    0.40 & 1.000 & 1.000 & 1.000 \\
    \hline
    \end{tabular}
\end{table}

\begin{figure*}[t]
    \centering
    \includegraphics[width=\textwidth]{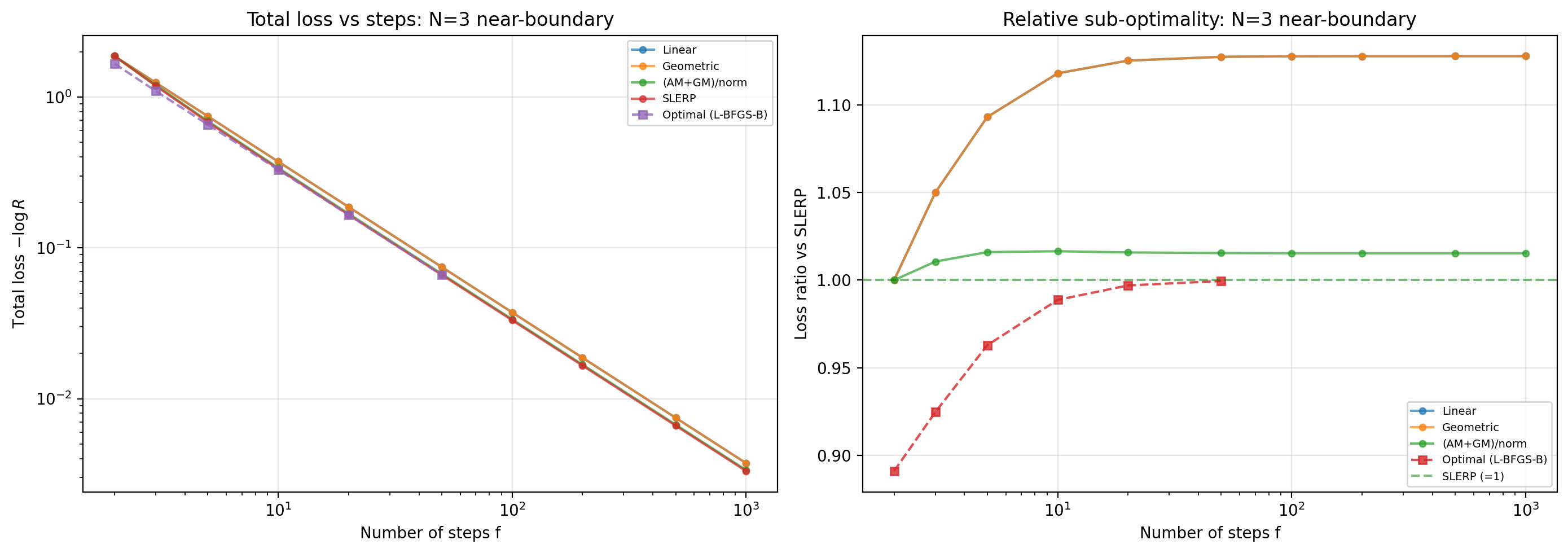}
    \caption{$N=3$ near-boundary configuration $\v w: (0.01, 0.01, 0.98) \to (0.49, 0.49, 0.02)$. Left: total arbitrage loss vs.\ step count $f$ for each interpolation method and the brute-force L-BFGS-B optimum. Right: relative sub-optimality (ratio of method loss to brute-force optimal loss minus one). SLERP closely tracks the brute-force optimum across all $f$; linear interpolation diverges by up to ${\sim}20\%$ near the boundary.}
    \label{fig:convergence}
\end{figure*}

\begin{figure*}[h]
    \centering
    \includegraphics[width=0.95\textwidth]{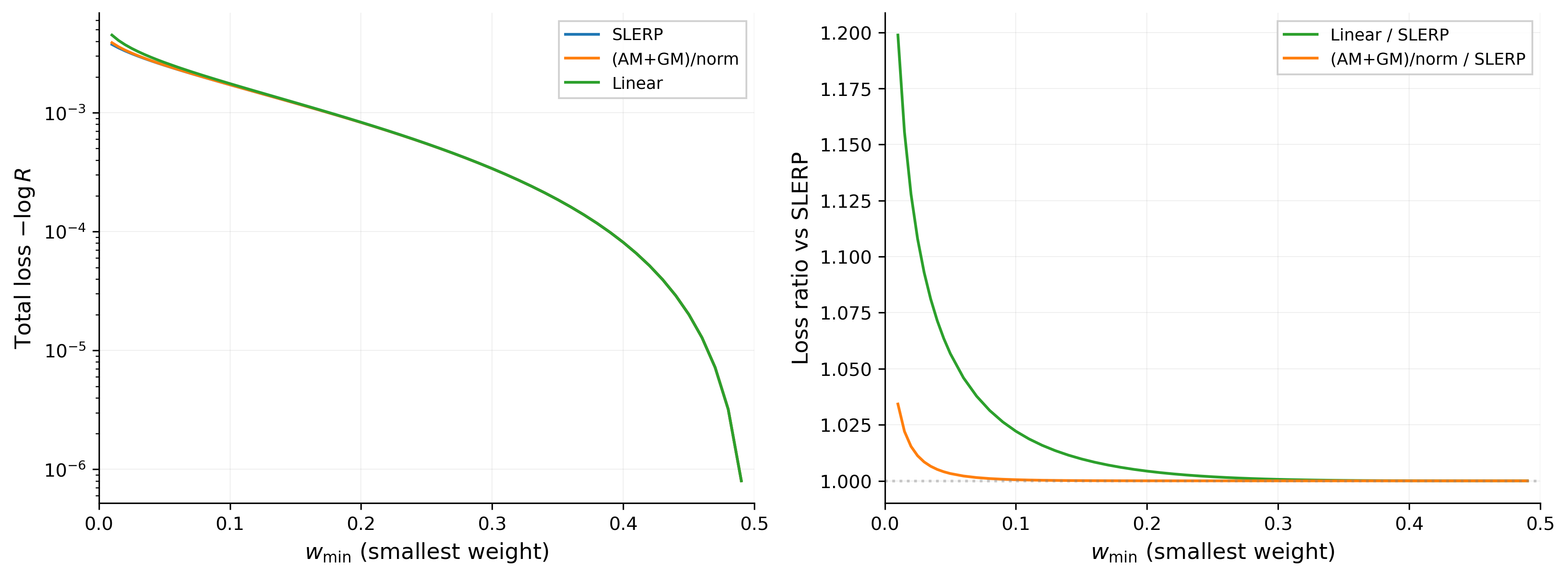}
    \caption{Left: total loss at $f=1000$ as a function of the smallest weight $w_{\min}$, for $N=2$ symmetric weight changes $(w_{\min},\, 1{-}w_{\min}) \to (1{-}w_{\min},\, w_{\min})$. Right: loss ratio relative to SLERP. SLERP's advantage over both linear and (AM+GM)/normalise grows as weights approach the boundary.}
    \label{fig:multistep_boundary}
\end{figure*}

Table~\ref{tab:midpoint_boundary} confirms that for large single-step updates near the boundary, Lambert~W, which optimises the exact finite-step retention, outperforms SLERP by up to $10\%$ at $w_{\min} = 0.01$.
Both converge to the true optimum as the weight change shrinks (i.e.\ as $w_{\min} \to 0.5$), consistent with the quadratic approximation becoming exact in the small-step limit.
\paragraph{Convergence across step counts}

Figure~\ref{fig:convergence} shows total loss and relative sub-optimality as a function of the step count~$f$ for the $N=3$ near-boundary configuration.
SLERP closely tracks the brute-force optimal across steps, validating the $O(\Omega^2/f^2)$ sub-optimality bound of Theorem~\ref{thm:suboptimality}.
Linear interpolation diverges from the optimum as $f$ grows, incurring up to ${\sim}20\%$ excess loss at this near-boundary configuration.
(AM+GM)/normalise sits much closer to SLERP than to linear, as expected from the third-order agreement established in \S\nolinebreak\ref{ssec:comparison_general_t}.

\paragraph{Multi-step boundary stress test}

Figure~\ref{fig:multistep_boundary} shows that SLERP's advantage over both linear and (AM+GM)/normalise \emph{widens} as weights approach the boundary: at $w_{\min} = 0.01$, linear interpolation incurs ${\sim}20\%$ more loss than SLERP, and (AM+GM)/normalise incurs ${\sim}3.5\%$ more.
For interior weights ($w_{\min} \geq 0.2$), all three methods are within $1\%$ of each other.

\begin{figure*}[h]
    \centering
    \includegraphics[width=0.85\textwidth]{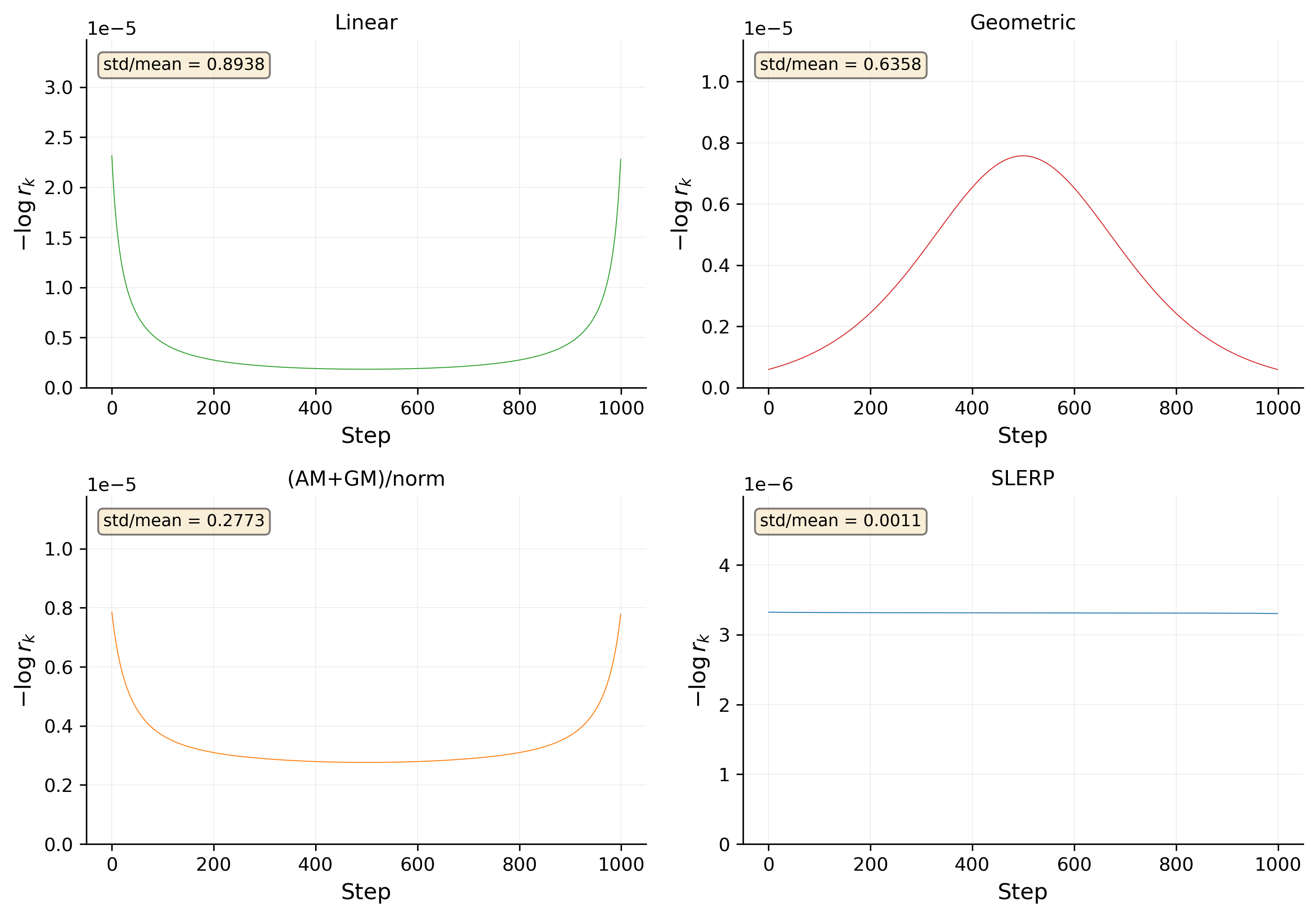}
    \caption{Per-step loss $-\log r_k$ near the simplex boundary ($N=3$, $f=1000$, $\v w: (0.01, 0.01, 0.98) \to (0.49, 0.49, 0.02)$). SLERP achieves near-constant loss per step (std/mean $= 0.0011$); other methods show large variation, especially at steps where weights pass through their smallest values.}
    \label{fig:perstep_boundary}
\end{figure*}

Figure~\ref{fig:perstep_boundary} shows per-step losses for the $N=3$ near-boundary configuration at $f=1000$.
SLERP maintains near-perfect uniformity (std/mean $= 0.0011$) even with weights at the $1\%$ floor: compare with std/mean $= 0.0002$ for the interior configuration of Figure~\ref{fig:perstep_loss}.
Linear interpolation shows extreme non-uniformity (std/mean $= 0.89$) with a loss spike at the boundary-crossing steps.

\FloatBarrier
\newpage
\section{Derivation of the LVR potential}
\label{app:lvr_potential}

The factorisation $V_{k+1}/V_k = r_k \cdot \prod_i (P_i^{(k+1)}/P_i^{(k)})^{w_i^{(k+1)}}$ (Eq~\eqref{eq:value_factorisation}) gives, over $f$ steps:

\[
V_T/V_0 = \prod_{k=1}^f r_k \cdot \prod_{k=1}^f \prod_i \!\left(P_i^{(k+1)}/P_i^{(k)}\right)^{w_i^{(k+1)}}
\]

Since the driftless GBM increments are independent across blocks, the expectation factorises:

\[
\mathbb{E}[V_T/V_0] = \prod_k r_k \cdot \prod_k \mathbb{E}\!\left[\prod_i (P_i^{(k+1)}/P_i^{(k)})^{w_i^{(k+1)}}\right]
\]

The $k$-th factor in the second product is $\mathbb{E}[\exp(\sum_i w_i \log(P_i'/P_i))]$.
Under driftless GBM, the weighted log-price return $\sum_i w_i \log(P_i'/P_i)$ is normally distributed with mean $-\frac{\Delta t}{2}\sum_i w_i\sigma_i^2$ and variance $\Delta t \cdot \v w^T\Sigma\v w$.
By the moment generating function of the normal distribution:
\begin{align}
\mathbb{E}\!\left[\exp\!\left(\sum_i w_i \log(P_i'/P_i)\right)\right]
  &= \exp\!\left(-\tfrac{\Delta t}{2}\sum_i w_i\sigma_i^2 + \tfrac{\Delta t}{2}\v w^T\Sigma\v w\right) \nonumber\\
  &= \exp\bigl(-\ell(\v w)\,\Delta t\bigr)
\end{align}
Hence maximising $\mathbb{E}[V_T]$ is equivalent to minimising

\[
\sum_k D_{\mathrm{KL}}\!\bigl(\v w^{(k)} \| \v w^{(k-1)}\bigr) + \Delta t\sum_k \ell(\v w^{(k)}),
\]

with the LVR potential $\ell(\v w) = \frac{1}{2}[\sum_i w_i\sigma_i^2 - \v w^T\Sigma\v w]$.

\section{Asymptotic validation of the perturbative regime}
\label{app:asymptotics}

The two-token pendulum (\S\ref{sec:optimal_f}) can be written with $\epsilon = 16/(fT\sigma^2)$:
\[
\epsilon\,\ddot{\theta} = \sin(4\theta), \quad \theta(0) = \theta_0, \quad \theta(1) = \theta_1
\]
This is a singular perturbation problem: the coefficient $\epsilon$ multiplying the highest derivative vanishes in the limit $\epsilon \to 0$.
Standard boundary-layer theory~\cite{bender_orszag1999} gives layers of width $\delta = \sqrt{\epsilon} = 4/\sqrt{fT\sigma^2}$.
The matched expansion requires $\delta \ll 1$, i.e.\ $fT\sigma^2 \gg 16$.

\begin{center}
\small
\begin{tabular}{@{}llll@{}}
\toprule
Scenario & $fT\sigma^2$ & $\delta$ & Regime \\
\midrule
$f{=}f^*{\approx}2400$, $T{=}8$\,hr & ${\approx}\,0.7$ & ${\approx}\,5$ & Perturbative ($\delta > 1$) \\
$f{=}7200$ (1\,day) & ${\approx}\,6.5$ & ${\approx}\,1.6$ & Still perturbative \\
$f{=}72\,000$ (10\,days) & ${\approx}\,648$ & ${\approx}\,0.16$ & Boundary-layer \\
\bottomrule
\end{tabular}
\end{center}

For practical parameters, $\delta > 1$: the boundary layers are wider than the domain and the perturbative expansion is rigorously justified.
The boundary-layer regime ($f \sim 72\,000$) corresponds to $f \gg f^*$, where the pool incurs far more LVR than it saves in rebalancing cost.

\paragraph{Perturbative correction via Green's function}

The SLERP baseline is $\theta_0(s) = \theta_{\mathrm{start}} + \Omega s$.
Writing $\theta(s) = \theta_0(s) + \epsilon_1(s) + O(\mu^2)$ and linearising the pendulum (Eq~\eqref{eq:pendulum}) about $\theta_0$, the first-order correction satisfies
\[
\ddot{\epsilon}_1 = \mu\,\sin(4\theta_{\mathrm{start}} + 4\Omega s), \quad \epsilon_1(0) = \epsilon_1(1) = 0
\]
where $\mu = fT\sigma^2/16$.
The solution is given by the standard Green's function for the two-point boundary problem:
\begin{gather}
\epsilon_1(s) = \mu\!\int_0^1 G(s,s')\sin(4\theta_{\mathrm{start}} + 4\Omega s')\,ds', \nonumber\\
G(s,s') = \min(s,s')\bigl(1-\max(s,s')\bigr)
\end{gather}
The integral is elementary (products of $s'$ with sines), and the magnitude is bounded by $|\epsilon_1| \leq \mu/4$.
For the worked example of \S\ref{sec:optimal_f} at $f = f^*$: $\mu \approx 0.045$, giving $|\epsilon_1| \lesssim 0.011$\,rad, a weight perturbation of $O(10^{-2})$.
The correction pushes the path away from $\theta = \pi/4$ (equal weights, maximum LVR) toward the vertices (minimum LVR).

\paragraph{The outer solution and boundary-layer structure}

In the $\epsilon \to 0$ limit (many steps, high LVR), the outer solution of the singular perturbation demands $\sin(4\theta_{\mathrm{outer}}) = 0$, selecting LVR-minimising vertices ($\theta = 0$ or $\pi/2$).
Boundary layers of width $\delta = \sqrt{\epsilon}$ connect the outer solution to the boundary conditions.
Near a vertex ($\theta \approx 0$), linearisation gives $\sin(4\theta) \approx 4\theta$ and the inner solution decays exponentially: $\theta \sim \theta_0\exp(-2S)$ where $S = s/\delta$.
The composite expansion $\theta_{\mathrm{comp}}(s) = \theta_{\mathrm{outer}} + [\theta_{\mathrm{inner,L}}(s/\delta) - \theta_{\mathrm{outer}}] + [\theta_{\mathrm{inner,R}}((1-s)/\delta) - \theta_{\mathrm{outer}}]$ describes a trajectory that jumps to a vertex (zero LVR) and remains there for the bulk of the interpolation; but this regime ($\delta \ll 1$) requires $fT\sigma^2 \gg 16$, which is never reached at $f = f^*$ for practical parameters.

\section{Jacobi metric: joint path and speed optimisation}
\label{app:jacobi}
\begin{figure*}[h]
    \centering
    \includegraphics[width=\textwidth]{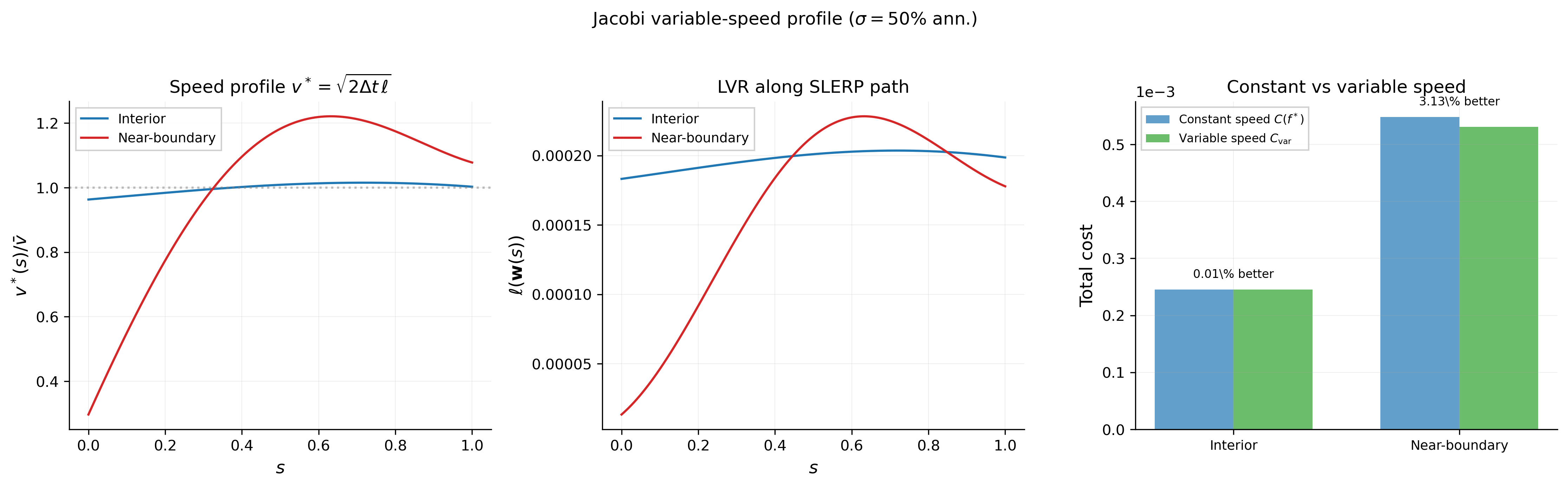}
    \caption{Jacobi variable-speed profile at $\sigma = 50\%$ ann.\ for the interior $(0.5,0.5)\to(0.7,0.3)$ and near-boundary $(0.01,0.99)\to(0.99,0.01)$ configurations.
    Left: optimal speed $v^*(s) = \sqrt{2\Delta t\,\ell}$ normalised by the constant-speed baseline $\bar{v}$.
    Centre: LVR rate $\ell(\v w(s))$ along the SLERP path.
    Right: total cost comparison; variable speed reduces $C(f^*)$ by $0.01\%$ for interior weights and $3.13\%$ near the boundary.}
    \label{fig:jacobi_speed}
\end{figure*}

Fix a path parameterised by arc-length $s \in [0,\Omega]$.
At each block the pool advances by arc-length~$v$, paying rebalancing cost $v^2/2$ and LVR cost $\Delta t_{\mathrm{block}}\,\ell(\v w)$.
Covering an arc element $ds$ requires $ds/v$ blocks, so the cost per unit arc-length is $v/2 + \Delta t\,\ell/v$.
By AM-GM:
\[
\frac{v}{2} + \frac{\Delta t\,\ell}{v} \geq \sqrt{2\,\Delta t_{\mathrm{block}}\,\ell(\v w(s))}
\]
with equality at $v^*(s) = \sqrt{2\,\Delta t_{\mathrm{block}}\,\ell(\v w(s))}$: fast through high-LVR regions, slow near the vertices.

The minimum cost $C^* = \int_0^\Omega \sqrt{2\Delta t\,\ell}\;ds_{\mathrm{FR}}$ defines a Jacobi metric $\tilde{g}_{ij} = \ell(\v w)\cdot g^{\mathrm{FR}}_{ij}$.
The jointly optimal trajectory is a geodesic of this conformally deformed metric (Jacobi--Maupertuis principle~\cite{arnold1989}).
\begin{remark}[Hierarchy of approximations]
\label{rem:hierarchy}
The three levels of optimisation form a natural hierarchy:
(i)~SLERP uses the Fisher--Rao metric~$g^{\mathrm{FR}}$ and optimises the path only, ignoring LVR, at constant speed;
(ii)~the modified geodesic (Proposition~\ref{prop:modified_geodesic}) uses~$g^{\mathrm{FR}}$ with the LVR potential and optimises the path at fixed~$f$, still at constant speed;
(iii)~the Jacobi geodesic uses the conformal metric $\ell \cdot g^{\mathrm{FR}}$ and optimises both path and speed jointly, with $v^*(s) = \sqrt{2\,\Delta t\,\ell(\v w(s))}$.
SLERP is the $\ell \to \mathrm{const}$ limit of the Jacobi geodesic; the modified geodesic is the fixed-speed restriction of the Jacobi problem; all three coincide when LVR is constant on the simplex.
\end{remark}

The MEV guardrail speed limit~\cite{willetts2024multiblockmevopportunities} at the simplex boundary scales as $v_{\max} \propto u_{\max}\sqrt{w_{\min}}$, the same $\sqrt{w_{\min}}$ dependence as $v^*$.
For crypto parameters ($\sigma = 3\%$/day, $\Delta t_{\mathrm{block}} = 12$\,s, $u_{\max} \approx 0.05$):
$v^*/v_{\max} \approx 0.007$.
The LVR-optimal speed is ${\sim}100\times$ slower than the guardrails allow; the two constraints do not conflict.

For $N \geq 3$, the Jacobi geodesic follows a different \emph{path} than SLERP, curving away from the simplex interior where LVR is highest.
Figure~\ref{fig:jacobi_n3} compares SLERP and the numerically optimised Jacobi geodesic for the paper's $N = 3$ configuration.

The Jacobi path visibly avoids the high-LVR centroid region (Panel~A), producing reduced non-monotonicity in individual weight components (Panel~B).
However, the cost improvement is small: the Jacobi geodesic saves only ${\sim}0.4\%$ over SLERP (Panel~C), confirming that SLERP is near-optimal even when the path shape is freed.

\begin{figure*}[h]
    \centering
    \includegraphics[width=\textwidth]{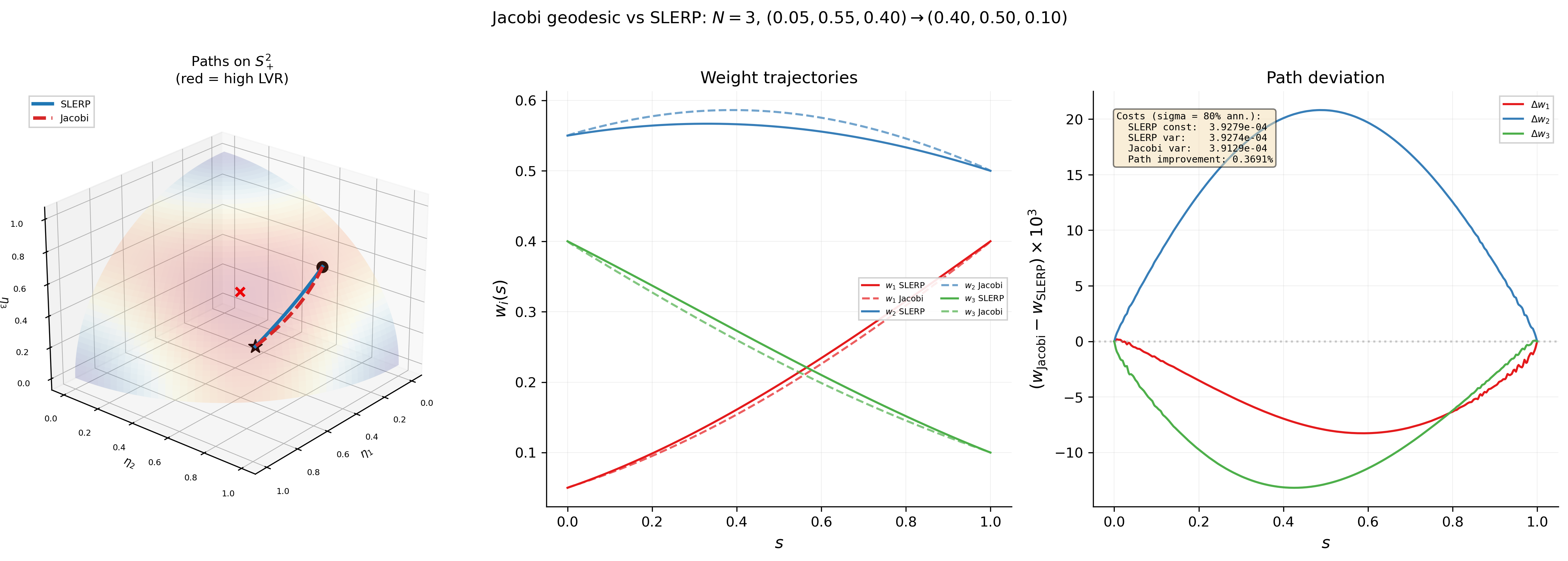}
    \caption{Jacobi geodesic vs SLERP on $S^2_+$ for $(0.05, 0.55, 0.40) \to (0.40, 0.50, 0.10)$ with equal volatilities $\sigma = 80\%$ ann.
    Left: paths on the positive octant of the Hellinger sphere; surface colour encodes LVR (red = high, blue = low).
    Centre: weight trajectories $w_i(s)$.
    Right: deviation from SLERP ($\times 10^3$).
    Inset: optimal cost hierarchy; Jacobi improves on SLERP by $0.37\%$.}
    \label{fig:jacobi_n3}
\end{figure*}

\FloatBarrier
\newpage
\section{Pad\'e approximant for the per-step KL cost}
\label{app:pade}

The SLERP optimality of \S\ref{sec:slerp} rests on the quadratic truncation of the per-component KL function $h(u) = (1+u)\log(1+u) - u$, where $u_i = \Delta w_i / w_i$.
The full Taylor series is $h(u) = u^2/2 - u^3/6 + u^4/12 - \cdots$.
A $[1,1]$ Pad\'e approximant~\cite{baker_graves_morris1996} of $g(u) = 2h(u)/u^2 = 1 - u/3 + u^2/6 - \cdots$, multiplied back by $u^2/2$, gives
\begin{equation}
h_{\mathrm{P}}(u) = \frac{u^2}{2}\cdot\frac{1 + u/6}{1 + u/2}
\label{eq:pade}
\end{equation}
which matches $h(u)$ through $O(u^4)$, two orders beyond the quadratic truncation (Figure~\ref{fig:pade}).

\begin{figure}[ht!]
\centering
\includegraphics[width=0.6\columnwidth]{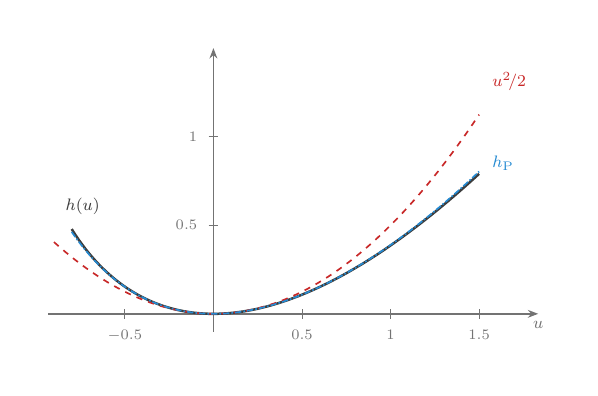}
\caption{The per-component KL function $h(u) = (1+u)\log(1+u) - u$ (solid), the quadratic approximation $u^2/2$ (dashed), and the Pad\'e approximant $h_{\mathrm{P}}$ (dot-dashed).
The quadratic overestimates for $u > 0$ (weight increases) and underestimates for $u < 0$ (weight decreases); the Pad\'e captures this cubic asymmetry.}
\label{fig:pade}
\end{figure}

The decomposition $h_{\mathrm{P}}(u) = u^2/2 - (u^3/6)/(1+u/2)$ separates the symmetric (Fisher--Rao) piece from the resummed $\alpha$-connection correction.
For $u > 0$, $h_{\mathrm{P}} < u^2/2$: larger steps where weight increases are cheaper than SLERP assumes; for $u < 0$, the inequality reverses.
This is the cubic asymmetry that the Lambert~W midpoint (Proposition~\ref{prop:lambert_kl}) exploits.

The remainder beyond the Pad\'e is $h(u) - h_{\mathrm{P}}(u) = O(u^5)$, so optimising $h_{\mathrm{P}}$ instead of $u^2/2$ would yield a ``Pad\'e-SLERP'' with relative sub-optimality $O(\Omega^6/f^6)$, compared to SLERP's $O(\Omega^2/f^2)$.
For large $f$, all three surrogates (quadratic, Pad\'e, exact) produce indistinguishable trajectories; the Pad\'e is most useful at $f = 2$, where it gives a midpoint closer to the Lambert~W formula.

\begin{figure}[h]
    \centering
    \includegraphics[width=0.9\columnwidth]{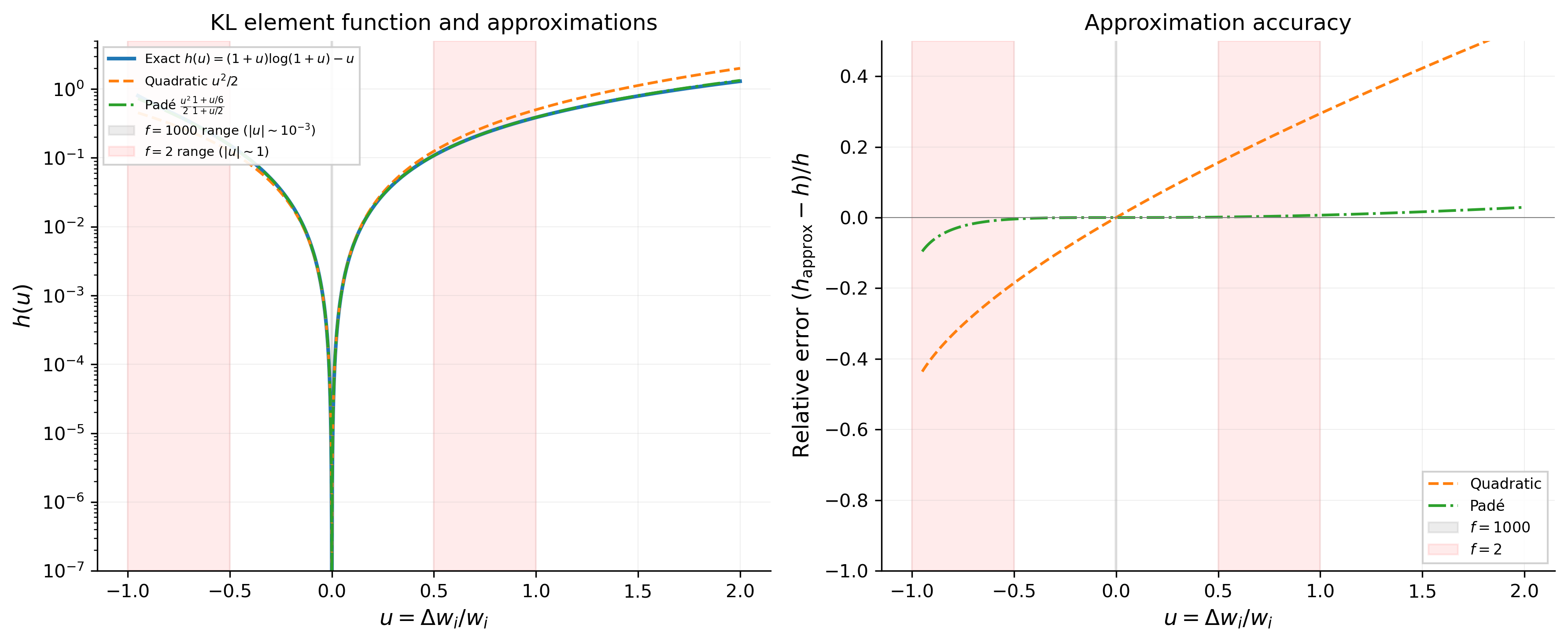}
    \caption{Left: $h(u)$ and its approximations on a log scale; shaded bands show the typical $|u|$ range at $f=1000$ (grey) and $f=2$ (pink).
    Right: relative error $(h_{\mathrm{approx}} - h)/h$.
    The quadratic diverges for $|u| \gtrsim 0.5$; the Pad\'e stays within a few percent across the entire $f=2$ range.}
    \label{fig:pade_accuracy}
\end{figure}

\section{Robustness to non-GBM prices}
\label{app:non_gbm}

The analytical cost $C(f)$ and optimal step count $f^*$ are derived under GBM.
Figure~\ref{fig:non_gbm} tests robustness by simulating the same two-token rebalancing under Merton jump-diffusion~\cite{merton1976} and GARCH(1,1) prices, both calibrated to the same annualised volatility.
All three price models produce nearly identical loss curves; the GBM analytical $C(f)$ and $f^*$ remain accurate.

\begin{figure}[h]
    \centering
    \includegraphics[width=0.6\columnwidth]{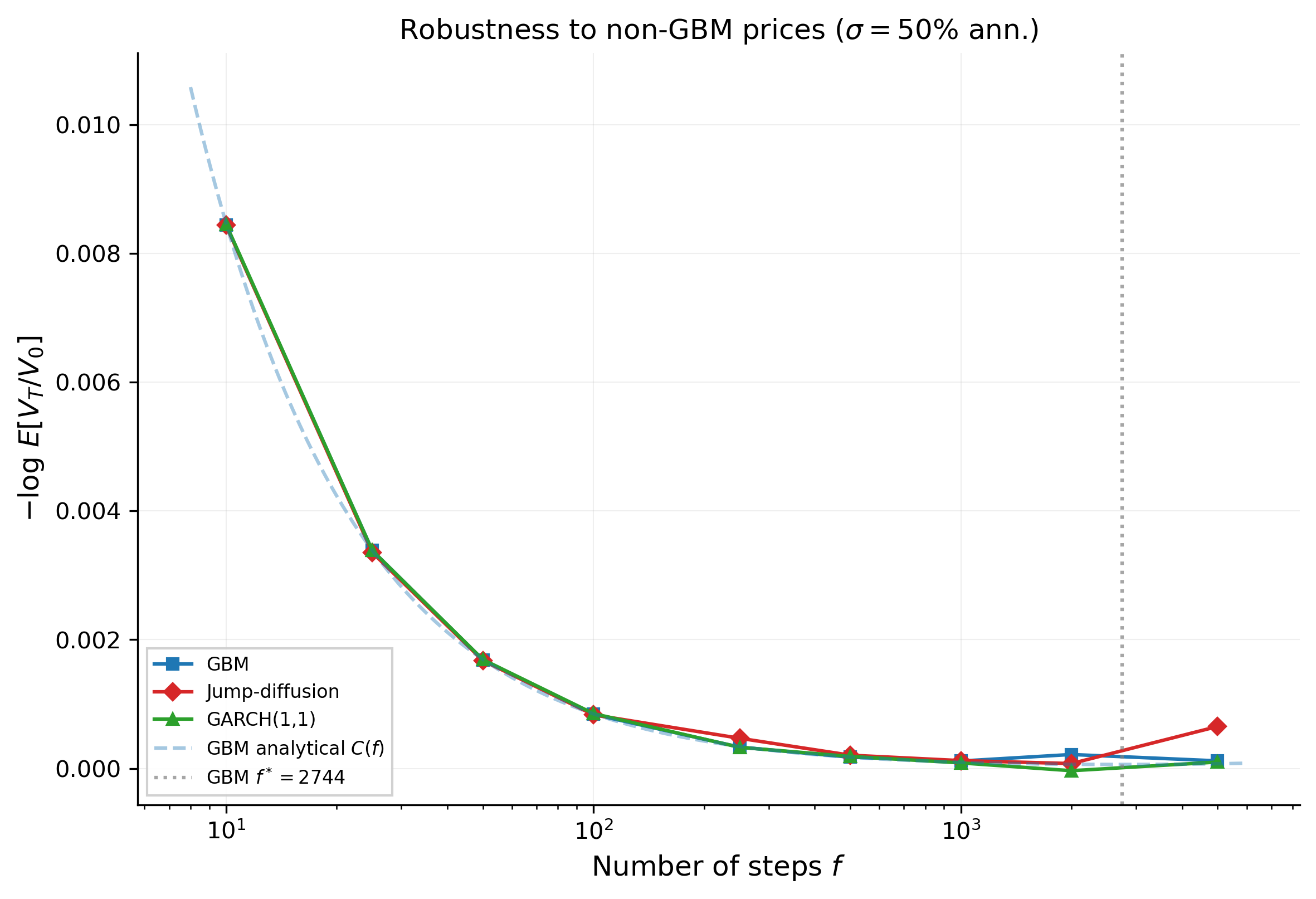}
    \caption{Expected log-loss vs.\ step count under GBM, Merton jump-diffusion, and GARCH(1,1) prices ($\sigma = 50\%$ ann., $10\,000$ paths each).
    The analytical $C(f)$ (dashed) and $f^* = 2744$ (dotted) are computed under GBM.
    All three models agree closely; the GBM-derived $f^*$ remains a good choice.}
    \label{fig:non_gbm}
\end{figure}

\section{The LVR-dominated regime and $\lambda$ regimes}
\label{app:lvr_regime}

\paragraph{Why $\lambda \gg 1$ is degenerate}

The modified geodesic equation (Eq~\eqref{eq:modified_geodesic}) has $1/f$ multiplying the highest derivative.
When $\lambda = fT\sigma^2/\Omega^2 \gg 1$, this becomes a \emph{singular perturbation}: the ``outer'' (interior) solution demands $g^{ij}\partial\ell/\partial w^j = 0$, i.e.\ the path must sit at a critical point of the LVR potential.

For expected-value maximisation, $\ell(\v w) = \frac{\sigma^2}{2}(1 - \|\v w\|^2)$ is minimised at the simplex vertices ($\ell = 0$ when all weight is in one asset).
The formal $\lambda \to \infty$ solution is therefore a \emph{bang-bang trajectory}:
\begin{enumerate}
    \item Boundary layer at $s = 0$: sprint from $\v w^{\mathrm{start}}$ to the nearest vertex.
    \item Interior: sit at the vertex (zero LVR).
    \item Boundary layer at $s = 1$: sprint to $\v w^{\mathrm{end}}$.
\end{enumerate}
The boundary layers have width $O(1/\sqrt{\lambda})$, set by the balance between the kinetic (rebalancing) and potential (LVR) terms.
In practice, the minimum-weight constraint ($w_i \geq w_{\min}$) in current G3M implementations prevents reaching the vertices, further limiting any benefit from this structure.
The enormous rebalancing cost of the two jumps dominates any LVR saving.

\paragraph{The correct response: reduce $f$}

Being in the $\lambda \gg 1$ regime means $f \gg f^*$: the pool is using too many interpolation steps.
Each additional step beyond $f^*$ costs more in LVR than it saves in rebalancing.
The correct response is to rebalance faster (reduce $f$ toward $f^*$), rather than to find a cleverly shaped curve through the interior of the simplex.

\paragraph{Regime summary}

Table~\ref{tab:lambda_regimes} classifies the three operating regimes by the dimensionless LVR strength $\lambda$.
At the optimum ($f = f^*$), $\lambda^* = 2\sigma^2/\bar\ell = O(1)$, so the optimum always sits in the perturbative regime where SLERP is near-optimal.

\begin{table*}[h]
\centering
\begin{tabular}{llll}
\hline
Regime & Condition & $\lambda$ & Prescription \\
\hline
Constant prices & $\sigma = 0$ & $0$ & SLERP; maximise $f$ \\
Near-optimal & $f \approx f^*$ & $O(1)$ & Choose $f^*$; SLERP + small correction \\
Too many steps & $f \gg f^*$ & $\gg 1$ & Reduce $f$; do not optimise the curve \\
\hline
\end{tabular}
\caption{Operating regimes classified by the dimensionless LVR strength $\lambda = fT\sigma^2/\Omega^2$.
The near-optimal regime is the only one where the path shape matters, and there SLERP is sufficient (Appendix~\ref{app:asymptotics}).}
\label{tab:lambda_regimes}
\end{table*}

\section{MEV guardrail compatibility}
\label{app:mev_guardrails}

The weight trajectory $\{\v w^{(k)}\}$ is deterministic and publicly visible on-chain.
This appendix shows that the interpolation method interacts favourably with the MEV guardrails of~\cite{willetts2024multiblockmevopportunities}.

\paragraph{The open-pool property}

Over a typical rebalancing ($f \sim 2400$ blocks), a strategic attacker could in principle coordinate across multiple weight updates.
In practice, the \emph{open-pool property} prevents multi-step attack strategies from composing:
\begin{enumerate}
    \item \textbf{Immediate arbitrage of any displacement.}
    If an attacker pushes the pool off-market in block $k$, competitive arbitrageurs restore equilibrium in block $k{+}1$, before the attacker can exploit subsequent weight updates.
    \item \textbf{Independent per-step opportunities.}
    Competitive arbitrage resets the pool to market prices (within the fee band) between every pair of blocks, preventing the accumulation of a manipulated state across updates.
    \item \textbf{Open extraction.}
    The arbitrage opportunity created by each weight update is available to all participants, not exclusively to the original manipulator.
\end{enumerate}
The exception is \emph{consecutive block control} (proposer collusion), where an attacker can censor competitive arbitrage across a window of blocks.
The 5-block protection window of~\cite{willetts2024multiblockmevopportunities} is designed for precisely this threat.
Beyond the consecutive-control window, competitive dynamics provide the security guarantee.
Per-step guardrails therefore need only make each individual weight update safe; the interpolation method determines the per-step weight changes and hence interacts with these guardrails.

\paragraph{SLERP's variable absolute weight change}

SLERP's constant metric speed means each step has equal KL cost but variable absolute weight change $|\Delta w_i|$.
In Hellinger coordinates, $\Delta\eta_i \approx \mathrm{const}$ per step, so $\Delta w_i \approx 2\sqrt{w_i}\,\Delta\eta_i$: smaller absolute changes when $w_i$ is small, larger when $w_i$ is large.
This is naturally aligned with the guardrail constraint, which is most restrictive at small weights (where the same $|\Delta w_i|$ is the largest proportional change and creates the largest MEV opportunity).

\paragraph{The $\sqrt{w_{\min}}$ scaling match}

If the guardrail bounds the proportional change ($|u_i| \leq u_{\max}$, binding at $w_{\min}$), the arc-length speed limit is $v_{\max} \propto u_{\max}\sqrt{w_{\min}}$.
The Jacobi-optimal speed (Appendix~\ref{app:jacobi}) near a vertex where $w \approx w_{\min}$ is $v^* = \sqrt{2\,\Delta t_{\mathrm{block}}\,\ell} \propto \sigma\sqrt{\Delta t_{\mathrm{block}}\,w_{\min}}$, giving the same $\sqrt{w_{\min}}$ dependence.
Their ratio,
\[
\frac{v^*}{v_{\max}} = \frac{\sigma\sqrt{\Delta t_{\mathrm{block}}}}{u_{\max}}\,,
\]
is $\approx 0.007$ for crypto parameters ($\sigma = 3\%$/day, $\Delta t_{\mathrm{block}} = 12$\,s, $u_{\max} \approx 0.05$).
The LVR-optimal speed is ${\sim}100\times$ slower than the guardrails allow, even at the simplex boundary.

\section{Alternative ordering of weight and price changes}
\label{app:ordering}

The value factorisation Eq~\eqref{eq:value_factorisation} models weight and price changes as occurring simultaneously.
An alternative~\cite{tfmm_litepaper,tfmm} applies them in two steps: first prices change (the pool being arbitraged at the old weights), then weights change (and the pool is traded against again).
This gives
\[V_{\mathrm{new}}/V_{\mathrm{old}} = r \cdot \prod_i (P^{{\mathrm{new}}}_i/P^{{\mathrm{old}}}_i)^{w_i^{\mathrm{old}}},\]
contrasting with
\[V_{\mathrm{new}}/V_{\mathrm{old}} = r \cdot \prod_i (P^{{\mathrm{new}}}_i/P^{{\mathrm{old}}}_i)^{w_i^{\mathrm{new}}}\]
as in Eq~\eqref{eq:value_factorisation}.
The difference (old vs new weight exponents on the price return) does not affect any downstream results.
The cross-term telescopes under either ordering: over $f$ steps, $\sum_k \sum_i w_i^{(k)} \sigma_i^2$ and $\sum_k \sum_i w_i^{(k-1)} \sigma_i^2$ differ by an endpoint-only term $\sum_i (w_i^{(f)} - w_i^{(0)})\sigma_i^2$, which drops out of the path variation.
The path optimisation is unchanged.

\section{Deferred proof: SLERP optimality with fees (Corollary~\ref{cor:slerp_fees})}
\label{app:slerp_fees_proof}

\begin{proof}
The pool's total net loss along any path $P$ is $C(P) = V\!\sum_k D_{\mathrm{KL}}\!\bigl(\v w^{(k)} \big\| \v w^{(k-1)}\bigr) - F_{\mathrm{total}}(P) + O((1{-}\gamma)^2)$.
SLERP uniquely minimises the KL sum (Corollary~\ref{cor:slerp_optimal}).
Since the net loss is KL cost minus fee revenue, it remains to show that SLERP's fee revenue is no less than any alternative's: $F_{\mathrm{total}}(\mathrm{SLERP}) \geq F_{\mathrm{total}}(P)$.

By Proposition~\ref{prop:fee_telescoping}, the fee revenue along any monotonic path equals the endpoint-only value $F_0 = \frac{1-\gamma}{2\gamma}V\sum_i|w_i^e - w_i^s| + O((1{-}\gamma)^2)$.
For $N = 2$, SLERP is monotonic ($\theta = \arcsin(\sqrt{w})$ interpolates linearly), so $F_{\mathrm{total}}(\mathrm{SLERP}) = F_0$.
For $N \geq 3$, the great-circle arc on $S^{N-1}_+$ bulges outward from the chord: components whose start and end values are close may overshoot by $O(\Omega^2)$ in Hellinger coordinates.
This makes SLERP non-monotonic in those components, so $F_{\mathrm{total}}(\mathrm{SLERP}) = F_0 + O(\Omega^2(1{-}\gamma)/\gamma) \geq F_0$.

Any monotonic alternative path $P_{\mathrm{mono}}$ has $\sum D_{\mathrm{KL}}(P_{\mathrm{mono}}) > \sum D_{\mathrm{KL}}(\mathrm{SLERP})$ (non-geodesic) and $F_{\mathrm{total}}(P_{\mathrm{mono}}) = F_0 \leq F_{\mathrm{total}}(\mathrm{SLERP})$, so $C(P_{\mathrm{mono}}) > C(\mathrm{SLERP})$.
Any non-monotonic alternative $P'$ is also non-geodesic, so its KL cost exceeds SLERP's.
While $P'$ may generate more fee revenue from additional non-monotonicity, the extra KL cost dominates: a detour of Hellinger size $\delta$ costs $O(\delta^2)$ in KL but generates at most $O(\delta(1{-}\gamma))$ in fee revenue, and detours smaller than $(1{-}\gamma)$ fall inside the fee band and generate no arb or fee revenue.
\end{proof}

\section{Fee mechanics: supporting derivations and empirical detail}
\label{app:fee_mechanics}

This appendix provides the full derivation supporting Proposition~\ref{prop:fee_adjusted_f} and additional empirical detail on the mechanisms contributing to~$\phi$.

\paragraph{Cost plateau above $f_{\mathrm{threshold}}$ (constant-price model)}

In the idealised constant-price model, we can compute the rebalancing cost exactly when $f > f_{\mathrm{threshold}}$.
Each SLERP step advances the Hellinger coordinate by $\Omega/f$, too small to breach the fee band.
Displacement accumulates linearly over consecutive steps (the rising edge of the sawtooth observed in~\cite{willetts2026poolsportfoliosobservedarbitrage}).
After $n_{\mathrm{tooth}} = f/f_{\mathrm{threshold}}$ steps, the cumulative marginal-price displacement reaches the fee band width and an arb trade fires (the falling edge).
Each tooth accumulates a Hellinger displacement of $n_{\mathrm{tooth}} \cdot \Omega/f = \Omega/f_{\mathrm{threshold}}$, which is independent of $f$.
The KL cost of clearing this displacement in one arb trade is quadratic:
\[
D_{\mathrm{KL}}^{\mathrm{tooth}} \;\approx\; 2\!\left(\frac{\Omega}{f_{\mathrm{threshold}}}\right)^{\!2}.
\]
The number of teeth is $f/n_{\mathrm{tooth}} = f_{\mathrm{threshold}}$, also independent of $f$.
The total rebalancing cost is therefore
\[
C_{\mathrm{rebal}} = f_{\mathrm{threshold}} \cdot \frac{2\Omega^2}{f_{\mathrm{threshold}}^2} = \frac{2\Omega^2}{f_{\mathrm{threshold}}},
\]
which is the zero-fee cost evaluated at $f = f_{\mathrm{threshold}}$.
The fee band sets a minimum resolution for realised weight changes: regardless of how finely the intended interpolation is subdivided, the actual reserve adjustments occur in chunks of displacement $\Omega/f_{\mathrm{threshold}}$, so the total cost is determined by the chunk size and the total distance, not by $f$.

In practice, arber under-extraction, fee revenue on each arb trade, gas thresholds, and the interaction with price-driven arb all modify this picture.
The constant-price calculation shows why $f_{\mathrm{threshold}}$ is the relevant scale; the empirical $\phi$ captures the net effect of the additional mechanisms.

\paragraph{Derivation of the price-driven arb rate $\nu$}

Under driftless GBM, the per-block log-price change $\Delta\log P \sim \mathcal{N}(0,\, \sigma^2\Delta t_{\mathrm{block}})$.
After an arb trade, the market price sits at the fee band boundary.
The expected number of blocks for a symmetric random walk with per-step variance $\sigma^2\Delta t_{\mathrm{block}}$ to exit an interval of half-width $a$ is $a^2/(\sigma^2\Delta t_{\mathrm{block}})$.
Setting $a = (1{-}\gamma)$ (the fee band half-width in log terms):
\[
\mathbb{E}[\text{blocks between price-driven arb}] = \frac{(1{-}\gamma)^2}{\sigma^2\Delta t_{\mathrm{block}}},
\]
so $\nu = \sigma^2\Delta t_{\mathrm{block}}/(1{-}\gamma)^2$ (Eq~\eqref{eq:nu}).

The fee band width $(1{-}\gamma)$ thus sets two characteristic scales: $f_{\mathrm{threshold}} \propto (1{-}\gamma)^{-1}$ for the step count at which weight-driven displacement is absorbed (Eq~\eqref{eq:f_threshold}), and $1/\nu \propto (1{-}\gamma)^2$ for the interval between price-driven arb events.
The ratio $n_{\mathrm{tooth}}/(1/\nu) = f\nu/f_{\mathrm{threshold}}$ determines which type of arb dominates the sawtooth: weight-driven when $n_{\mathrm{tooth}} < 1/\nu$ ($f < f_{\mathrm{threshold}}/\nu$), price-driven otherwise.
For the worked example at $f = f^*$: $n_{\mathrm{tooth}} \approx 5$, $1/\nu \approx 72$, so the sawtooth is weight-driven at the practical operating point.
On high-gas chains, the effective arb threshold is wider than the fee band (gas cost raises the minimum profitable displacement), which further increases the effective $n_{\mathrm{tooth}}$ and reduces the arb cadence.

\paragraph{Typical-case $f_{\mathrm{threshold}}$ with $\sqrt{N}$ correction}

The worst-case bound $f_{\mathrm{threshold}} = 4\Omega/(\sqrt{w_{\min}}(1-\gamma))$ (Eq~\eqref{eq:f_threshold}) uses $|\delta\eta_i| \leq \Omega/f$ for every component.
Along a SLERP path, the per-component displacement is typically $|\delta\eta_i| \sim \Omega/(\sqrt{N}\,f)$, since the total arc-length displacement is distributed across $N$ Hellinger coordinates.
This gives a tighter typical-case threshold:
\begin{equation}
f_{\mathrm{threshold}}^{\mathrm{typical}} = \frac{4\Omega}{\sqrt{N\,w_{\min}}\,(1-\gamma)}
\label{eq:f_threshold_typical}
\end{equation}
For the $N{=}3$ setup ($\Omega = 0.21$, $w_{\min} = 0.05$ along the path, $\gamma = 0.997$):
$f_{\mathrm{threshold}}^{\mathrm{worst}} \approx 1250$ and $f_{\mathrm{threshold}}^{\mathrm{typical}} \approx 720$.
Since $f^* \approx 2400$, the optimum sits well above the typical-case threshold: per-step weight changes are inside the fee band for all component pairs.
For more interior weight configurations ($w_{\min} = 0.2$), $f_{\mathrm{threshold}}^{\mathrm{typical}} \approx 360 \ll f^*$, and fee band absorption is essentially complete.

\paragraph{Arber under-trading and Dutch reverse auction compression}

Empirically, arbitrageurs under-trade relative to the theoretical optimum~\cite{willetts2026poolsportfoliosobservedarbitrage}: they execute 54--71\% of the optimal trade volume, yet capture 81--99\% of the available profit (because the profit function is concave near the optimum, so halving the trade loses only ${\sim}25$\% of the profit).
The resulting extraction efficiency is 61--89\% of the theoretical maximum.

The rebalancing mechanism is a \emph{Dutch reverse auction}: the arb opportunity accumulates over blocks as allocation drift grows, until the expected profit exceeds some arber's gas-plus-risk threshold.
Competition among arbers compresses the extraction over time.
On-chain data from live TFMM pools shows this compression directly: on Ethereum mainnet, the per-trade extraction in one pool fell from \$2.58 to \$0.28 over six months, while trade frequency quadrupled, without any change to pool parameters~\cite{willetts2026poolsportfoliosobservedarbitrage}.
This implies that $\phi$ is not a fixed parameter but decreases with ecosystem maturity.

\paragraph{Incidental routing}

On L2 chains, low gas costs enable a qualitatively different regime.
The swap fee is unchanged (e.g.\ 30\,bp), but the low cost of execution makes high-frequency multi-hop arbitrage viable.
DEX aggregators route through the TFMM pool as one leg of broader cross-venue trades, paying the standard fee while extracting profit from \emph{other} pools.
In one observed TFMM pool on Base (30\,bp fee), 98\% of trades extracted less than \$0.01, with many subsidising the pool; per-trade profit from the TFMM pool was effectively zero, yet the pool still rebalanced~\cite{willetts2026poolsportfoliosobservedarbitrage}.
This is the mechanism underlying $\phi < 0$: the pool receives rebalancing as a side effect of ecosystem-wide routing, and the fee revenue on these trades exceeds the negligible extraction.
The key enabler is gas cost, not fee level.

\paragraph{Summary of mechanisms contributing to $\phi$}

The empirical $\phi$ (defined in \S\ref{sec:fees}) absorbs four independent mechanisms:
\begin{enumerate}
    \item \textbf{Fee band absorption} (\S\ref{sec:fees}, Eq~\eqref{eq:f_threshold}):
    at $f \geq f_{\mathrm{threshold}}$, per-step weight changes do not trigger arb on their own, capping the useful number of interpolation steps.
    \item \textbf{Arber under-trading}:
    when arb does fire, arbers capture only 61--89\% of the theoretical maximum extraction.
    \item \textbf{Fee revenue}:
    each triggered arb trade pays $(1{-}\gamma)$ of the inbound amount to LPs, recapturing a fraction of the extraction.
    \item \textbf{Incidental routing} (low-gas chains):
    DEX aggregator routing provides rebalancing at standard fees while extracting profit elsewhere, potentially pushing $\phi$ negative.
\end{enumerate}
None of these depend on price direction; they are consequences of fee band geometry, gas economics, arber competition, and market microstructure.
The empirical $\phi$ captures their combined effect.
On L1, $\phi \in (0.1, 0.5)$ and the finite-$f$ tradeoff persists.
On L2 with sufficient routing flow ($\phi \leq 0$), these effects eliminate the tradeoff entirely, and the constant-price prescriptions of \S\ref{sec:slerp}--\S\ref{sec:suboptimality} apply without modification.

\end{appendices}

\end{document}